\documentclass[11pt,a4paper]{article}

\usepackage[margin=1in]{geometry}
\usepackage{amsmath}
\usepackage{amssymb}
\usepackage{amsthm}
\usepackage{array}
\usepackage{booktabs}
\usepackage{graphicx}
\graphicspath{{../}}
\usepackage{tikz}
\usepackage[numbers]{natbib}
\usepackage[hidelinks,unicode]{hyperref}
\newcommand{\doi}[1]{\href{https://doi.org/#1}{\nolinkurl{doi:#1}}}

\hypersetup{
  pdftitle={Certified Split Points for Parallel Lexing: Exact and Modulo Discarded Tokens},
  pdfauthor={Nicklas Nidhögg},
  pdfsubject={Parallel lexical analysis; deterministic finite automata; maximal munch tokenization},
  pdfkeywords={parallel lexing, tokenization, deterministic finite automata, maximal munch, split points}
}

\theoremstyle{definition}
\newtheorem{definition}{Definition}

\theoremstyle{plain}
\newtheorem{lemma}{Lemma}
\newtheorem{theorem}{Theorem}
\newtheorem{corollary}{Corollary}
\newtheorem{proposition}{Proposition}

\newcommand{\code}[1]{\texttt{#1}}

\title{Certified Split Points for Parallel Lexing:\\Exact and Modulo Discarded Tokens}
\author{Nicklas Nidh\"ogg\\[0.4em] \normalsize Independent Researcher\\ \normalsize
\texttt{nicklas.nidhogg@gmail.com}\\ \normalsize ORCID: 0009-0006-6161-6150}
\date{August 4, 2026}

\begin{document}

\maketitle

\begin{abstract}
Table-driven DFA lexing is sequential, since each transition depends on the previous byte's state. Scanning one input in
parallel needs each chunk's entry state, which existing methods recover by simulation, speculation, prescanning, or
overlap. We give two conditions under which none is needed. For a longest-match scanner restarting from $q_0$ at every
token boundary, a byte $b$ is a \emph{certified split symbol} when no reachable state other than $q_0$ has a
$b$-transition whose target can reach acceptance, and $q_0$ is not re-entrant if it has one. Every occurrence of such a
byte in completely tokenizable input begins a token, so chunks starting there reproduce the serial sequence of kinds and
lengths. The condition is necessary and sufficient, but one string, comment, or whitespace run can remove every useful
certificate. Comments and whitespace are usually discarded, so we weaken the guarantee to equality after deleting a
declared discarded set and give a second condition, a precomputed constant-time one-bit test. It is
\emph{sound but incomplete} and more permissive, strictly so on suitable token and discarded sets. It recovers newline
for a conventional C-like tokenization and tab, newline and carriage return for JSON, and refuses newline where block
comments are unrestricted. Safety modulo the discarded set is PSPACE-complete for general token automata and decidable
in polynomial time for literal vocabularies, and every unsafe byte has a replayable counterexample. On the fifteen token
sets studied the second condition is already exact. Splitting at exact certificates in munch~\cite{munch} reaches
92.6--95.3\%
parallel efficiency at eight threads on a restricted CPU set and a 512\,MiB corpus, and a $3.46$--$3.94\times$
end-to-end speedup at four threads, on one machine. It turns delimiter-based parallel lexing from a language-specific
assumption into a property a compiler checks.
\end{abstract}

\section{Introduction}

A table-compiled scanner performs, for each byte examined during a token match, a transition of the form \[
\code{state = table[row(byte) + state]}. \] Within one match the transition-table load is loop-carried, its address
depending on the state the previous load returned, so throughput is bounded by the load-to-use latency of the cache
holding the table. Token boundaries reset the state to $q_0$, which can expose independent short matches to out-of-order
execution, but the recurrence still limits long table walks. This work exposes additional independent chains by scanning
several regions of one input concurrently.

Splitting is the obstacle. A scanner dropped at an arbitrary offset does not know the automaton state there. The offset
may fall inside a string literal, halfway through an identifier, or in the middle of a multi-byte operator. Starting
from an incorrect entry state can change the emitted stream until the runs converge, and whether and when they converge
depends on the automaton and the input.

Practice has an answer already: cut at a newline and rescan a little on each side of the cut until the two scans agree.
It bets that a newline begins a token and pays for the bet with the rescan. Where it matters the bet loses, and the
repair has no bound. Inside a block comment or a string that may span lines a newline begins nothing, so a chunk entered
there scans the rest of the comment or string as if it were code. The rescan meant to repair this may have to run back
to the token's opener or forward to its terminator, neither distance bounded by anything in the token set. The
certificate is the exact form of the same idea with the bet removed. Instead of assuming which bytes begin tokens, it
asks the compiled automaton once, at construction, and tests the answer in one bit per byte. Where cutting at newline
was sound all along, the answer is newline; where it was not, the answer is a refusal, with no rescan to repair the cut.
Section~\ref{sec:applicability} shows that the usual C-like token set with block comments is refused, and what it would
cost to change that.

This paper describes a property that removes the uncertainty for a useful class of token sets. The property is a
per-symbol certificate, extracted from the compiled automaton at construction time, such that every occurrence of a
certified byte in every completely tokenizable input begins a token. For a completely tokenizable input, chunk
boundaries placed immediately before those occurrences need no speculation and no reconciliation, because ordered
concatenation of the chunk streams equals the serial stream.

The contributions are:

\begin{itemize}
\item A static condition on a compiled token automaton identifying the bytes at which a longest-match scanner is
provably between tokens, including the re-entrancy requirement without which the natural criterion is unsound
(Section~\ref{sec:certificate}).
\item A proof that the condition is necessary as well as sufficient once the analysis is restricted to states that an
input can reach and from which acceptance is still reachable. Consequently, if no useful byte certifies, no byte
occurring in a completely tokenizable input is universally safe before every occurrence (Section~\ref{sec:necessity}).
\item An $O(|Q|\,|\Sigma|)$ derivation over the logical compiled transition table, a boundary planner, and an exactness
theorem for the resulting parallel scan, implemented and released in a general-purpose lexer library
(Sections~\ref{sec:derivation} and~\ref{sec:evaluation}). The planner and executor decide boundaries with the exact
condition, and the evaluation measures that path.
\item A second condition, which weakens equality to hold only after the tokens in a caller's discarded set $I$ are
deleted and is \emph{sound}, and \emph{conservative} in that it may refuse a byte at which every cut is safe. Every
exact certificate remains certified modulo $I$, and the inclusion is strict for some token sets and discarded sets. The
condition is precomputed from the same tables as a constant-time one-bit test, and it recovers newline for the
conventional C-like and JSON tokenizations that certify no useful byte exactly (Section~\ref{sec:modulo}). A stronger
local test, the rule of release 2.2.0 of munch, is proved sound, strictly stronger and still incomplete. The library
ships that test as a second constant-time one-bit query outside the planner, so a caller using it places the boundaries
itself (Section~\ref{sec:coalesced}). \item An exact decision procedure for safety modulo $I$ at every occurrence, a
problem that is PSPACE-complete for general token automata and decidable in polynomial time for literal vocabularies.
Every unsafe byte comes with an input and a cut that replay the failure. The decision runs a complete-scan automaton
with obligation sets, a failure automaton beside it, and a finite test of output equality by group delays
(Section~\ref{sec:exact}). Release 3.1.0 of munch~\cite{munch310} implements it as \code{modulo\_counterexample()}.
\item A controlled applicability study of fifteen token sets under both conditions, showing where each condition yields
useful bytes, how one added token kind can eliminate every useful exact certificate, and when line-based splitting is
sound for the C-like tokenizations studied. The exact decision confirms every cell of the relaxed column
(Section~\ref{sec:applicability}).
\end{itemize}

\begin{table}[!ht]
\centering
\begin{tabular}{>{\raggedright\arraybackslash}p{6.2cm}>{\raggedright\arraybackslash}p{5.1cm}
    >{\raggedright\arraybackslash}p{3.35cm}}
\toprule
Question & Answer & Where \\
\midrule
When may a chunk boundary be taken without knowing what precedes it? & No live state other than $q_0$ has a live
    transition on the byte, and $q_0$ is not re-entrant if it has one & Theorem~\ref{thm:main} \\
\addlinespace
Is that condition merely sufficient? & No, it is necessary as well & Theorem~\ref{thm:necessity} \\
\addlinespace
What does it cost to decide? & Time linear in the compiled tables, then one bit per byte & Section~\ref{sec:derivation}
    \\
\addlinespace
Do real token sets satisfy it? & 15 token sets measured, and the conventional C-like and JSON rows certify no useful
    byte exactly & Table~\ref{tab:applicability} \\
\addlinespace
Can a token set that certifies no useful byte be rescued? & Weakening to equality modulo discarded tokens recovers
    newline there, and eight of the fifteen rows gain bytes & Section~\ref{sec:modulo}, Table~\ref{tab:applicability} \\
\addlinespace
Is a refusal of the relaxed condition final? & No. Safety modulo $I$ is decidable, PSPACE-complete for general token
    automata, and every unsafe byte has a replayable counterexample. On the fifteen rows every refusal stood &
    Theorems~\ref{thm:exact} and~\ref{thm:pspace}, Table~\ref{tab:applicability} \\
\addlinespace
Can it be bought by design? & In the studied grammar the block comment is the one token spanning lines. Once newline is
    its own token, a line-bounded block comment suffices, while barring the kinds from each other's openers buys nothing
    & Table~\ref{tab:applicability} \\
\addlinespace
Does it pay? & 92.6--95.3\% parallel efficiency at eight threads, $3.46$--$3.94\times$ end to end at four &
    Section~\ref{sec:evaluation} \\
\bottomrule
\end{tabular}
\caption{What this paper establishes. Every row is stated and proved or measured at the location named. The
applicability results are asserted cell by cell by \code{figures/applicability.cpp}, and the throughput row rests on one
machine and two benchmark revisions, as Section~\ref{sec:evaluation} details.}
\label{tab:results}
\end{table}

Table~\ref{tab:results} gathers the claim in one place before the definitions arrive. The paper separates the result
from the study around it. The result is Sections~\ref{sec:prelim}--\ref{sec:modulo}: the definitions, the exact
certificate with its necessity proof and the re-entrancy subtlety, the relaxed condition modulo discarded tokens, the
stronger coalesced test of Section~\ref{sec:coalesced}, and the exact decision of Section~\ref{sec:exact}. The study and
the artifact are Sections~\ref{sec:derivation}--\ref{sec:evaluation}: the derivation and planner, the applicability of
both conditions over fifteen token sets, and the measured parallel scan. Section~\ref{sec:prior} places the certificate
among the existing answers to the entry-state problem before the result, and Section~\ref{sec:related} returns to that
literature after it. A companion paper generalizes the certificate from single bytes to short byte
windows~\cite{nidhogg2026splitwindows}, and a second reuses the same certificates to choose where a scan resumes after
an error~\cite{nidhogg2026panicmode}. Nothing here depends on either.

\section{Prior approaches}
\label{sec:prior}

Published solutions mostly accept the unknown-state problem and manage it. The few that avoid it assume or prove a
delimiter for one format, as the paragraph on delimiters below records.

\paragraph{Compile the simulation into the automaton.} In simultaneous finite automata~\cite{sinya2013sfa} each state of
the extended automaton is a mapping from entry state to exit state, so one ordinary pass over a chunk yields that
chunk's whole transfer function. The mappings compose associatively, and chunks combine in a parallel reduction. The
per-state simulation moves from scan time to construction, and the authors report almost no runtime overhead. The cost
is the size of the constructed automaton, whose states are maps on states, so the worst case is $n^n$ from a DFA with
$n$ states and $2^{n^2}$ from an NFA, though for the expressions they survey it is usually far smaller. Of the more than
20,000 SNORT expressions they measure, 98.6\% give a D-SFA no larger than the square of the minimal DFA, 279 exceed it
and six exceed its cube. The reduction step also remains.

Composition can also be applied directly instead of compiled in, and in that form it is a classical answer. Hillis and
Steele~\cite{hillis1986dataparallel} treat each character as a unary function on states, observe that the induced
composition is associative, and recover the state after every character with one parallel-prefix operation. Their worked
example is lexing program text, and Yang~\cite{yang1996mealy} shows that the example needs care. The scan requires an
automaton that can tell a token has begun as soon as its first character is read; that holds for one-character lookahead
but not in general, so an automaton needing more must be transformed before the scan applies at all. A Mealy machine
attaches token output to transitions instead of states, and that makes the lookahead behavior representable. The object
composed for maximal-munch lexing is richer than a map from states to states. Data-parallel finite-state
machines~\cite{mytkowicz2014dpfsm} take the same route on SIMD and multicore hardware, enumerating transitions from
every possible start state so that the enumeration is exactly the transition function, and citing Hillis and Steele as
the basis for doing so. The resulting transfer function is exact, and once the chunk summaries determine the correct
entry states the chunks are run again to produce output, so nothing is guessed and nothing needs repair. The cost is a
factor of $|Q|$ in work, which convergence reduces in practice, though the authors report that convergence to a single
state is rare. A GPU lexer applies the same scan to recover the complete state stream, the state after every input
position~\cite{voetter2021gpu}. Holding one function table per input position costs $O(|Q|n)$ space, so that work
precomputes the reachable compositions and identifies each by an integer, reducing composition to a two-dimensional
lookup. The trade is the one the simultaneous construction makes in states, here paid in a table.

\paragraph{Carry fewer entry states.} Composition is exact but pays for every state. A second family attacks that cost
by shrinking the set a chunk must carry. One line of work uses single-state speculation. Jones et
al.~\cite{jones2009browser} observe that in lexing the automaton reaches a stable state within a few characters, so
prepending a short suffix of its left neighbor to each chunk usually recovers the entry state. Prabhu et
al.~\cite{prabhu2010speculative} make the pattern explicit as language constructs for prediction, validation on
completion and re-execution on a miss, with lexical analysis as the motivating workload and the entry state predicted by
scanning a few characters before the cut. Luchaup et al.~\cite{luchaup2009speculation} place the same wager for
intrusion-detection signature matching, where a secondary scan enters its chunk in the DFA's start state and records its
state after every character, and the primary validates by running on into the chunk and comparing against that history
until the two couple, so even a wrong guess usually costs only the short prefix before convergence. Enumerative
speculation sits between the extremes of guessing one state and enumerating all of them, speculating transitions from
several states chosen by a lookback over the preceding input~\cite{jiang2017enumerative}. Reduced-interface
DFAs~\cite{borsotti2025ridfa} attack the same overhead from the automaton side, cutting the number of starting states a
chunk automaton must carry by combining an NFA's state reduction with deterministic transitions. The remaining costs
differ. A lookback speculation can miss, and a miss costs a rescan, whereas a reduced-interface DFA carries its
interface in full and pays instead in proportion to that interface's size. Both approaches may still carry more than one
candidate entry state, and that extra work is one limit on the speedup.

\paragraph{Relocate cuts to a language's separators.} Parallel lexers have long been built by cutting near equal
divisions and sliding each cut to a language-specific separator~\cite{barenghi2015parallel}. The separator and the
search bound are chosen by hand, and finding a separator does not by itself remove the entry-state problem. The
accompanying parallel lexical analysis still enumerates the possible start states of a chunk. Each worker carries one
computation per alternative, at most three in their worked cases and four in the stated worst case, when the bounded
separator search finds none~\cite{barenghi2015parallel}. The technique is therefore separator relocation combined with
state enumeration, not a proof that the separator is safe. Its worked cases are the direct antecedent of
Section~\ref{sec:applicability}. Newline preserves the JSON token stream that survives discarding, since no surviving
lexeme admits it, yet it is rejected in practice because generated JSON may contain no newline~\cite{li2021plex}. Lua
fails for a stronger reason than any certificate could address: its long-bracket delimiters \code{[=$^n$[} and
\code{]=$^n$]} must agree on $n$, so the lexical grammar is not regular, the set of possible delimiters is unbounded,
and with it the set of possible entry states. No fixed lookahead can decide which delimiter, if any, encloses a
chunk~\cite{barenghi2015parallel}. Barenghi et al. recover a workable schema only by constraining the language,
admitting just \code{[[} and \code{]]} for strings and requiring the multiline comment's closing delimiter to occur at
line end, after which newline does serve as their split point. The missing piece is a way to decide, from the token set
alone, whether every relevant occurrence of a proposed byte is safe to cut immediately before.

\paragraph{Prescan for context.} Plex~\cite{li2021plex} removes the need for delimiters entirely. From the scanner's own
DFA it derives a backtrack-free prescanning automaton, runs a prescan to recover each chunk's context, and then scans
the chunks in parallel without language-specific delimiter analysis. The price is a pass over the input, and the benefit
is that it applies to grammars that certify nothing here.

\paragraph{Analyze the grammar for streaming.} Deciding something about a maximal-munch token set statically, before any
input, is an established move. Yang, Tsay, and Chan~\cite{yang2002longest} decide automatically whether the
longest-match rule is applicable at all for a given token set and parser grammar, and identify precisely the situations
in which it is not; the question is different from ours, but the setting is the same. StreamTok~\cite{li2026streaming}
is methodologically nearer still, since it too analyzes such a grammar statically and partitions grammars into those its
technique serves and those it does not. Its maximum token neighbor distance bounds how far a longest-match decision can
depend on future input, and that bound makes bounded-memory streaming possible. StreamTok derives no input-independent
split symbols or certified chunk boundaries, and parallelizing it is left as future work there. The two analyses are
complementary, since StreamTok bounds the context needed to confirm a token while the certificate identifies positions
where no preceding context is needed at all.

\paragraph{Recover a restart point after an edit.} Incremental lexers face the mirror of this
question~\cite{wagner1997lexing}: after an edit, how far back must re-lexing begin for the result to equal a full
re-scan? They answer it dynamically and per edit, saving the batch machine's state with each token as it is created so
analysis can restart at any token boundary, and tracking lookahead dependencies as they arise, with no bound set in
advance. The certificate answers a static question instead, once per token set and before any input exists, and the two
are complementary because a restart point is a fact about one input and its history while a certified byte is a fact
about the grammar. An incremental divide-and-conquer lexer~\cite{hugo2015incremental} takes the other route open to it,
storing a result for every possible entry state of a fragment and composing the resulting transition maps, which places
it in the all-state family.

\paragraph{Restrict the automaton.} Holub and \v{S}tekr's parallel run~\cite{holub2009parallel} is exact and efficient
for $k$-local automata, in which any $k$ consecutive symbols force a unique state regardless of the starting state. The
property is uniform over the whole automaton, and a lexical grammar need not have it.

\paragraph{Realign after the fact.} Parallel tokenization for language-model vocabularies faces the same boundary
problem. The overlap-based answer to it extends chunks so neighbors share a region and merges inside it, but it does not
guarantee the sequential result. LoPT~\cite{shao2026lopt} remains overlap-based, matches overlap tokens by character
position, and retries with a doubled chunk whenever no overlap token matches. For WordPiece and BPE, its Theorem~3.1 is
conditional on a sufficient overlap condition: equality with sequential tokenization holds when every position-aligned
overlap spans more characters than the longest vocabulary token. TokTier~\cite{zhang2026toktier} reports that LoPT's
published experimental configuration ``does not satisfy its safety theorem's stated precondition'' and that TokTier's
clean-room reproduction ``found inputs on which the length threshold stated in its paper admits a boundary outside the
premise of its theorem''. We cite both findings as theirs. Not all of that work is overlap-based. Gigatoken certifies
cuts before selected punctuation in SentencePiece input where no vocabulary piece contains the two symbols at the cut in
that order. It argues in a source comment that a merge's result is always a vocabulary piece, so that no merge can span
such a boundary~\cite{rod2026gigatoken}. The runtime of SentencePiece Lite states the same principle for its
vocabularies~\cite{sentencepiecelite2026}, and Xu checks a merge table for cuts across which no merge
joins~\cite{xu2026parallel}. Each is a sufficient condition read off one tokenizer family's vocabulary before any input,
the subword counterpart of the certified split windows~\cite{nidhogg2026splitwindows} rather than of a certified byte.

\paragraph{Assume a delimiter.} Data systems split logs and CSV at newlines because ``records do not contain newlines'',
adjusting each cut to the next delimiter. The assumption is per-format and usually informal, and it is famously unsound
for CSV, where RFC 4180 explicitly permits CRLF inside double-quoted fields~\cite{shafranovich2005csv}, which is why
speculative CSV parsing is a research topic~\cite{ge2019csv}. Where it has been replaced by a proof, the proof is about
one format. Ge et al.\ show that a quote followed by an ordinary character, and an ordinary character followed by a
quote, each rule out some starting states of their CSV automaton and leave every surviving run in a single state, so a
chunk containing either pattern is unambiguous~\cite[Section~4.4 and Appendix~A, Property~2]{ge2019csv}. Massively
parallel delimiter parsing~\cite{stehle2020parparaw} attacks the same context problem on GPUs. Parallel lexical analysis
has also been built by choosing cut positions directly, and Barve and Joshi~\cite{barve2014parallel} take newline,
whitespace and selected language constructs as pivot elements. Those pivots are chosen for the language by hand, and
deriving them from its token set is the step the certificate supplies. Format-specific structural scans hand-derive
comparable facts per format: Mison~\cite{li2017mison} and simdjson~\cite{langdale2019simdjson} locate JSON's structural
characters and mask away the occurrences inside strings, while Parabix~\cite{cameron2008parabix} builds XML lexical and
validation bit streams, a different construction aimed at a different format.

Among the preceding methods surveyed here, none derives from a compiled token automaton the complete set of raw-input
bytes safe to cut immediately before at every occurrence in a completely tokenizable input. Lexer generators come
closest, and re2c analyzes the compiled automaton to reject rules in which one user-declared end-of-input sentinel may
occur before a lexeme ends~\cite{re2cdocs}. The check is a terminal-in-lexeme test for a declared byte, a safe-after
property under the model used here, and derives no safe-before set. Section~\ref{sec:related} shows that the two
properties are incomparable.

\section{Preliminaries}
\label{sec:prelim}

Let $\Sigma$ be the byte alphabet. We take a compiled token set to be a deterministic automaton
$A = (Q, \Sigma, \delta, q_0, \tau)$ with a partial transition function $\delta : Q \times \Sigma \rightharpoonup Q$, an
initial state $q_0 \in Q$, and a partial accepting map $\tau : Q \rightharpoonup T$ assigning a token to accepting
states. Both $Q$ and $T$ are finite, $\Sigma = \{0, \ldots, 255\}$, and the single-valued $\tau$ already contains the
winner of any same-length priority tie resolved during compilation, lower priority number first and lower token ID
second. We write $\delta^*(q, u)$ for the state reached from $q$ on $u \in \Sigma^*$ when every transition along the way
is defined, and say $\delta^*(q,u)$ is undefined otherwise.

The scanner is the usual longest-match loop. At offset $i$ of an input $w \in \Sigma^*$ it starts in $q_0$, consumes
bytes while transitions are defined, remembers the last position at which $\tau$ was defined, emits the token found
there, and restarts at that position in $q_0$. If no accepting position is reached, the scan halts at $i$. Tokens are
positive-length: an accepting position is eligible only after at least one symbol has been consumed, so an accepting
$q_0$ never emits the empty word. If no positive-length accepting prefix exists the scan halts; emitting a length-zero
match there would restart at the same offset forever, and Section~\ref{sec:subtlety} needs that ruled out. The
whole-input entry point and each chunk-local scan used by the parallel entry point treat a final length-zero match as
failure and halt at that offset, and one failed chunk stops no other. The single-match entry point reports the
length-zero match instead, and a caller looping over it must reject that match itself. The composite system is not a
single automaton run but a run that \emph{resets to $q_0$ at every token boundary}, and the certificate exploits that
reset.

Write $\mathrm{tok}(w)$ for the sequence of (token, length) pairs the scanner emits on $w$, and $\mathrm{con}(w)$ for
the scanner's final committed offset, equivalently the sum of the lengths in $\mathrm{tok}(w)$. Failed lookahead does
not itself advance the committed offset. Call $w$ \emph{completely tokenizable} when $\mathrm{con}(w) = |w|$, and call
an offset a \emph{token boundary} of $w$ when it is $0$ or a cumulative sum of the emitted lengths. Then
$\mathrm{con}(w)$ is always a boundary, and $|w|$ is one exactly when $w$ is completely tokenizable.

\section{The certificate}
\label{sec:certificate}

A state can lie on an emitted token only if some input reaches it and some continuation from it still accepts, and the
certificate is defined over those alone. Write \[ Q^+ = \{\, q \in Q \mid \delta^*(q_0, u) = q \text{ for some } u \in
\Sigma^*, \text{ and } \delta^*(q, v) \text{ is accepting for some } v \in \Sigma^* \,\}
\] for the states that are both reachable and co-accessible, and let $\delta^+(q, a) = \delta(q, a)$ when that
transition is defined and $\delta(q, a) \in Q^+$, and be undefined otherwise. All that follows is stated over the
\emph{live subautomaton} $A^+ = (Q^+, \Sigma, \delta^+, q_0, \tau)$. Section~\ref{sec:necessity} shows why neither half
of the restriction is cosmetic. A pattern denoting the empty language leaves reachable states outside $Q^+$ whose
transitions would otherwise de-certify symbols that no token can contain. A state no input enters lies on no
counterexample, so admitting one would let a transition that no scan ever takes de-certify a symbol that every scan
treats as safe. We assume $q_0 \in Q^+$ throughout, since otherwise the token set accepts nothing and no symbol is
useful.

One consequence is used repeatedly below: the match path of any emitted token lies in $A^+$ throughout, since every
transition along it leads to that token's accepting position, so reasoning about emitted tokens may read $\delta$ as
$\delta^+$ without further comment. A second consequence is of the same kind: every $q \in Q^+$ is reachable by a path
lying in $A^+$ throughout, since an intermediate state of a $\delta$-path from $q_0$ to $q$ is reachable by the path's
prefix and co-accessible by appending the path's remainder and then an accepting continuation from $q$. Wherever
reachability of a state in $Q^+$ is invoked below, $(\delta^+)^*(q_0, u) = q$ may be read directly.

\begin{definition}[Re-entrant initial state]
\label{def:reentrant}
$q_0$ is \emph{re-entrant} if $(\delta^+)^*(q_0, u) = q_0$ for some $u \in \Sigma^+$.
\end{definition}

Because every state of $Q^+$ is reachable, Definition~\ref{def:reentrant} is equivalent to $q_0$ having an incoming live
transition from a state of $Q^+$, which is how an implementation tests it.

\begin{definition}[Certified split symbol]
\label{def:certified}
A symbol $b \in \Sigma$ is \emph{certified} if every $q \in Q^+$ with $\delta^+(q, b)$ defined satisfies $q = q_0$, and
if $\delta^+(q_0, b)$ is defined then $q_0$ is not re-entrant.
\end{definition}

The exemption for $q_0$ makes useful certificates possible, since a token may legitimately begin with $b$. The exemption
is sound only while no input can return the automaton to $q_0$ mid-scan, which is precisely the re-entrancy condition.
Section~\ref{sec:subtlety} shows that dropping it makes the certificate false.

The definition has an instance every reader has met. Take the token set whose tokens are the encoding forms of
UTF-8~\cite{yergeau2003utf8}, one token per form. After a byte that begins a form only continuation bytes,
\code{0x80}--\code{0xBF}, are admitted until the form completes, and no form is a prefix of another, so no live state
other than $q_0$ has a transition on a byte that begins a form and no transition returns to $q_0$. Every such byte,
\code{0x00}--\code{0x7F}, \code{0xC2}--\code{0xDF}, \code{0xE0}--\code{0xEF} and \code{0xF0}--\code{0xF4}, is certified,
and no continuation byte is, since only states other than $q_0$ consume one. The bytes no form contains, \code{0xC0},
\code{0xC1} and \code{0xF5}--\code{0xFF}, are certified vacuously, and the shipped predicate withholds them as
Corollary~\ref{cor:useful} says. That is the property the RFC lists among the encoding's characteristics, ``character
boundaries are easily found from anywhere in an octet stream'', and a definition that knows nothing of UTF-8 recovers it
from the compiled table. \code{figures/applicability.cpp} asserts it byte for byte. The question this paper asks is
which token sets have that property, and Section~\ref{sec:applicability} answers for fifteen of them.

\begin{lemma}[Occurrences begin tokens]
\label{lem:boundary}
Let $b$ be certified and let $w$ be completely tokenizable. Then every offset of $w$ holding $b$ is a token boundary of
$w$.
\end{lemma}

\begin{proof}
Let $w_i = b$ and suppose $i$ is not a token boundary. Since $w$ is completely tokenizable, every offset is either a
token boundary or lies strictly inside some emitted token, so $i$ lies strictly inside a token that starts at some
$s < i$. Let $u = w_s \cdots w_{i-1}$, which is nonempty, and let $q = (\delta^+)^*(q_0, u)$, the state immediately
before $b$ is consumed. The token's match path continues through offset $i$, so $\delta^+(q, b)$ is defined. By
Definition~\ref{def:certified}, $q = q_0$, hence $(\delta^+)^*(q_0, u) = q_0$ with $u$ nonempty, so $q_0$ is re-entrant
by Definition~\ref{def:reentrant}. But $\delta^+(q_0, b) = \delta^+(q, b)$ is defined, so Definition~\ref{def:certified}
also requires $q_0$ not to be re-entrant, a contradiction.
\end{proof}

\begin{lemma}[Prefix stability]
\label{lem:truncation}
Let $w$ be completely tokenizable, let $s$ be a token boundary of $w$, and let $v$ be any prefix of $w$ with
$|v| \geq s$. Then $v$ and $w$ emit the same tokens on $[0, s)$, and the scan of $v$ reaches $s$ in state $q_0$.
\end{lemma}

\begin{proof}
If $s = 0$ both claims are immediate, since $[0, s)$ is empty and the scan of $v$ starts at offset $0$ in $q_0$. Assume
$s > 0$. First, no token of $w$ beginning at some $s' < s$ records an accepting position after $s$, since longest match
takes the last accepting position reached, so such a token would span past $s$, placing $s$ strictly inside it and
contradicting that $s$ is a boundary.

The rest is an induction over the tokens preceding $s$. Suppose both scans begin a token at the same offset $s' < s$ in
$q_0$, which holds at offset $0$. From $s'$ the two scans read the same bytes, since $v$ agrees with $w$ byte for byte
as far as it goes. They traverse the same states while both inputs last, and the accepting positions they record at
offsets up to $|v|$ coincide. By the bound above, every accepting position this token records in $w$ lies at or before
$s \leq |v|$, so both scans record the same set of accepting positions, take the same last member of it, and emit the
same kind and length. Truncation removed only lookahead that failed. Both then restart in $q_0$ at the same next offset,
which is at most $s$. The restart offsets increase strictly, and $s$ is a boundary of $w$, so the induction reaches
offset $s$ exactly. Every token before $s$ is emitted identically in both scans, and the scan of $v$ reaches $s$ in
$q_0$.
\end{proof}

\begin{theorem}[Split invariance]
\label{thm:main}
Let $w$ be completely tokenizable and let $0 = c_0 < c_1 < \cdots < c_m = |w|$ be offsets such that each interior $c_j$
holds a certified symbol. Then
\[
\mathrm{tok}(w) \;=\; \mathrm{tok}(w_{[c_0,c_1)}) \cdot \mathrm{tok}(w_{[c_1,c_2)}) \cdots
\mathrm{tok}(w_{[c_{m-1},c_m)}),
\]
where $\cdot$ denotes concatenation of token sequences and $w_{[a,b)}$ the corresponding slice.
\end{theorem}

\begin{proof} The offsets are strictly increasing, so every chunk has positive length, which the argument uses where it
takes a nonempty prefix before an interior boundary. A driver must supply distinct offsets for the theorem to apply at
all. For empty $w$ the partition degenerates to the single offset $0$, both sides are the empty sequence, and the claim
is trivial; assume $|w| > 0$. By Lemma~\ref{lem:boundary} every interior $c_j$ is a token boundary of $w$, and $c_0$ and
$c_m$ are boundaries trivially. Two facts are then needed. First, the scanner is deterministic and enters each token in
state $q_0$ with no state carried across boundaries, so a scan started at $c_j$ begins as the global scan does there.
Second, while matching the last token of the chunk the global scan records no accepting position beyond $c_{j+1}$. Let
$q$ be the state reached on the nonempty prefix ending just before $c_{j+1}$. If $\delta^+(q, b)$ were defined for the
certified $b$ at $c_{j+1}$, the argument of Lemma~\ref{lem:boundary} would force $q = q_0$ and make $q_0$ re-entrant,
which Definition~\ref{def:certified} forbids, so $\delta^+(q, b)$ is not defined. Either $\delta(q, b)$ is undefined,
and the scan stops at $c_{j+1}$, or its target lies outside $Q^+$, since otherwise $q$ would itself be co-accessible and
$\delta^+(q, b)$ would be defined. In the second case the scan may read on but cannot reach an accepting state again.
Note that the scan is not required to \emph{stop} at $c_{j+1}$, only to find nothing past it. Lookahead may cross the
boundary through transitions that lead nowhere. Either way the last accepting position it recorded lies at or before
$c_{j+1}$, and the chunk scan, seeing the same bytes from the same state, records the same one, so the tokens emitted
between $c_j$ and $c_{j+1}$ coincide. The second fact is needed only for interior boundaries. When $j + 1 = m$ the right
edge is the end of the input. It carries no symbol and the global scan has no bytes left there either, so the final
chunk sees the suffix the global scan sees, from the same state, and the two agree without further argument. Induction
over $j$ gives the claim.
\end{proof}

\begin{lemma}[Boundaries before the first failure]
\label{lem:prefix}
Let $b$ be certified and let $w$ be any input. Every offset $i < \mathrm{con}(w)$ with $w_i = b$ is a token boundary of
the serial scan of $w$.
\end{lemma}

\begin{proof}
Every offset below $\mathrm{con}(w)$ is either a token boundary or lies strictly inside a token the scan emitted, so if
$i$ is not a boundary it lies strictly inside a token starting at some $s < i$. The argument of Lemma~\ref{lem:boundary}
applies unchanged to that token. The state $q$ reached on the nonempty $w_s \cdots w_{i-1}$ has $\delta^+(q,b)$ defined,
hence $q = q_0$ by Definition~\ref{def:certified}, hence $q_0$ is re-entrant, which Definition~\ref{def:certified}
forbids once $\delta^+(q_0,b)$ is defined. Completeness of the tokenization is not used, only that the offsets below
$\mathrm{con}(w)$ were successfully traversed.
\end{proof}

\begin{corollary}[Malformed input] \label{cor:malformed} Let $w$ be any input with $\mathrm{con}(w) < |w|$, and let
$0 = c_0 < c_1 < \cdots < c_m = |w|$ be offsets such that each interior $c_j$ holds a certified symbol. Let
$[c_j, c_{j+1})$ be the chunk containing the offset $\mathrm{con}(w)$ of the first failure, which is the chunk's own
first byte when $\mathrm{con}(w) = c_j$. By Lemma~\ref{lem:prefix} every interior boundary below $\mathrm{con}(w)$ is a
token boundary of the serial scan, and $\mathrm{con}(w)$ is one always, so the reasoning of Theorem~\ref{thm:main}
applies to every chunk that ends at or before $\mathrm{con}(w)$. Every chunk before $c_j$ is fully consumed and
contributes exactly the tokens the serial scan emits there, and the chunk containing the failure halts at the same
offset. The serial stream is therefore a prefix of the concatenation, but chunks after the failure may emit tokens the
serial scan never reaches. A caller must check every chunk's consumed length before treating the concatenation as a
successful tokenization.
\end{corollary}

\begin{proposition}[Vacuous certificates]
\label{prop:vacuous}
If no $q \in Q^+$ has $\delta^+(q,b)$ defined, then $b$ is certified, and no completely tokenizable input contains $b$.
\end{proposition}

\begin{proof}
The first claim holds because Definition~\ref{def:certified} quantifies over an empty set. For the second, a completely
tokenizable $w$ containing $b$ would have $b$ consumed by some emitted token, whose match path lies in $A^+$, so some
state of $Q^+$ would consume $b$ live.
\end{proof}

Such symbols are certified but useless for planning, since they cannot appear in valid input. Vacuous certificates can
also be identified mechanically, with no further analysis. If $b$ is certified then every live state consuming it is
$q_0$, so $b$ is consumed live by some state exactly when $\delta^+(q_0, b)$ is defined:

\begin{corollary}[Useful certificates]
\label{cor:useful}
A certified $b$ is non-vacuous if and only if $\delta^+(q_0, b)$ is defined.
\end{corollary}

An implementation can therefore obtain the useful set with no extra analysis, by intersecting the certificate with the
initial state's live transitions, and the predicate in the accompanying library reports that set. Useful certificates
are the \emph{occurring} ones in the terminology of the companion paper on certified
splitting~\cite{nidhogg2026certifiedsplitting}, called the splitting companion below. A useful $b$ occurs in a
completely tokenizable input $bv$ that scans as one token, by the argument in the proof of Theorem~\ref{thm:necessity},
and by Proposition~\ref{prop:vacuous} no vacuous one occurs in any. Vacuous certificates never reach a caller, since a
caller cannot act on one and a planner searching for one scans the whole input for nothing.

\subsection{The condition is also necessary}
\label{sec:necessity}

Definition~\ref{def:certified} is stated as a test, and Lemma~\ref{lem:boundary} shows it is sufficient. The test is
also necessary, so nothing is given away by using it. Necessity needs only the standing assumption that the token set
accepts some word ($q_0 \in Q^+$), since the states of $Q^+$ are reachable and co-accessible by construction, as the
counterexample below requires. The theorem therefore applies to every such token set whatever automaton it was compiled
from, and a token set accepting nothing certifies no useful symbol in the first place.

\begin{theorem}[Characterization]
\label{thm:necessity}
Assume the standing $q_0 \in Q^+$. Then $b$ is certified if and only if every occurrence of $b$, in every completely
tokenizable input, begins a token.
\end{theorem}

\begin{proof}
Sufficiency is Lemma~\ref{lem:boundary}. For necessity we show that a rejected $b$ admits a counterexample: a completely
tokenizable input in which an occurrence of $b$ is not a token boundary. Definition~\ref{def:certified} rejects $b$ in
exactly two ways.

Suppose some $q \neq q_0$ has $\delta^+(q,b)$ defined. Then $q \in Q^+$ is reachable, so $(\delta^+)^*(q_0,u) = q$ for
some $u$, and $u$ is nonempty because $q \neq q_0$. The target $\delta^+(q,b)$ lies in $Q^+$ by the definition of
$\delta^+$, so $(\delta^+)^*(\delta^+(q,b), v)$ is accepting for some $v$. Take $w = ubv$. Reading $w$ from $q_0$
traverses $u$ to $q$, then $b$, then $v$ to an accepting state, so the scan from offset $0$ reaches an accepting
position at $|w|$, and no later one exists because the input ends there. Longest match emits a single token spanning
$w$, so $w$ is completely tokenizable, and the occurrence of $b$ at offset $|u| \geq 1$ begins no token.

Otherwise $\delta^+(q_0,b)$ is defined and $q_0$ is re-entrant. Re-entrancy gives a nonempty $u$ with
$(\delta^+)^*(q_0,u) = q_0$, and co-accessibility of $\delta^+(q_0,b)$ gives $v$ as before. The same $w = ubv$ is
completely tokenizable and places $b$ at offset $|u| \geq 1$ inside its only token.
\end{proof}

The proof depends on restricting to $A^+$, and the restriction is not cosmetic. A pattern denoting the empty language
leaves reachable states behind from which no accepting state can be reached. Taking $T_1 = \code{ab}\cdot\emptyset$
alongside $T_2 = \code{b}$, the compiled automaton has $q_0 \xrightarrow{a} q_2 \xrightarrow{b} q_1$ beside
$q_0 \xrightarrow{b} q_3$ accepting, and $q_1, q_2$ are reachable but not co-accessible. Had
Definition~\ref{def:certified} been stated over $\delta$ rather than $\delta^+$, it would reject $b$, because the
noninitial $q_2$ consumes it; yet every completely tokenizable input here is a sequence of \code{b} tokens in which
every $b$ begins one, so no counterexample exists and necessity would fail. Over $\delta^+$ the offending transition is
invisible, since $q_1 \notin Q^+$, and $b$ certifies as it should. Minimization does not remove such chains, because the
minimizer used here treats a missing transition as distinct from one into a state that cannot accept. A scanner needs
that distinction, since the two differ in how far a longest match reads before failing. This minimization is
deliberately weaker than Myhill-Nerode minimality, under which every state with an empty right language is equivalent.

Reachability is the mirror image, and it matters for automata not produced by subset construction. Give the same
$T_2 = \code{b}$ a detached $q_4 \xrightarrow{b} q_5$ with $q_5$ accepting. Now $q_4$ is co-accessible and consumes $b$,
so a definition quantifying over co-accessible states alone would reject $b$; yet no input reaches $q_4$, every
completely tokenizable input is again a sequence of \code{b} tokens, and no counterexample exists. Requiring
reachability and co-accessibility in $Q^+$ keeps the definition exactly as permissive as the scans it describes.

The derivation computes $Q^+$ first, in two sweeps. One forward sweep from $q_0$ marks the reachable states, and one
reverse sweep marks those from which an accepting state remains reachable. The certificate sweep then skips every
unmarked source and every transition whose target is unmarked, since neither can lie on an emitted token and so neither
carries information about where tokens may begin. Because the derivation computes both marks,
Theorem~\ref{thm:necessity} applies to a hand-built automaton as much as to a compiled one. Each sweep visits each table
entry once and leaves the $O(|Q|\,|\Sigma|)$ bound of Section~\ref{sec:derivation} unchanged.

Corollary~\ref{cor:useful} depends on the same restriction. In the example $\delta(q_0, a)$ is defined but
$\delta^+(q_0, a)$ is not, since the target $q_2$ is not co-accessible. A test phrased over $\delta$ would classify $a$
as a useful certificate, even though no input this token set accepts contains an $a$ at all. Phrased over $\delta^+$,
the test reports $a$ as vacuous, which the caller needs to know.

Consequently, for a token set compiled this way, ``no useful byte certifies'' is not a failure of the test but means
that no single byte can serve as a split symbol at all under the semantics of Section~\ref{sec:prelim}. The
applicability results of Section~\ref{sec:applicability} inherit that strength, and the negative rows there are
statements about the tokenizations themselves.

\subsection{The subtlety}
\label{sec:subtlety}

Dropping the re-entrancy condition from Definition~\ref{def:certified} makes the result false. Consider the token set
consisting of the single nullable pattern $a^*$. Minimization yields an automaton whose initial state is accepting and
carries a self-loop on $a$; the only state consuming $a$ is $q_0$, so the naive certificate admits $a$. Splitting the
input $aa$ between its two bytes then yields two tokens where the serial scan, by longest match, yields one. The
condition is therefore required for soundness.

Figure~\ref{fig:reentrancy} shows why the condition is stated as an incoming transition rather than as a self-loop. In
the nullable case the initial state re-enters itself directly, which a self-loop test would catch. In the cyclic variant
$(ab)^*c$ it re-enters through the cycle $a, b$, and no self-loop exists anywhere; the only state consuming $c$ is
$q_0$, so a self-loop test would certify $c$ and split $abc$ in the middle of its only token. Both cases are carried as
regression tests.

\section{Certification modulo discarded tokens}
\label{sec:modulo}

Fix a set $I \subseteq T$ of \emph{discarded} tokens, chosen by the caller. Write $\pi_I$ for the map on token sequences
deleting every pair whose token lies in $I$, and call two sequences \emph{$I$-equivalent} when $\pi_I$ sends them to the
same sequence. Comparing streams modulo a designated token class is not new, since ZipLex's separator construction
proves its printing guarantee in this form, with re-lexing equal up to inserted separator
tokens~\cite{chassot2026ziplex}. The question here is different, namely which symbols of raw input allow a cut that
preserves $I$-equivalence. Such a cut can land inside an existing discarded token, away from any seam between formed
tokens.

\begin{figure}[!ht]
\centering
\begin{minipage}[b]{0.3\linewidth}
\centering
\includegraphics[width=\linewidth]{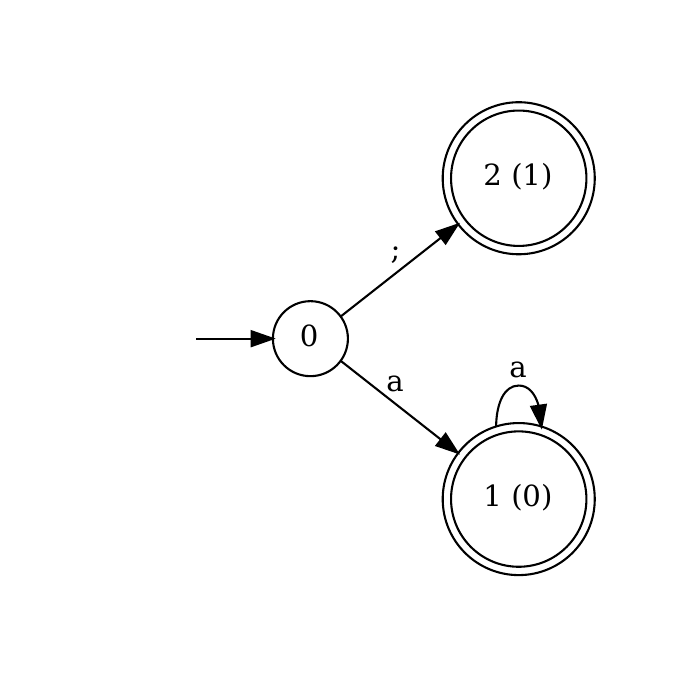}\\[0.5em]
(a) $a^+$ and \code{;}
\end{minipage}
\hfill
\begin{minipage}[b]{0.24\linewidth}
\centering
\includegraphics[width=\linewidth]{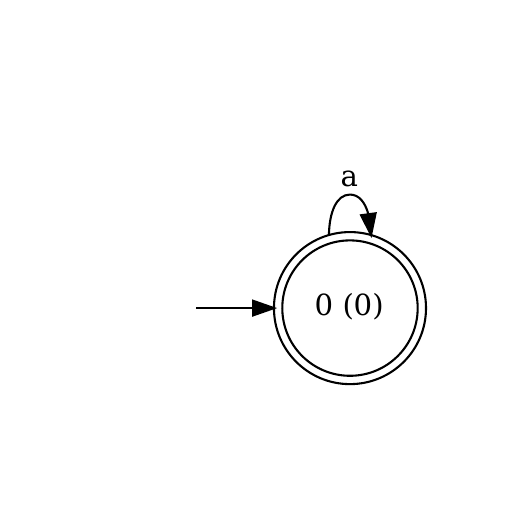}\\[0.5em]
(b) $a^*$
\end{minipage}
\hfill
\begin{minipage}[b]{0.3\linewidth}
\centering
\includegraphics[width=\linewidth]{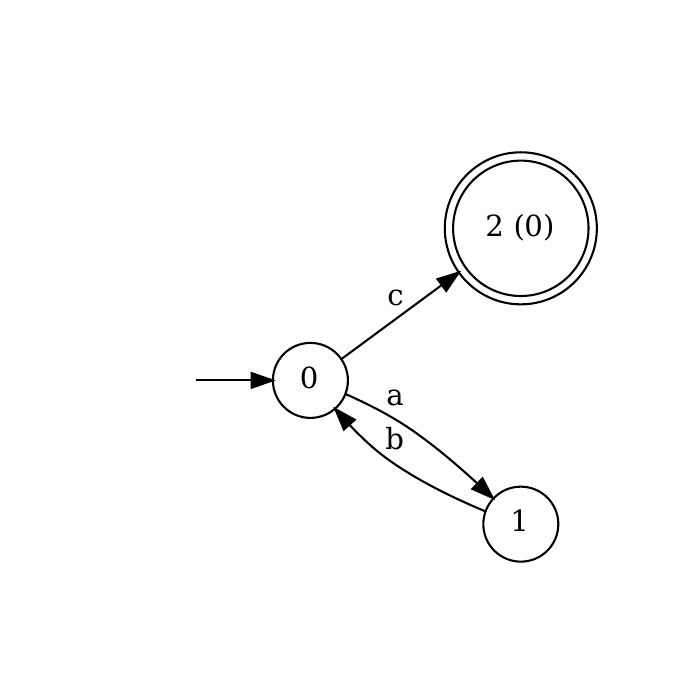}\\[0.5em]
(c) $(ab)^*c$
\end{minipage}
\caption{Minimized automata drawn by the library described here. Double circles are accepting states, labeled with the
state number and the token identifier. In (a) nothing enters state 0 and no other state consumes \code{;}, so \code{;}
is certified. In (b) state 0 accepts and re-enters itself on $a$, and in (c) it is re-entered through the cycle $a, b$.
In both examples the only state consuming the candidate byte is the initial one, so dropping the re-entrancy condition
would certify it and break longest match. The shipped predicate rejects the candidate in (b) and (c) and accepts
\code{;} in (a).}
\label{fig:reentrancy}
\end{figure}

Being discarded is not by itself sufficient. Cutting the line comment \code{//ab} between its last two bytes leaves
\code{//a} on the left, which is a comment and is discarded, and \code{b} on the right, which is an identifier and is
not. The requirement is that both halves of the severed token are discarded, and that the restarted scan rejoins the
scan it interrupted. Write
\[
T(q) = \{\, \tau(p) \mid p = (\delta^+)^*(q, v) \text{ for some } v \in \Sigma^*,\ \tau(p) \text{ defined} \,\}
\]
for the tokens still reachable from $q$, noting $\tau(q) \in T(q)$ whenever $\tau(q)$ is defined.

\begin{definition}[Certified modulo $I$]
\label{def:trivia-certified}
A symbol $b \in \Sigma$ is \emph{certified modulo $I$} if every $q \in Q^+$ with $\delta^+(q, b)$ defined satisfies
$A(q)$ or $C(q)$, where
\[
A(q) \;\equiv\; q = q_0 \text{ and } q_0 \text{ is not re-entrant},
\]
and $C(q)$ holds when the following three are satisfied:
\begin{enumerate}
  \item $\tau(q)$ is defined,
  \item $\delta^+(q_0, b)$ is defined and equal to $\delta^+(q, b)$, and
  \item $T(q) \subseteq I$.
\end{enumerate}
\end{definition}

The two alternatives are not symmetric, and the asymmetry matters. $A(q)$ is Definition~\ref{def:certified}, under which
the initial state may consume $b$ only while nothing returns to it mid-scan. $C(q)$ is the relaxation, and it is
available to \emph{every} consuming state including a re-entered $q_0$. Requiring $A$ of $q_0$ unconditionally would be
strictly weaker than necessary. For a discarded token \code{a*} the minimized automaton is one accepting, self-looping
state, so $q_0$ is re-entrant and $A$ fails. Yet $C(q_0)$ holds, and cutting between two \code{a}s replaces one
discarded token by two, which $\pi_I$ deletes either way.

Setting $I = \emptyset$ recovers Definition~\ref{def:certified} exactly, since $\tau(q) \in T(q)$ makes clause 3
unsatisfiable and only $A$ survives. Every certified symbol is certified modulo $I$ for every $I$.

The three conditions of $C$ do distinct work. Requiring $\tau(q)$ to be defined lets the left chunk stop at the cut with
a complete token. Requiring $T(q) \subseteq I$ makes both that token and whatever the severed remainder becomes
discardable, and since $\tau(q) \in T(q)$ it also forces the left chunk's token to be discarded. The equality
$\delta^+(q, b) = \delta^+(q_0, b)$ is the load-bearing condition, since after consuming $b$ the restarted scan occupies
the same state as the scan it interrupted. The two agree on every subsequent byte, and the disturbance is confined to
the single token containing the cut. With the equality the definition is a test on the transition table, and without it
a statement about scanner equivalence. Section~\ref{sec:coalesced} weakens the equality to an equivalence that is still
a test on the tables, and Section~\ref{sec:exact} decides the statement itself. Where the two local tests must be told
apart, Definition~\ref{def:trivia-certified} is called the \emph{state-equality condition}.

\begin{theorem}[Split invariance modulo $I$]
\label{thm:trivia}
Let $w$ be completely tokenizable and let $0 = c_0 < c_1 < \cdots < c_m = |w|$ be offsets such that each interior $c_j$
holds a symbol certified modulo $I$. Then
\[
\pi_I(\mathrm{tok}(w)) \;=\; \pi_I\!\left(\mathrm{tok}(w_{[c_0,c_1)}) \cdot \mathrm{tok}(w_{[c_1,c_2)}) \cdots
\mathrm{tok}(w_{[c_{m-1},c_m)})\right),
\]
where $\cdot$ denotes concatenation of token sequences and $w_{[a,b)}$ the corresponding slice. Every slice
$w_{[c_j,c_{j+1})}$ is also completely tokenizable, so every symbol certified modulo $I$ is safe modulo $I$ in the sense
of Definition~\ref{def:safe}.
\end{theorem}

\begin{proof} For empty $w$ the partition degenerates as in Theorem~\ref{thm:main} and both images are empty; assume
$|w| > 0$. By induction on $m$ it suffices to treat one interior cut at $c$ holding $b$. The induction is well founded
because the one-cut argument below shows each chunk consumes its whole slice, so both slices are themselves completely
tokenizable and the remaining cuts may be applied to them in turn.

Suppose first that $c$ is a token boundary of $w$. The left slice $w_{[0,c)}$ emits the same tokens as $w$ on $[0,c)$ by
Lemma~\ref{lem:truncation} taken with $s = c$. The scanner enters every token in $q_0$ and carries no state across
boundaries, so the chunk beginning at $c$ starts in the same state as the scan of $w$ does there and reads the same
bytes. The token sequences on the two sides are equal before $\pi_I$ is applied, and $\pi_I$ preserves equality.

Otherwise $c$ lies strictly inside an emitted token spanning $[s, e)$ and carrying token $t$. Let
$u = w_s \cdots w_{c-1}$, which is nonempty, and let $q = (\delta^+)^*(q_0, u)$ be the state immediately before $b$ is
consumed. The token's match path continues through $c$, so $\delta^+(q, b)$ is defined and $q$ satisfies $A(q)$ or
$C(q)$. The state $q$ cannot satisfy $A(q)$, which would need $q = q_0$ with $q_0$ not re-entrant, yet
$(\delta^+)^*(q_0, u) = q_0$ with $u$ nonempty is precisely re-entrancy. It follows that $C(q)$ holds, and it holds
whether or not $q$ happens to be a re-entered $q_0$; that is the only place the generalization over
Definition~\ref{def:certified} is used.

Consider the chunk ending at $c$, which is the prefix $w_{[0,c)}$ with $c > s$. By Lemma~\ref{lem:truncation} the chunk
emits the same tokens as $w$ on $[0, s)$ and reaches $s$ in $q_0$. The chunk then consumes $u$ and arrives at $q$ with
the chunk exhausted. Since $\tau(q)$ is defined by condition 1, the last accepting position recorded is $c$ itself, so
the chunk emits the single token $\tau(q)$ spanning $[s, c)$ and consumes to its end.

Consider the chunk beginning at $c$. The chunk starts in $q_0$ and consumes $b$, reaching $\delta^+(q_0, b)$. By
condition 2 this equals $\delta^+(q, b)$, the state the scan of $w$ occupies after consuming the same byte at the same
offset. From there both runs read identical bytes from identical states, so they record identical accepting positions,
and the chunk emits $t$ spanning $[c, e)$. Both runs then stand at $e$ in $q_0$, and $e$ is a boundary of $w$, so the
remainders agree.

The two sides differ only in that $t$ on $[s, e)$ is replaced by $\tau(q)$ on $[s, c)$ followed by $t$ on $[c, e)$. Both
$\tau(q)$ and $t$ lie in $T(q)$, which condition 3 places inside $I$, so $\pi_I$ deletes all three occurrences and the
images coincide. In both cases each chunk consumed its whole slice, so a single cut meets both requirements of
Definition~\ref{def:safe}.
\end{proof}

Corollary~\ref{cor:malformed} has no analog here. When the serial scan does not tokenize the input completely, no
$\pi_I$-prefix relation is claimed, and a caller must still check every chunk's consumed length before trusting the
concatenation at all.

\begin{proposition}[Useful certificates modulo $I$]
\label{prop:useful}
Proposition~\ref{prop:vacuous} and Corollary~\ref{cor:useful} carry over: a symbol no live state consumes is certified
modulo $I$ for every $I$ and occurs in no completely tokenizable input, and a symbol certified modulo $I$ is non-vacuous
exactly when $\delta^+(q_0, b)$ is defined.
\end{proposition}

\begin{proof}
The first two claims are those of Proposition~\ref{prop:vacuous}. Its argument does not depend on which of the two
conditions admitted $b$, since Definition~\ref{def:trivia-certified} likewise quantifies over an empty set. For the
third, every state consuming $b$ live satisfies $A$ or $C$. Under $A$ it is $q_0$, and under $C$ clause 2 gives
$\delta^+(q, b) = \delta^+(q_0, b)$, so either way some live state consumes $b$ exactly when $\delta^+(q_0, b)$ is
defined.
\end{proof}

As in the exact setting, symbols that no live state consumes are certified but useless, since a caller cannot act on one
and a planner searching for one scans the whole input for nothing. Both the library predicate and every table below
report the \emph{useful} set, intersecting the certificate with the initial state's live transitions. This matters for
reading Table~\ref{tab:applicability}. The definitions vacuously certify the raw control bytes that no token admits, and
the table's JSON row omits them and reports only the bytes an input can actually contain.

\begin{proposition}[Strict extension]
\label{prop:extends}
Every certified symbol is certified modulo $I$ for every $I \subseteq T$, and there are token sets and sets $I$ for
which the converse fails.
\end{proposition}

\begin{proof}
The first claim is immediate, since Definition~\ref{def:trivia-certified} offers $q = q_0$ as an alternative and its
re-entrancy proviso is that of Definition~\ref{def:certified}. For the second, take identifiers, punctuation and a
whitespace run over space and newline, with the whitespace token in $I$. Newline is not certified, since the state
inside a whitespace run consumes it into a co-accessible state. It is, however, certified modulo $I$, since that state
accepts the whitespace token and reaches no other token, and advancing on newline from it and from $q_0$ reaches the run
state either way.
\end{proof}

\subsection{A stronger local test}
\label{sec:coalesced}

Definition~\ref{def:trivia-certified} asks the restarted scan to stand in the very state of the scan it interrupted. The
coalesced test, which release 2.2.0 of munch ships, asks less, namely that the two states have the same future once
discarded kinds are not told apart. Complete $A^+$ with a sink $\bot$, reading an undefined $\delta^+(q, a)$ as $\bot$
and letting $\delta^+(\bot, a) = \bot$. Give every state an observation: \emph{reject} when $\tau$ is undefined there,
as at $\bot$, \emph{discard} when $\tau(q) \in I$, and $\mathrm{keep}(\tau(q))$ otherwise. Write $p \equiv_I r$ when for
every $z \in \Sigma^*$ the states reached from $p$ and from $r$ on $z$ carry the same observation, which is Moore
equivalence of the completed automaton with the discarded kinds coalesced into one color. Discard stays distinct from
reject, and the empty $z$ compares the two states themselves.

\begin{definition}[Certified modulo $I$, coalesced]
\label{def:coalesced}
A symbol $b \in \Sigma$ is \emph{certified modulo $I$ under the coalesced test} if every $q \in Q^+$ with
$\delta^+(q, b)$ defined satisfies $A(q)$ or $C'(q)$, where $A$ is as in Definition~\ref{def:trivia-certified} and
$C'(q)$ holds when the following three are satisfied:
\begin{enumerate}
  \item $\tau(q) \in I$,
  \item $\delta^+(q_0, b)$ is defined and $\delta^+(q_0, b) \equiv_I \delta^+(q, b)$, and
  \item $T(\delta^+(q, b)) \subseteq I$.
\end{enumerate}
\end{definition}

Two things changed: equality of the two states became equivalence, and the requirement that the future be discarded
applies to the state past the cut rather than to the state before it. The old clause 3 asked that everything reachable
from $q$ is discarded, but only the continuations through $b$ belong to the severed token, and the new clause 1 discards
the left fragment itself.

\begin{theorem}[Soundness of the coalesced test]
\label{thm:coalesced}
Theorem~\ref{thm:trivia} holds with Definition~\ref{def:coalesced} in place of Definition~\ref{def:trivia-certified},
and so does Proposition~\ref{prop:useful}.
\end{theorem}

\begin{proof}
As in Theorem~\ref{thm:trivia}, one interior cut at $c$ holding $b$ suffices, and a cut at a token boundary of $w$ is
settled as in that proof, without any use of $C$. Otherwise $c$ lies strictly inside an emitted token spanning $[s, e)$
and carrying $t$. Let $u = w_s \cdots w_{c-1}$ and $q = (\delta^+)^*(q_0, u)$. Then $\delta^+(q, b)$ is defined, and
$A(q)$ fails as before, so $C'(q)$ holds.

The chunk ending at $c$ emits the same tokens as $w$ on $[0, s)$ and reaches $s$ in $q_0$, by
Lemma~\ref{lem:truncation}. It then consumes $u$ and arrives at $q$ with the chunk exhausted. By clause 1, $\tau(q)$ is
defined and lies in $I$, so the chunk emits $\tau(q)$ spanning $[s, c)$ and consumes to its end.

Write $p = \delta^+(q, b)$ for the state the scan of $w$ occupies after $c$, and $p_0 = \delta^+(q_0, b)$ for the state
the chunk beginning at $c$ occupies there. From $c + 1$ on, both scans read the same bytes, and every accepting position
a scan records is reached by a path in $A^+$, so for each $k \geq 0$ the two scans record position $c + 1 + k$ exactly
when the states they reach from $p$ and $p_0$ on the same $k$ bytes accept. By clause 2 those two states carry the same
observation, so one accepts exactly when the other does. The scan of $w$ records $e$ last, so the chunk records $e$ last
as well and emits a token spanning $[c, e)$. The kind the scan of $w$ emits there is $t \in T(p)$, which clause 3 places
in $I$, so the state reached from $p$ at $e$ observes discard. The state reached from $p_0$ at $e$ observes the same, so
the chunk's token is discarded too. Both scans restart in $q_0$ at $e$, which is a boundary of $w$, and the remainders
agree.

The two sides differ only in that $t$ on $[s, e)$ is replaced by $\tau(q)$ on $[s, c)$ followed by a discarded token on
$[c, e)$. All three lie in $I$, so $\pi_I$ deletes them and the images coincide. Each chunk again consumed its whole
slice.

Proposition~\ref{prop:useful} carries over with its proof. Definition~\ref{def:coalesced} likewise quantifies over an
empty set when no live state consumes $b$, and its clause 2 likewise requires $\delta^+(q_0, b)$ to be defined.
\end{proof}

Nothing about lengths is given up, since both scans end the severed token at $e$. Compared with equality, the
equivalence gives up only the identity of the state at $e$, which the projection never sees.

\begin{proposition}[Strictly stronger, and still incomplete]
\label{prop:coalesced-strict}
Every symbol certified by Definition~\ref{def:trivia-certified} is certified by Definition~\ref{def:coalesced}. The
converse fails, once for each of the two changes. And there are token sets and sets $I$ with a symbol safe modulo $I$ at
every occurrence that Definition~\ref{def:coalesced} refuses.
\end{proposition}

\begin{proof}
For the inclusion, suppose $C(q)$. Then $\tau(q) \in T(q) \subseteq I$, which is clause 1; equal states are equivalent,
which is clause 2; and $T(\delta^+(q, b)) \subseteq T(q) \subseteq I$, which is clause 3. $A$ is unchanged.

For the first change, take the literals \code{a} and \code{ab} of one discarded kind $D$, the literal \code{b} of a
second discarded kind $E$, and a kept literal \code{k}. After \code{a} the scan stands in a state $q$ accepting $D$.
Advancing on \code{b} from $q$ reaches the state accepting \code{ab}, of kind $D$, and advancing from $q_0$ reaches the
state accepting \code{b}, of kind $E$. The two accept different kinds, so the minimizer keeps them apart and
Definition~\ref{def:trivia-certified} refuses \code{b}. Neither state has a live transition, so once $D$ and $E$ share a
color the two are equivalent. Clause 1 holds since $\tau(q) = D$, and clause 3 since only $D$ lies past the cut. The
only other state consuming \code{b} is $q_0$, which is not re-entrant, so Definition~\ref{def:coalesced} admits
\code{b}.

For the second change, take the literals \code{a}, \code{b}, \code{c} and \code{ab} of one discarded kind and the kept
literal \code{ac}. After \code{a} the scan stands in a state $q$ from which \code{ac} is still reachable, so
$T(q) \not\subseteq I$ and Definition~\ref{def:trivia-certified} refuses \code{b}. The minimizer merges the states
accepting \code{ab} and \code{b}, both discarded leaves of one kind, so $\delta^+(q, b) = \delta^+(q_0, b)$ and clauses
1 and 2 of Definition~\ref{def:trivia-certified} hold at $q$, where only its clause 3 fails. That leaf reaches nothing
kept, $\tau(q)$ is discarded and the only other state consuming \code{b} is again $q_0$, so
Definition~\ref{def:coalesced} admits \code{b}.

For incompleteness, take the literals \code{a}, \code{b}, \code{c} and \code{abc} of one discarded kind and the kept
literals \code{k} and \code{kk} of one kind $K$. Every input over these four letters tokenizes completely, since each
letter is a token. Kept output comes only from maximal runs of \code{k}, where a run of length $n$ emits $(K, 2)$
exactly $\lfloor n/2 \rfloor$ times and then $(K, 1)$ when $n$ is odd, by longest match. A cut before a \code{b} lies
inside no such run and leaves every run maximal in its chunk, so both chunks tokenize completely and the kept streams
concatenate to the whole input's kept stream. Hence every cut before \code{b} is safe. Yet after \code{a}, advancing on
\code{b} reaches the non-accepting prefix of \code{abc}, which observes reject, while advancing from $q_0$ reaches the
state accepting \code{b}, which observes discard. Clause 2 fails on the empty continuation, so
Definition~\ref{def:coalesced} refuses \code{b}, and Definition~\ref{def:trivia-certified} with it.
\end{proof}

The obstruction in the last example is structural. A safe cut may split one discarded token into several, with the
restarted scan rejoining the interrupted scan only after further restarts. Equivalence of the two states past the cut
forces the two scans to end their current tokens together, which in such a case they do not. For all three token sets
and the named example of Section~\ref{sec:decide}, \code{validation.cpp} asserts the outcomes of both local tests, the
decision of Section~\ref{sec:exact} and a bounded oracle.

\subsection{The exact decision}
\label{sec:exact}

Both local tests are sufficient conditions, and this section decides the property they approximate. Throughout, the
alphabet is restricted to the bytes some live state consumes, since no other byte occurs in a completely tokenizable
input, and $n = |Q^+|$.

\begin{definition}[Safe modulo $I$]
\label{def:safe}
A symbol $b$ is \emph{safe modulo $I$} if for every completely tokenizable $w$ and every offset $0 < c < |w|$ with
$w_c = b$, both $w_{[0,c)}$ and $w_{[c,|w|)}$ are completely tokenizable and
\[
\pi_I(\mathrm{tok}(w)) = \pi_I(\mathrm{tok}(w_{[0,c)})) \cdot \pi_I(\mathrm{tok}(w_{[c,|w|)})).
\]
\end{definition}

The contract needs the chunk-completeness clause. Over the single discarded token \code{ab}, cutting \code{ab} before
\code{b} leaves two empty kept streams and two chunks neither of which tokenizes. The clause is also what makes one cut
enough, since applying Definition~\ref{def:safe} to one interior cut leaves two completely tokenizable chunks. The
remaining cuts apply to those in turn, so the conclusion of Theorem~\ref{thm:trivia} follows for any partition at safe
symbols.

\begin{theorem}[Exact decision]
\label{thm:exact}
For every token automaton and every $I \subseteq T$, whether $b$ is safe modulo $I$ is decidable. A procedure singly
exponential in $n$ and polynomial in $|\Sigma|$ and $|T|$ decides it, and returns a completely tokenizable input and a
cut that violate Definition~\ref{def:safe} whenever $b$ is not safe. The question lies in PSPACE, and
Theorem~\ref{thm:pspace} shows it complete for that class.
\end{theorem}

The proof builds three finite automata from the token automaton and tests two products for emptiness and one for
equality of outputs. The first automaton recognizes the completely tokenizable words. It guesses where each token ends.
After choosing a token end it keeps following that token's match path as an obligation, since any acceptance after a
further byte means the scanner would have matched a longer token there, which invalidates the choice. Over the tokens
\code{a}, \code{ab} and \code{b}, guessing that the input \code{ab} scans as \code{a} then \code{b} commits \code{a}.
The path of \code{a} accepts again at the next byte, where the scanner matches \code{ab}, so the guess is rejected. This
automaton is obtained from the byte projection of the splitting companion's armed-run
automaton~\cite{nidhogg2026certifiedsplitting} by using one compiled token automaton. The companion's armed-run
automaton reads the boundary marks this one guesses. That paper establishes that greedy complete segmentation is
regular.

\begin{lemma}[Complete scans]
\label{lem:complete}
Let $S$ be the nondeterministic automaton with states $(O, q)$, where $O \subseteq Q^+$ and
$q \in Q^+ \cup \{\mathrm{fresh}\}$, initial state $(\emptyset, \mathrm{fresh})$, and final states those with
$q = \mathrm{fresh}$. On a byte $a$ from $(O, q)$, let
$O' = \{\delta^+(o, a) \mid o \in O,\ \delta^+(o, a) \text{ defined}\}$. If some member of $O'$ accepts there is no
transition. Otherwise let $p = q_0$ when $q$ is fresh and $p = q$ when it is not, and let $q' = \delta^+(p, a)$. If $q'$
is undefined there is no transition. Otherwise $(O', q')$ is a successor, and when $\tau(q')$ is defined so is
$(O' \cup \{q'\}, \mathrm{fresh})$. Then $S$ accepts exactly the completely tokenizable words, and the positions at
which an accepting run moves to fresh are exactly the positive token boundaries of the scan.
\end{lemma}

The members of $O$ are the \emph{obligations}, the current states of the splitting companion's armed runs. Each is the
state reached along the match path of a committed token, and it accepts when it is added. If one of them accepts after a
further byte, the committed token had a longer match and the commitment was wrong. Fresh is distinct from the state
$q_0$ reached by a nonempty path, which a re-entrant initial state allows.

\begin{proof}
Every accepting position a scan records is reached by a path lying in $A^+$, since the accepting state is itself live
and so is every state before it. The scanner from offset $i$ records a position $j > i$ exactly when
$(\delta^+)^*(q_0, w_{[i,j)})$ accepts, and $\delta$ may be read as $\delta^+$ below without loss.

Let $w$ be completely tokenizable, with tokens spanning $[s_0, s_1), \ldots, [s_{m-1}, s_m)$ and $s_m = |w|$. Run $S$
along $w$, continuing inside each token and moving to fresh at each $s_j$ with $j \geq 1$. Each token's match path lies
in $A^+$, so every $q'$ along it is defined, and its end state accepts, so each move to fresh is permitted. An
obligation added at $s_j$ is the end state of the token before $s_j$. Were it to accept after $k \geq 1$ further bytes,
the scanner would have recorded $s_j + k$ for that token and emitted a longer one, by longest match. Hence no obligation
accepts after a further byte, the run is defined throughout, and it ends at $|w|$ in fresh. $S$ accepts $w$.

Conversely, let an accepting run of $S$ on $w$ move to fresh at $0 < s_1 < \cdots < s_m = |w|$, with $s_0 = 0$, and
$m = 0$ when $w$ is empty. Each part $w_{[s_j, s_{j+1})}$ is nonempty, since a move to fresh consumes a byte, and has a
$\delta^+$-path from $q_0$ ending in an accepting state, so it is an accepted word. Suppose the scanner stands at $s_j$
in $q_0$, which holds at $0$. Its token there is the longest accepted prefix of $w_{[s_j, |w|)}$. The prefix
$w_{[s_j, s_{j+1})}$ is accepted. A longer accepted prefix, of length $s_{j+1} - s_j + k$ with $k \geq 1$, would have a
$\delta^+$-path through the state the run committed at $s_{j+1}$ and $k$ bytes further to an accepting state. The state
committed at $s_{j+1}$ entered $O$ there, the run follows it along the same path while the path is defined, and that
obligation would accept after $k$ bytes, which the run forbids. The scanner thus emits exactly $w_{[s_j, s_{j+1})}$ and
restarts at $s_{j+1}$ in $q_0$. By induction the scanner's boundaries are $s_0, \ldots, s_m$, and it consumes all of
$w$.
\end{proof}

The second automaton recognizes the words the scanner does not tokenize completely.

\begin{lemma}[Failed scans]
\label{lem:failed}
Let $F$ have the states and the initial state of $S$, none of those states final, and in addition the
\emph{failure states} $(O, q, \mathrm{fail})$ with $q \in Q^+$ non-accepting or $q = \bot$, all of them final.
Transitions among the states of $S$ are those of $S$. From $(O, \mathrm{fresh})$ on $a$, form $O'$ as in
Lemma~\ref{lem:complete}, with no transition if some member of $O'$ accepts, and let $q' = \delta^+(q_0, a)$. The
successor is $(O', \bot, \mathrm{fail})$ when $q'$ is undefined, and $(O', q', \mathrm{fail})$ when $q'$ is defined and
does not accept. When $q'$ accepts there is no failure successor. From $(O, q, \mathrm{fail})$ on $a$, form $O'$ with
the same proviso and let $q' = \delta^+(q, a)$, with $\bot$ fixed. The successor is $(O', \bot, \mathrm{fail})$ when
$q'$ is undefined, $(O', q', \mathrm{fail})$ when $q'$ does not accept, and there is none when $q'$ accepts. Then $F$
accepts exactly the words the scanner does not tokenize completely.
\end{lemma}

\begin{proof}
Let the scanner halt on $w$ at $\mathrm{con}(w) < |w|$ after emitting tokens with boundaries
$s_0, \ldots, s_m = \mathrm{con}(w)$. Run $F$ as in Lemma~\ref{lem:complete} up to $s_m$, which is permitted because
each emitted token was a longest match against the whole of $w$, and then enter failure on $w_{s_m}$. The suffix
$v = w_{[s_m, |w|)}$ is nonempty and has no accepted prefix of positive length, so the $\delta^+$-path from $q_0$ along
$v$ visits no accepting state while it is defined. Each failure step is permitted, and the obligations from earlier
tokens never accept, for the same reason as before. The run ends in a failure state, so $F$ accepts $w$.

Conversely, an accepting run of $F$ moves to fresh at $s_1 < \cdots < s_m$ and enters failure on the byte at $s_m$, with
a nonempty suffix $v$. The argument of Lemma~\ref{lem:complete} shows that the scanner's tokens are exactly the
committed parts, since the obligations are carried across $v$ as well. From $s_m$ the failure path shows that no prefix
of $v$ of positive length is accepted, because the path visits no accepting state while it is defined, and once it is
undefined no longer prefix can be accepted. The scanner halts at $s_m < |w|$.

The empty word is not accepted, since entering failure consumes a byte. An accepting $q_0$ needs no exception, since
acceptance is tested only after a byte is consumed, exactly as the scanner emits only positive-length tokens.
\end{proof}

Both constructions are direct, and neither complements the other by subset construction. Each has at most $2(n + 1)2^n$
states.

\begin{lemma}[Chunk completeness]
\label{lem:chunks}
Let $M_b$ be the marked inputs $u \# v$ with $u$ nonempty, $v$ beginning with $b$, and $\#$ a letter outside $\Sigma$.
Let $P_L$ run $S$ on $uv$, carrying its state across $\#$, beside $F$ on $u$, requiring $F$ to be in a failure state at
$\#$ and ignoring $v$. Let $P_R$ run $S$ on $uv$ in the same way, beside $F$ started afresh at $\#$ and run on $v$,
requiring it to be in a failure state at the end. Then $P_L$ accepts $u \# v$ exactly when $uv$ is completely
tokenizable and $u$ is not, and $P_R$ exactly when $uv$ is and $v$ is not. Hence both chunks of every cut before $b$ are
completely tokenizable if and only if both products are empty. Reachability decides that, and a shortest accepting path
of a nonempty product spells a violating input and cut.
\end{lemma}

\begin{proof}
Immediate from Lemmas~\ref{lem:complete} and~\ref{lem:failed}. The format of $M_b$ is checked by a four-state automaton
that admits exactly the inputs and cuts Definition~\ref{def:safe} quantifies over, and the product of finitely many
finite automata is finite.
\end{proof}

The third automaton carries the kept output. Let $\Gamma = \{z\} \cup \{K_t \mid t \in T \setminus I\}$. Encode a kept
token $(t, \ell)$ as $z^\ell K_t$ and a discarded token as the empty word, and write $\mathrm{code}$ for the extension
to streams. The map is injective on kept streams, since every $K_t$ ends exactly one unary length and every length is
positive. Equality of codes is equality of the images under $\pi_I$, lengths included.

\begin{lemma}[Kept output]
\label{lem:output}
Let $S_{\mathrm{out}}$ have states $(O, q, \kappa)$ with $(O, q)$ a state of $S$ and
$\kappa \in \{\mathrm{keep}, \mathrm{discard}, \mathrm{none}\}$, with $\kappa = \mathrm{none}$ exactly when $q$ is
fresh. Its initial state is $(\emptyset, \mathrm{fresh}, \mathrm{none})$ and its final states are those with
$q = \mathrm{fresh}$. A transition of $S$ from a fresh state first guesses $\kappa$; a transition inside a token keeps
it. The transition emits $z$ when $\kappa = \mathrm{keep}$ and nothing otherwise. A move to fresh at a state $q'$ is
permitted only when $\kappa = \mathrm{keep}$ exactly if $\tau(q') \notin I$, and then emits $K_{\tau(q')}$ when kept.
Then $S_{\mathrm{out}}$ has an accepting run on $w$ exactly when $w$ is completely tokenizable. Every accepting run
emits $\mathrm{code}(\pi_I(\mathrm{tok}(w)))$.
\end{lemma}

\begin{proof}
By Lemma~\ref{lem:complete} any accepting run moves to fresh exactly at the scanner's token boundaries, so the tokens it
commits are the scanner's. Each guess is checked at the token's end against the kind the scanner emits, so a wrong guess
has no accepting continuation. A right guess emits $z$ once per byte of a kept token and then its kind, which is its
code. The guess exists for every completely tokenizable $w$, so an accepting run exists exactly then.
\end{proof}

Let $P_O$ run $S_{\mathrm{out}}$ on $uv$ across $\#$ beside a second copy of $S_{\mathrm{out}}$. The second copy runs on
$u$, must be fresh at $\#$, restarts there in $(\emptyset, \mathrm{fresh}, \mathrm{none})$, and runs on $v$. Each edge
of $P_O$ carries two output words, $h$ from the whole scan and $g$ from the split one, each of at most two letters. Once
$P_L$ and $P_R$ are empty, Lemmas~\ref{lem:chunks} and~\ref{lem:output} show that the accepting paths of $P_O$ are
exactly the inputs and cuts of Definition~\ref{def:safe}, and $b$ is safe exactly when every accepting path has $h = g$.
That is a question about a finite labeled graph.

\begin{lemma}[Output equality on a finite graph]
\label{lem:equality}
Let $G$ be a finite directed multigraph with $N$ vertices, an initial vertex, a set of final vertices, and two words
$h(e), g(e) \in \Gamma^*$ on every edge $e$, extended to paths by concatenation. Whether $h(\pi) = g(\pi)$ for every
accepting path $\pi$ is decidable in time polynomial in the size of $G$ and its labels. When it fails, an accepting path
with $h(\pi) \neq g(\pi)$ of at most $2N - 1$ edges exists and is constructed.
\end{lemma}

\begin{proof}
Use breadth-first search in both directions to find the vertices on accepting paths, and discard all other vertices.
Every retained edge then lies on an accepting path. Read words over $\Gamma$ as elements of the free group on $\Gamma$.
A positive word, one without inverse letters, is its own reduced form, so two positive words are equal in the group
exactly when they are equal as words. For each retained vertex $q$ choose a shortest path $\alpha_q$ from the initial
vertex and let $d_q$ be the reduced form of $h(\alpha_q)^{-1} g(\alpha_q)$, the \emph{delay} at $q$. Check two
conditions: (E) $d_q = h(e)^{-1} d_p\, g(e)$ for every retained edge $e$ from $p$ to $q$, and (F) $d_q = 1$ for every
retained final $q$.

Suppose every accepting path has equal outputs. Fix a retained $q$ and an accepting suffix $\beta$ from $q$, with
$C = h(\beta)$ and $D = g(\beta)$. For any path $\alpha$ from the initial vertex to $q$, with $A = h(\alpha)$ and
$B = g(\alpha)$, the path $\alpha\beta$ is accepting, so $AC = BD$ and hence $A^{-1}B = CD^{-1}$ in the group. The delay
of every path into $q$ is the one value $CD^{-1}$. Applying this to $\alpha_q$ and to $\alpha_p e$ gives (E). Applying
it to $\alpha_q$ with $q$ final, where $\alpha_q$ is itself accepting, gives (F).

Suppose (E) and (F). Along any path from the initial vertex the delay starts at $1$, which is the value of $d$ at the
initial vertex since its chosen path is empty, and appending $e$ from $p$ to $q$ carries $d_p$ to
$h(e)^{-1} d_p\, g(e) = d_q$. Every path into $q$ has delay $d_q$, every accepting path has delay $1$ by (F), and its
two outputs are equal in the group, hence as words.

For the counterexample, choose for each retained $q$ a shortest accepting suffix $\beta_q$, by breadth-first search
backwards from the finals. If (F) fails at $q$, then $\alpha_q$ is accepting with unequal outputs. If (E) fails at $e$
from $p$ to $q$, consider $\alpha_q \beta_q$ and $\alpha_p e \beta_q$. If both had equal outputs, the computation above
with the common suffix $\beta_q$ would give both prefixes the delay $CD^{-1}$, contradicting the failure of (E). One of
them has unequal outputs, and comparing the two explicit output pairs identifies it. Both have at most $2N - 1$ edges.
Each chosen path has at most $N - 1$ edges, so each delay has at most $(N - 1)M$ letters, where $M$ bounds
$|h(e)| + |g(e)|$. One check costs $O(NM)$, and the whole test is polynomial.
\end{proof}

\begin{proof}[Proof of Theorem~\ref{thm:exact}]
Build $P_L$ and $P_R$ and test them for emptiness, by Lemma~\ref{lem:chunks}. If either is nonempty, $b$ is unsafe and a
shortest accepting path spells the input and cut. Otherwise build $P_O$ and apply Lemma~\ref{lem:equality}. If every
accepting path has equal outputs, $b$ is safe by Lemma~\ref{lem:output} and the injectivity of the code. If not, the
constructed path spells an input $uv$ and the cut $|u|$ whose chunks are complete and whose kept streams differ. In each
case that yields a counterexample, the scanner replays it, and \code{validation.cpp} does so for every negative answer
it gives.

For the size, $S$ has at most $(n + 1)2^n$ states, $F$ at most twice that, and $S_{\mathrm{out}}$ at most $(2n + 1)2^n$.
With the four phases of the format every product has at most $N = 16(n + 1)^2 4^n$ states and $O(|\Sigma| N)$ edges,
each with at most two letters per output. Building these products and running Lemma~\ref{lem:equality} takes time singly
exponential in $n$ and polynomial in $|\Sigma|$ and $|T|$. The kind identifiers affect only the size of a letter.

For the space bound, a vertex of any product consists of two subsets of $Q^+$, two token states, two guesses and a
phase, so it has a description of polynomial length. Its successors on a byte are computable in polynomial time from the
tables. Unsafety has a counterexample path of fewer than $N$ edges in $P_L$ or $P_R$, or of at most $2N - 1$ edges in
$P_O$ by Lemma~\ref{lem:equality}, whose two outputs have fewer than $4N$ letters each. A nondeterministic machine
guesses such a path one vertex at a time with a step counter of $O(n)$ bits. For output inequality it also guesses a
position $j \leq 4N$ and keeps the two output lengths and the letter each output has at position $j$. At a final vertex
it accepts when the lengths differ or the two letters do, which detects exactly unequal outputs. Unsafety is therefore
in NPSPACE, which is PSPACE, and PSPACE is closed under complement.
\end{proof}

The upper bound of Theorem~\ref{thm:exact} is matched by a lower bound that holds with three symbols and at most two
token kinds. The source problem is intersection non-emptiness: given complete DFAs $A_1, \ldots, A_m$ over a common
alphabet, is some word accepted by all of them? Kozen proved it PSPACE-complete~\cite[Lemma~3.2.3,
p.~261]{kozen1977lower}. We take the source DFAs over the alphabet $\{\code{0}, \code{1}\}$, which loses nothing, since
a fixed-length binary code of the letters, decoded by a polynomial number of added states with a rejecting sink for
malformed codes, carries any alphabet to this one.

\begin{theorem}[PSPACE-completeness]
\label{thm:pspace}
Deciding whether $b$ is safe modulo $I$, given the token automaton, $I$ and $b$, is PSPACE-complete. It remains
PSPACE-hard for token automata over three symbols with one kept and one discarded kind for which every input over the
three symbols is completely tokenizable. It also remains PSPACE-hard for token automata over three symbols with a single
kind, which is discarded.
\end{theorem}

The two restricted families separate the two clauses of Definition~\ref{def:safe}. In the first every chunk of every
input is completely tokenizable, so only the equality of the kept streams can fail. In the second every kept stream is
empty, so only chunk completeness can fail.

\begin{proof}
Membership is Theorem~\ref{thm:exact}. For hardness, the reduction is first given over six symbols, which keeps its
tokens readable, and is then coded over three. Let $A_i = (Q_i, \{\code{0}, \code{1}\}, \delta_i, s_i, F_i)$ for
$i = 1, \ldots, m$, with $m \geq 1$, be complete DFAs with languages $L(A_i)$. The token automaton has the symbols
$\code{0}, \code{1}, \code{x}, \code{\symbol{36}}, \code{[}, \code{\#}$ and two kinds, a kept kind $K$ and a discarded
kind $D$, so that $I = \{D\}$, and the symbol in question is $b = \code{\#}$, a letter of $\Sigma$. Its tokens have four
forms. The single symbol $\code{\symbol{36}}$ is a token of kind $K$, and each of the single symbols $\code{0}$,
$\code{1}$, $\code{x}$, $\code{[}$ and $\code{\#}$ is a token of kind $D$. The \emph{testers} are the words
$\code{x}^i y\,\code{\symbol{36}}$ with $1 \leq i \leq m$, $y \in \{\code{0}, \code{1}\}^*$ and $y \notin L(A_i)$, and
the \emph{wrappers} are the words $\code{[}\,\code{\#}\,\code{x}^m y\,\code{\symbol{36}}$ with
$y \in \{\code{0}, \code{1}\}^*$. Testers and wrappers are tokens of kind $D$.

A deterministic automaton accepts these tokens. Its states are $q_0$, a chain $h_1, \ldots, h_m$, a chain
$g, g_0, \ldots, g_m$, a copy of every state of every $Q_i$, and three accepting states without outgoing transitions,
$f_K$ of kind $K$ and $f_D$ and $f$ of kind $D$. From $q_0$ the symbols $\code{0}$, $\code{1}$ and $\code{\#}$ lead to
$f_D$, $\code{\symbol{36}}$ leads to $f_K$, $\code{[}$ leads to $g$ and $\code{x}$ leads to $h_1$, and both $g$ and
$h_1$ accept with kind $D$. The symbol $\code{\#}$ leads from $g$ to $g_0$ and $\code{x}$ leads from $g_j$ to $g_{j+1}$
for $j < m$, while $g_m$ loops on $\code{0}$ and $\code{1}$ and moves to $f$ on $\code{\symbol{36}}$. The symbol
$\code{x}$ also leads from $h_i$ to $h_{i+1}$ for $i < m$. On a symbol $a \in \{\code{0}, \code{1}\}$ the state $h_i$
moves to the copy of $\delta_i(s_i, a)$ and the copy of $p \in Q_i$ moves to the copy of $\delta_i(p, a)$, so that after
$\code{x}^i y$ the automaton stands at $h_i$ when $y$ is empty and at the copy of the state $A_i$ reaches on $y$
otherwise. Finally $\code{\symbol{36}}$ leads to $f$ from $h_i$ when $s_i \notin F_i$ and from the copy of $p \in Q_i$
when $p \notin F_i$. The automaton has $|Q_1| + \cdots + |Q_m| + 2m + 6$ states and is computed from the $A_i$ in
polynomial time. States outside $Q^+$, such as the copies of the states of an $A_i$ that accepts every word, lie on no
token's match path and play no part below.

Every one of the six symbols is a token, so every nonempty word over them begins with a token, the scanner never halts
with input remaining, and every input over the six symbols is completely tokenizable. The chunk-completeness clause of
Definition~\ref{def:safe} therefore holds at every cut, and $\code{\#}$ is safe modulo $I$ exactly when the kept streams
agree at every cut before a $\code{\#}$. At each token start the scanner emits the longest token that is a prefix of the
remaining input, and three observations about these prefixes are used below. First, the only tokens containing
$\code{\#}$ are the single symbol and the wrappers, and a wrapper has it as its second symbol. Second, the only token
beginning with $\code{0}$, $\code{1}$ or $\code{\#}$ is that single symbol. Third, let the remaining input be
$\code{x}^j y\,\code{\symbol{36}}\,r$ with $1 \leq j \leq m$, $y \in \{\code{0}, \code{1}\}^*$ and $r$ arbitrary. A
token beginning with $\code{x}$ is the single symbol or a tester. A tester consists of symbols $\code{x}$, then symbols
$\code{0}$ and $\code{1}$, then one $\code{\symbol{36}}$, so the only word of that shape that is a prefix of this input
is $\code{x}^j y\,\code{\symbol{36}}$, which is a tester exactly when $y \notin L(A_j)$. The scanner thus emits the
tester $\code{x}^j y\,\code{\symbol{36}}$ when $y \notin L(A_j)$ and the single symbol $\code{x}$ when $y \in L(A_j)$.

Suppose some word $y$ is accepted by every $A_i$, and let $w = \code{[}\,\code{\#}\,\code{x}^m y\,\code{\symbol{36}}$.
The scan of $w$ emits $w$ as one wrapper, which is discarded, so its kept stream is empty. The cut $c = 1$ holds
$\code{\#}$. The left chunk $\code{[}$ is one discarded token. The scan of the right chunk
$\code{\#}\,\code{x}^m y\,\code{\symbol{36}}$ emits $\code{\#}$ by the second observation, then stands at
$\code{x}^j y\,\code{\symbol{36}}$ for $j = m, \ldots, 1$ in turn and by the third observation emits the single symbol
$\code{x}$ each time, because $y \in L(A_j)$. It then emits the symbols of $y$ one at a time and finally the token
$\code{\symbol{36}}$ of kind $K$. The chunks have the kept stream $(K, 1)$ and the whole input has the empty one, so
$\code{\#}$ is not safe.

Suppose no word is accepted by every $A_i$. Let $w$ be any completely tokenizable input, that is, any input over the six
symbols, and let $0 < c < |w|$ with $w_c = \code{\#}$, and write $u = w_{[0,c)}$ and $v = w_{[c,|w|)}$. If $c$ is a
token boundary of $w$, Lemma~\ref{lem:truncation} shows that the scan of $u$ emits the tokens of $w$ on $[0, c)$ and
ends there in $q_0$. The scan of $w$ restarts at $c$ in $q_0$ and proceeds as the scan of $v$ does, so
$\mathrm{tok}(w) = \mathrm{tok}(u) \cdot \mathrm{tok}(v)$ and the kept streams agree. Otherwise $c$ lies strictly inside
a token of $w$, which by the first observation is a wrapper beginning at $c - 1$, so that
$w_{[c-1,e)} = \code{[}\,\code{\#}\,\code{x}^m y\,\code{\symbol{36}}$ for some $y \in \{\code{0}, \code{1}\}^*$ and some
$e \leq |w|$. For the left chunk, $c - 1$ is a token boundary of $w$, so by Lemma~\ref{lem:truncation} the scan of $u$
emits the tokens of $w$ on $[0, c - 1)$ and reaches $c - 1$ in $q_0$, where it emits the last symbol $\code{[}$ of $u$
as a discarded token. For the right chunk, write $r = w_{[e,|w|)}$, so that
$v = \code{\#}\,\code{x}^m y\,\code{\symbol{36}}\,r$. The scan of $v$ emits $\code{\#}$ by the second observation and
stands at $\code{x}^m y\,\code{\symbol{36}}\,r$. Let $i$ be the largest index with $y \notin L(A_i)$, which exists
because no word is accepted by every $A_i$. For $j = m, \ldots, i + 1$ in turn the scan stands at
$\code{x}^j y\,\code{\symbol{36}}\,r$ with $y \in L(A_j)$ and emits the single symbol $\code{x}$, by the third
observation. It then stands at $\code{x}^i y\,\code{\symbol{36}}\,r$ and emits the tester
$\code{x}^i y\,\code{\symbol{36}}$, which is discarded, and restarts in $q_0$ at the start of $r$. The scan of $w$
restarts in $q_0$ at the same place, after its wrapper, and from there the two scans emit the same tokens. The scan of
$w$ thus emits the tokens of $w$ on $[0, c - 1)$, a discarded wrapper, and the tokens it finds on $r$. The two chunks
emit the same tokens on $[0, c - 1)$, then discarded tokens only, and then the same tokens on $r$. The kept streams
agree, and $\code{\#}$ is safe. Hence $\code{\#}$ is safe modulo $I$ exactly when no word is accepted by every $A_i$.

For a single kind, remove the token $\code{\symbol{36}}$, that is, the state $f_K$ and the transition into it. The
automaton keeps its six symbols and has the single kind $D$, with $I = T = \{D\}$, so every kept stream is empty and
$\code{\#}$ is safe exactly when both chunks of every cut before it in a completely tokenizable input are completely
tokenizable. The three observations concern tokens of kind $D$ only and still hold. If some word $y$ is accepted by
every $A_i$, the wrapper $\code{[}\,\code{\#}\,\code{x}^m y\,\code{\symbol{36}}$ is completely tokenizable, and the scan
of its right chunk at the cut $c = 1$ proceeds as before until it stands at the final $\code{\symbol{36}}$, where it
halts because no token begins with that symbol. The right chunk is not completely tokenizable, and $\code{\#}$ is not
safe. If no word is accepted by every $A_i$, let $w$ be completely tokenizable with $w_c = \code{\#}$ and $0 < c < |w|$,
and argue as in the two-kind case. When $c$ is a token boundary of $w$, the scan of $u$ consumes $u$ and the scan of $v$
proceeds as the scan of $w$ does from $c$, which consumes the rest of $w$. When $c$ lies inside a wrapper, the scan of
$u$ consumes $u$ as before, and the scan of $v$ emits $\code{\#}$, single symbols $\code{x}$ and a tester, restarts in
$q_0$ at the start of $r$, and consumes $r$ as the scan of $w$ does. Both chunks are completely tokenizable in either
case, and $\code{\#}$ is safe.

It remains to code the six symbols over the three symbols $\code{0}$, $\code{1}$ and $\code{\#}$. Let the map $\varphi$
keep $\code{\#}$, replace $\code{\symbol{36}}$ by $\code{1}$, and replace the symbols $\code{0}$, $\code{1}$, $\code{x}$
and $\code{[}$ by the \emph{code words} $\code{000}$, $\code{001}$, $\code{010}$ and $\code{011}$ in this order. None of
the six images is a prefix of another, so for words $t$ and $z$ over the six symbols $\varphi(t)$ is a prefix of
$\varphi(z)$ only when $t$ is a prefix of $z$. The \emph{coded automaton} of either automaton above has the tokens
$\varphi(t)$ for the tokens $t$ of that automaton, each with the kind of $t$, and one further token of kind $D$, the
single symbol $\code{0}$. It is obtained by replacing every transition on $\code{0}$, $\code{1}$, $\code{x}$ or
$\code{[}$ with a path of three transitions that spells the code word, where the paths out of one state share their
common prefixes, by labeling every transition on $\code{\symbol{36}}$ with $\code{1}$ instead, and by letting the state
that $\code{0}$ leads to from $q_0$ accept with kind $D$. The result is deterministic, and with two kinds it has
$3(|Q_1| + \cdots + |Q_m|) + 7m + 10$ states.

The coded scanner follows the scanner over six symbols in this sense. Let the remaining input be $\varphi(z)\,r$, where
$z$ is a nonempty word over the six symbols that ends with $\code{\symbol{36}}$ and $r$ is any word over the three. If
$\varphi(t)$ is a prefix of $\varphi(z)\,r$ for a token $t$, then $t$ is a prefix of $z$. A $\varphi(t)$ no longer than
$\varphi(z)$ is a prefix of $\varphi(z)$, and the property of $\varphi$ applies. A longer $\varphi(t)$ would have
$\varphi(z)$ as a proper prefix, so that $z$ would be a proper prefix of $t$, which is impossible because no token
continues past a $\code{\symbol{36}}$. Conversely the image of every token that is a prefix of $z$ is a prefix of
$\varphi(z)\,r$, and of two such tokens the longer has the longer image. The single symbol $\code{0}$ is a prefix of
$\varphi(z)\,r$ only when $z$ begins with $\code{0}$, $\code{1}$, $\code{x}$ or $\code{[}$, and that symbol is then a
token whose code word is longer. Hence the coded scanner emits $\varphi(t)$ for the longest token $t$ that is a prefix
of $z$, with the kind of $t$, and halts when no token is a prefix of $z$. After such an emission, the remaining input is
$\varphi(z')\,r$ with $z = tz'$, where $z'$ is empty or ends with $\code{\symbol{36}}$. By induction on the number of
tokens, the coded scan of $\varphi(z)\,r$ begins with the images under $\varphi$ of the tokens that the scan of $z$
emits, in the same order and with the same kinds. If the scan of $z$ halts at an offset $i < |z|$, the coded scan halts
at the offset $|\varphi(z_{[0,i)})|$, inside $\varphi(z)$, and emits nothing more. Otherwise these images consume
$\varphi(z)$, and the coded scan restarts in $q_0$ at the start of $r$, from where it proceeds as the scan of $r$.

Consider the coded automaton with two kinds. Each of its three symbols is a token, so every input over them is
completely tokenizable, and the tokens containing $\code{\#}$ are the single symbol and the coded wrappers, which have
it as their fourth symbol. Suppose some word $y$ is accepted by every $A_i$, and let
$w = \varphi(\code{[}\,\code{\#}\,\code{x}^m y\,\code{\symbol{36}})$, which is one discarded token with $\code{\#}$ at
the cut $c = 3$. The left chunk $\code{011}$ is one discarded token. The right chunk is
$\varphi(\code{\#}\,\code{x}^m y\,\code{\symbol{36}})$, so its scan emits the images of the tokens found above for
$\code{\#}\,\code{x}^m y\,\code{\symbol{36}}$, the last of which is $\code{1}$ of kind $K$. The chunks have the kept
stream $(K, 1)$ against the empty one, and $\code{\#}$ is not safe. Suppose no word is accepted by every $A_i$, and let
$w$ be any input over the three symbols with $w_c = \code{\#}$ and $0 < c < |w|$. If $c$ is a token boundary of $w$ the
kept streams agree as before. Otherwise $c$ lies strictly inside a coded wrapper that begins at $c - 3$, so that
$w_{[c-3,e)} = \varphi(\code{[}\,\code{\#}\,\code{x}^m y\,\code{\symbol{36}})$ for some $y$ and some $e \leq |w|$. By
Lemma~\ref{lem:truncation} the scan of $u$ emits the tokens of $w$ on $[0, c - 3)$ and reaches $c - 3$ in $q_0$, where
$\code{011}$ remains. The tokens that are prefixes of it are $\code{0}$ and $\code{011}$, and the scan emits the longer,
which is discarded. The right chunk is $\varphi(\code{\#}\,\code{x}^m y\,\code{\symbol{36}})\,r$ with $r = w_{[e,|w|)}$.
Its scan emits the images of $\code{\#}$, of single symbols $\code{x}$ and of a tester, all discarded, and restarts in
$q_0$ at the start of $r$, where the scan of $w$ restarts after its wrapper. The kept streams agree, and $\code{\#}$ is
safe.

In the coded automaton with the single kind the symbol $\code{1}$ is no token. If some word $y$ is accepted by every
$A_i$, the coded wrapper $\varphi(\code{[}\,\code{\#}\,\code{x}^m y\,\code{\symbol{36}})$ is completely tokenizable, and
the scan of its right chunk at the cut $c = 3$ halts at the final $\code{1}$, as the scan of
$\code{\#}\,\code{x}^m y\,\code{\symbol{36}}$ halts at its last symbol. The right chunk is not completely tokenizable,
and $\code{\#}$ is not safe. If no word is accepted by every $A_i$, let $w$ be completely tokenizable with
$w_c = \code{\#}$ and $0 < c < |w|$. The two cases just given apply again. When $c$ is a token boundary of $w$, the scan
of $u$ consumes $u$ and the scan of $v$ proceeds as the scan of $w$ does from $c$. When $c$ lies inside a coded wrapper,
the scan of $u$ ends with the token $\code{011}$, and the scan of $v$ restarts in $q_0$ at the start of $r$ and consumes
$r$ as the scan of $w$ does. Both chunks are completely tokenizable in either case, and $\code{\#}$ is safe.

In both coded automata $\code{\#}$ is therefore safe modulo $I$ exactly when no word is accepted by every $A_i$, and
both are computed from the $A_i$ in polynomial time. Since PSPACE is closed under complement, that question is
PSPACE-hard along with intersection non-emptiness. This proves hardness in both restricted forms, and with it the
theorem.
\end{proof}

The automata $S$ and $F$ are exponential in $n$ through their obligation sets. For both an exponential is unavoidable,
whatever nondeterministic automaton accepts the complete scans or the failed ones.

\begin{proposition}[Complete and failed scans need exponentially many states]
\label{prop:exponential}
For every $k \geq 1$ there is a token automaton with $k + 2$ states over three symbols such that the smallest
nondeterministic automaton accepting its completely tokenizable words has exactly $2^k + 2$ states, and so has the
smallest nondeterministic automaton accepting the words it does not tokenize completely. In both cases a deterministic
automaton attains the count.
\end{proposition}

\begin{proof}
The symbols are $\code{a}$, $\code{b}$ and $\code{!}$, the two \emph{letters} being $\code{a}$ and $\code{b}$, and there
is one kind. The automaton has the initial state $q_0$, a chain $s_1, \ldots, s_k$ and a state $h$. From $q_0$ the
letter $\code{b}$ leads back to $q_0$, the letter $\code{a}$ leads to $s_1$, and $\code{!}$ leads to $h$. Either letter
leads from $s_j$ to $s_{j+1}$ for $j < k$, and $\code{b}$ leads from $s_k$ to $h$, which loops on both letters and
returns to $q_0$ on $\code{!}$. The accepting states are $q_0$ and $s_1$. These are $k + 2$ states, all of them live.
The initial state accepts and is re-entrant, which Section~\ref{sec:prelim} permits, since the scanner emits tokens of
positive length only. The tokens are the nonempty words that lead from $q_0$ to $q_0$ or to $s_1$. Each letter is a
token, so the scan never halts inside a word of letters, and every word of letters is completely tokenizable.

For the upper bounds, call words $x$ and $x'$ \emph{interchangeable} if, for every word $t$, the word $xt$ is completely
tokenizable exactly when $x't$ is. Three facts are used. First, a nonempty word $x$ that leads from $q_0$ to $q_0$ is
interchangeable with the empty word. The scan of $xt$ reads $x$ and stands in $q_0$, which accepts, at the offset $|x|$,
and from there the automaton reads $t$ as it does in a scan of $t$. If $y$ is the first token of $t$, then $xy$ is the
first token of $xt$, and both scans restart on the same suffix of $t$. Otherwise the first token of $xt$ is $x$, and the
scan restarts in $q_0$ at the start of $t$. In both cases the scan of $xt$ continues as the scan of $t$ does. Second,
let $x$ be a word of at most $k - 1$ letters followed by $\code{!}$, or a word of $k - 1$ letters followed by
$\code{a}$. Then $\code{a}\,x$ is interchangeable with $x$. After the $\code{a}$ the automaton stands in $s_1$. It reads
all but the last symbol of $x$ through further states of the chain, none of which accepts, and has no transition on the
last symbol, so the first token of $\code{a}\,xt$ is the single $\code{a}$. Third, let $x$ and $x'$ lead from $q_0$ to
$h$, and suppose that $xy$ and $x'y$ are completely tokenizable for every word $y$ of letters, or that neither is for
any. Then $x$ and $x'$ are interchangeable. A word $t$ that is not a word of letters has the form $y\,\code{!}\,t'$ with
$y$ a word of letters, and $xy\,\code{!}$ and $x'y\,\code{!}$ lead to $q_0$, so by the first fact $xy\,\code{!}\,t'$ and
$x'y\,\code{!}\,t'$ are completely tokenizable exactly when $t'$ is.

Let $M$ be the deterministic automaton whose states are the $2^k$ words of $k$ letters, the \emph{histories}, and two
further states $\mathrm{held}$ and $\mathrm{pending}$. The initial state is $\code{b}^k$. On a letter $c$ the history
$x_1 x_2 \cdots x_k$ moves to $\mathrm{held}$ if $x_1 = \code{a}$ and $c = \code{b}$, and to $x_2 \cdots x_k c$
otherwise, and on $\code{!}$ every history moves to $\mathrm{pending}$. Both $\mathrm{held}$ and $\mathrm{pending}$ loop
on the letters and move to $\code{b}^k$ on $\code{!}$. Every state but $\mathrm{pending}$ is final. As
\emph{representative} take a history for itself, $\code{a}^k\code{b}$ for $\mathrm{held}$ and $\code{!}$ for
$\mathrm{pending}$. We show that every word $w$ is interchangeable with the representative of the state that $M$ reaches
on $w$. For the empty word that is the first fact, since $\code{b}^k$ leads from $q_0$ to $q_0$. Interchangeable words
stay interchangeable when the same symbol is appended to both, so it suffices to show, for each state with
representative $r$ and each symbol $c$, that $rc$ is interchangeable with the representative of the successor. Let
$r = x_1 y$ be a history and $c$ a letter. If $x_1 = \code{b}$, the first fact makes $\code{b}\,yc$ interchangeable with
$yc$. If $x_1 = \code{a}$ and $c = \code{a}$, the second fact does. If $x_1 = \code{a}$ and $c = \code{b}$, both
$\code{a}\,y\,\code{b}$ and $\code{a}^k\code{b}$ lead to $h$ and are words of letters, and the third fact applies. For a
history $r$ and the symbol $\code{!}$, remove the letters of $r\,\code{!}$ one at a time, a leading $\code{b}$ by the
first fact and a leading $\code{a}$ by the second, which applies because fewer than $k$ letters stand between it and the
$\code{!}$. Only $\code{!}$ remains. For $\mathrm{held}$ and a letter $c$, the words $\code{a}^k\code{b}\,c$ and
$\code{a}^k\code{b}$ lead to $h$ and are words of letters. For $\mathrm{pending}$ and a letter $c$, the words
$\code{!}\,c$ and $\code{!}$ lead to $h$, and no word $\code{!}\,y$ with $y$ a word of letters is completely
tokenizable, because the automaton reads it without reaching an accepting state. The third fact applies to both. Finally
$\code{a}^k\code{b}\,\code{!}$ and $\code{!!}$ lead to $q_0$ and are interchangeable with the empty word, and hence with
$\code{b}^k$. Taking $t$ empty, a word $w$ is completely tokenizable exactly when the representative of its state is,
which holds for the histories and for $\code{a}^k\code{b}$, words of letters, and fails for $\code{!}$. Hence $M$
accepts the completely tokenizable words, and with $\mathrm{pending}$ as its only final state it accepts the others.
Both automata have $2^k + 2$ states.

Let $u = u_1 \cdots u_k$ and $v = v_1 \cdots v_k$ be words of $k$ letters each, and call the pair $(u, v)$
\emph{blocked} if $u_j = \code{a}$ and $v_j = \code{b}$ for some $j$. On $u$ the automaton $M$ moves from $\code{b}^k$
to the history $u$, since the first letter of its history is a $\code{b}$ of the initial state until $k$ letters have
been read. While $M$ then reads $v$ from a history, that history is $u_j \cdots u_k v_1 \cdots v_{j-1}$ before the
letter $v_j$, so $M$ moves to $\mathrm{held}$ at the first $j$ with $u_j = \code{a}$ and $v_j = \code{b}$ and stays
there, and reaches the history $v$ if there is no such $j$. From $\mathrm{held}$ a first $\code{!}$ leads to
$\code{b}^k$ and a second to $\mathrm{pending}$, while from a history a first leads to $\mathrm{pending}$ and a second
to $\code{b}^k$. Hence $uv\,\code{!}$ is completely tokenizable exactly when $(u, v)$ is blocked, and $uv\,\code{!!}$
exactly when it is not.

For the lower bounds, let $N$ be a nondeterministic automaton accepting the completely tokenizable words, and consider
the $2^k + 2$ pairs $(u, u\,\code{!!})$ for the words $u$ of $k$ letters, $(\code{!}, \code{!})$ and
$(\code{a}^k\code{b}, \varepsilon)$, where $\varepsilon$ is the empty word. For each of these pairs $(x, y)$ the word
$xy$ is completely tokenizable. No pair $(u, u)$ is blocked, the word $\code{!!}$ is a token, and every word of letters
is completely tokenizable. Fix an accepting run of $N$ on each $xy$, and call its state after reading $x$ the
\emph{midpoint} of the pair. Suppose two of the pairs, $(x, y)$ and $(x', y')$, had the same midpoint. The part of the
first run that reads $x$, followed by the part of the second that reads $y'$, would be an accepting run of $N$ on $xy'$,
and with the pairs exchanged the same splice would be one on $x'y$, so both words would be completely tokenizable. No
two of the pairs allow that. Words $u \neq v$ of $k$ letters differ at some index $j$, where one of them has $\code{a}$
and the other $\code{b}$, so $(u, v)$ or $(v, u)$ is blocked and one of $uv\,\code{!!}$ and $vu\,\code{!!}$ is not
completely tokenizable. Both $\code{!}$ and $\code{a}^k\code{b}$ lead from $q_0$ to $h$. On $\code{!}\,u\,\code{!!}$ and
on $\code{a}^k\code{b}\,u\,\code{!!}$ the token automaton therefore loops in $h$ on $u$, returns to $q_0$ on the next
symbol and moves to $h$ on the last, where the input ends, so the scan emits a token and halts at the last $\code{!}$,
which is no token. Finally the word $\code{!}$, the first component of one of the two further pairs followed by the
second of the other, is not completely tokenizable. The $2^k + 2$ midpoints are therefore distinct. For an automaton
accepting the words that are not completely tokenizable, argue in the same way from the pairs $(u, u\,\code{!})$,
$(\code{!}, \varepsilon)$ and $(\code{a}^k\code{b}, \code{!!})$. Each $xy$ is not completely tokenizable, because no
pair $(u, u)$ is blocked, $\code{!}$ is no token, and the scan of $\code{a}^k\code{b}\,\code{!!}$ halts at its last
symbol as above. A common midpoint would give two further such words, but one of $uv\,\code{!}$ and $vu\,\code{!}$ is
completely tokenizable, the words $\code{!}\,u\,\code{!}$ and $\code{a}^k\code{b}\,u\,\code{!}$ lead from $q_0$ to $q_0$
and are tokens, and $\code{a}^k\code{b}$ is a word of letters. The two lower bounds equal the number of states of $M$.
\end{proof}

A token automaton with $n = k + 2$ live states can thus require exactly $2^{n-2} + 2$ states of an automaton for its
complete scans and of an automaton for its failed ones, deterministic or not. The count is exact for this family only,
and over five symbols a token automaton with $n$ live states requires more.

\begin{proposition}[A larger lower bound over five symbols]
\label{prop:five}
For every $n \geq 3$ there is a token automaton with $n$ states, all of them live, over five symbols such that every
nondeterministic automaton accepting its completely tokenizable words has at least $2^n$ states, and every
nondeterministic automaton accepting the words it does not tokenize completely has at least $3 \cdot 2^{n-1} - 3$.
\end{proposition}

\begin{proof}
Let $N = n - 1$. The states are $g$ and $p_0, \ldots, p_{N-1}$, the indices of the latter read modulo $N$, and $R$
denotes the set of the states $p_j$. The initial state is $q_0 = p_0$, the only accepting state is $g$, and there is one
kind. The symbols are $\code{r}$, $\code{c}$, $\code{d}$, $\code{m}$ and $\code{!}$. The symbol $\code{r}$ leads $p_j$
to $p_{j+1}$ and $g$ to $g$. The symbol $\code{c}$ leads $p_0$ to $g$, $g$ to $p_1$ and $p_1$ to $p_0$, and $\code{d}$
leads each of the three back, $g$ to $p_0$, $p_0$ to $p_1$ and $p_1$ to $g$. The symbol $\code{m}$ leads $p_0$ and $g$
to $g$, and $\code{!}$ leads $g$ to $p_0$ and has no transition from $p_0$. Each of $\code{c}$, $\code{d}$, $\code{m}$
and $\code{!}$ leads every state not named for it to itself. All $n$ states are live, since $\code{r}^j$ leads from
$p_0$ to $p_j$ and $\code{r}^{N-j}\code{c}$ from $p_j$ to $g$. The symbols $\code{r}$, $\code{c}$ and $\code{d}$ permute
the states.

For a word $w$ and an offset $s \leq |w|$, the \emph{path from $s$} is the sequence of states the automaton takes from
$q_0$ on $w_{[s, |w|)}$ for as long as it has a transition, and the path \emph{stands} in a state at the offset at which
it reaches that state. A scan that stands in $q_0$ at $s$ emits the token ending at the last offset beyond $s$ at which
the path from $s$ stands in $g$, and halts at $s$ if there is none. The path of a token is the path from its start. For
a completely tokenizable word $x$ that is not empty, the \emph{obligations} of $x$ are the states in which the paths of
its tokens other than the last stand at $|x|$, as in Lemma~\ref{lem:complete}. None of them is $g$, since the scan would
have extended that token to $|x|$, and the path of the last token stands in $g$ at $|x|$. Four facts are used.

First, let the scan of a word $x$ emit tokens that start at $s_0 < \cdots < s_{k-1}$, and let $z$ be any word. Within
$x$ the path of each token stands in $g$ for the last time at the token's end. The scan of $xz$ therefore emits the same
tokens in order as far as the first one whose path stands in $g$ at an offset beyond $|x|$, and that one it extends to
the last such offset. If there is no such token, the scan of $xz$ emits all $k$ tokens and stands in $q_0$ where the
scan of $x$ ended or halted. In either case a token of $xz$ that starts before $|x|$ starts where a token of $x$ starts
or where the scan of $x$ halted.

Second, call a nonempty word $y$ over $\code{r}$, $\code{c}$ and $\code{d}$ \emph{returning} if the automaton, started
in $g$, ends in $g$ on $y$ and stands in $p_0$ after no proper prefix of $y$. Let $y$ be returning and $x$ any word such
that $xy$ is completely tokenizable. Then its last token starts at an offset $s < |x|$, and the path from $s$ stands in
$g$ at $|x|$. For that path stands in $g$ at $|xy|$, and so does the run from $g$ on $y$ that starts at $|x|$. The three
symbols permute the states, so two runs on the same symbols that end in the same state stand in the same state at every
earlier offset they share. If $s \geq |x|$, the run from $g$ would stand in $p_0$ at $s$, after a proper prefix of $y$.
Hence $s < |x|$, and at $|x|$ the path stands where the run starts. By the first fact $s$ is the start of a token of $x$
or the offset at which the scan of $x$ halted, so $xy$ is not completely tokenizable when none of the paths from those
offsets stands in $g$ at $|x|$.

Third, let $x$ be completely tokenizable and not empty and let $y$ be returning. Then $xy$ is completely tokenizable
exactly when no prefix of $y$ leads an obligation of $x$ to $g$. If none does, the first fact lets the scan of $xy$ emit
the tokens of $x$ up to the last. The path of the last one stands in $g$ at $|x|$ and follows the run from $g$ on $y$,
which ends in $g$, so that token is extended to $|xy|$. Otherwise the scan of $xy$ extends a token of $x$ other than the
last to an offset beyond $|x|$, and every later token starts beyond $|x|$. If $xy$ were completely tokenizable, its last
token would start before $|x|$ by the second fact, so it would be the extended token, and its path would stand in $g$ at
$|x|$. That path stands in an obligation there.

Fourth, let $v$ be a nonempty word over $\code{r}$ and $\code{m}$. Both symbols lead $g$ to $g$, so a run on $v$ that
reaches $g$ ends in $g$. If $x$ is completely tokenizable and not empty, then so is $xv$. By the first fact the scan of
$xv$ extends some token of $x$, the last one if no other, and the path of the extended token stands in $g$ at $|xv|$, so
the token ends there. Now let the scan of a word $u$ halt at an offset $h < |u|$, and call \emph{pending} the states in
which the paths of its tokens and the path from $h$ stand at $|u|$. Then $uv$ is completely tokenizable exactly when $v$
leads some pending state to $g$. If it leads none there, none of these paths stands in $g$ beyond $|u|$, and by the
first fact the scan of $uv$ emits the tokens of $u$ and halts at $h$. Otherwise the scan of $uv$ extends a token of $u$,
or emits them all and then a token from $h$, and in both cases the path of that token stands in $g$ at $|uv|$, where the
token ends.

Three kinds of words are built from the symbols. For every list of distinct states of $R$ there is a completely
tokenizable word over $\code{c}$ and $\code{r}$, not empty, whose obligations are the states of the list in the order of
their tokens. To see this, let $x$ be completely tokenizable and not empty. The symbol $\code{r}$ leads no state of $R$
to $g$ and keeps the path of the last token in $g$, so $x\,\code{r}$ is completely tokenizable with the obligations that
$\code{r}$ leads those of $x$ to, in the same order. If $p_0$ is no obligation of $x$, the symbol $\code{c}$ leads no
obligation to $g$ and leads the path of the last token to $p_1$, so that token ends at $|x|$, and $\code{c}$ is a token.
Hence $x\,\code{c}$ is completely tokenizable, and its obligations are the states that $\code{c}$ leads those of $x$ to,
in the same order, and then $p_1$. For the empty list take $\code{c}$, which has no obligation. For any other list let
$p_{j+1}$ be its last state, where $0 \leq j < N$, and let $\ell$ be the list that $\code{r}^j$ leads onto it, which
ends with $p_1$. The symbol $\code{c}$ leads $p_1$ to $p_0$ and leaves the other states of $R \setminus \{p_0\}$, so the
states of $\ell$ before the last are those that $\code{c}$ leads the states of a list without $p_0$ to. By induction on
the length there is a word for that list, and appending $\code{c}$ and then $\code{r}^j$ gives a word for the given
list. For $O \subseteq R$ let $x_O$ be such a word for the members of $O$ in the order of their indices.

For $Y \subseteq R$ let $y_Y$ be the word that has, for $j = 0, \ldots, N - 1$ in turn, the two symbols
$\code{c}\,\code{d}$ if $p_{-j} \in Y$, and then $\code{r}$ in either case. Let $v_Y$ be the same word with $\code{m}$
in place of $\code{c}\,\code{d}$. The word $\code{c}\,\code{d}$ leads every state back to itself, by way of $g$ from
$p_0$ alone, by way of $p_0$ from $p_1$ alone, and by way of $p_1$ from $g$. A run from $p_i$ on $y_Y$ therefore stands
in $p_{i+j}$ before turn $j$, which is $p_0$ exactly when $p_{-j} = p_i$, and it visits $g$ exactly when $p_i \in Y$.
The run from $g$ ends in $g$ and never stands in $p_0$, and $y_Y$ is returning. By the third fact, if $x$ is completely
tokenizable and not empty, $x\,y_Y$ is completely tokenizable exactly when no obligation of $x$ is in $Y$. In the same
way $v_Y$ leads $p_i$ to $g$ when $p_i \in Y$ and back to $p_i$ otherwise, because $\code{m}$ leads no state of $R$ but
$p_0$ to $g$.

For every $P \subseteq R$ there is a word $u_P$ whose scan halts before $|u_P|$ with the set of pending states $P$. For
the empty set take the single symbol $\code{!}$, which has no transition from $p_0$. Next let $O \subseteq R$ with
$p_0 \notin O$, and consider $x_O\,\code{!}$. The symbol $\code{!}$ leads no state to $g$, leaves the obligations, and
leads the path of the last token to $p_0$, so the scan emits the tokens of $x_O$ and halts at $|x_O|$, where it has no
transition. The pending states are $p_0$ and the members of $O$. Appending $\code{r}^j$ changes neither the tokens nor
the halt and leads the pending states on by $j$. For a nonempty $P$ let $p_j$ be its member of largest index, let
$P^\circ$ be the set that $\code{r}^j$ leads onto $P \setminus \{p_j\}$, which does not contain $p_0$, and take
$u_P = x_{P^\circ}\,\code{!}\,\code{r}^j$. By the fourth fact, $u_P\,v_Y$ is completely tokenizable exactly when $P$ and
$Y$ have a common member.

For the complete scans consider the pairs $(x_O, y_{R \setminus O})$ for $O \subseteq R$, the pairs
$(u_P, \code{r}^{N-j}\code{m}\,y_{R \setminus P^\circ})$ for the nonempty $P \subseteq R$, with $j$ and $P^\circ$ as in
$u_P$, and the pair $(\varepsilon, \varepsilon)$. These are $2^{n-1} + (2^{n-1} - 1) + 1 = 2^n$ pairs. For each pair
$(x, z)$ the word $xz$ is completely tokenizable. In the first group that is the criterion for $x\,y_Y$. In the second,
$\code{r}^{N-j}$ leads the pending states of $u_P$ back to $p_0$ and the members of $P^\circ$, and $\code{m}$ then leads
only $p_0$ to $g$. The path in $p_0$ is that of the last token of $x_{P^\circ}$, which the scan extends to the end by
the first fact, so $u_P\,\code{r}^{N-j}\code{m}$ is completely tokenizable with the obligations $P^\circ$, and the
criterion applies to it. The empty word is completely tokenizable.

As in the proof of Proposition~\ref{prop:exponential}, an automaton accepting the completely tokenizable words has a
midpoint for each pair, and the midpoints of two pairs $(x, z)$ and $(x', z')$ differ if one of $xz'$ and $x'z$ is not
completely tokenizable. For two pairs of the first group, with sets $O \neq O'$, one set has a member outside the other,
say $O$, and then $x_O\,y_{R \setminus O'}$ is not completely tokenizable. The same holds for two pairs of the second
group with the same $j$, whose sets $P^\circ$ differ, with $u_P\,\code{r}^{N-j}\code{m}$ in place of $x_O$. For two
pairs of the second group whose sets $P$ and $P'$ have the members of largest index $p_j$ and $p_k$ with $j < k$, no
word $u_P\,\code{r}^{N-k}\code{m}\,y_Y$ is completely tokenizable. No member of $P$ has an index above $j$, so
$\code{r}^{N-k}$ leads no pending state of $u_P$ to $p_0$ and $\code{m}$ leads none to $g$. By the first fact the scan
of $u_P\,\code{r}^{N-k}\code{m}$ emits the tokens of $u_P$ and halts where the scan of $u_P$ halts, and none of the
paths from those offsets stands in $g$ at the end, as the second fact requires. Nor is any word $u_P\,y_Y$ completely
tokenizable, since the pending states of $u_P$ lie in $R$, or any word $y_Y$, by the second fact with $x$ empty, which
separates the second group and the pair $(\varepsilon, \varepsilon)$ from the first. No $u_P$ is completely tokenizable,
which separates $(\varepsilon, \varepsilon)$ from the second group. All the midpoints are therefore distinct, and there
are $2^n$ of them.

For the failed scans a first group holds the pairs $(u_P, v_{R \setminus P})$ for $P \subseteq R$ and a second the pairs
$(x_O, z_O)$ for $O \subsetneq R$, where $z_O = \code{r}^{N-i}\,\code{!}\,\code{r}^i\,v_{R \setminus O}$ for a nonempty
$O$ and a member $p_i$ of it, and
$z_\emptyset = \code{c}\,\code{m}\,\code{d}\,\code{r}^{N-1}\,\code{!}\,v_{R \setminus \{p_0\}}$. If $x$ is completely
tokenizable and not empty, then $x\,z_O$ is completely tokenizable exactly when some obligation of $x$ is outside $O$.
For a nonempty $O$ the scan of $x\,\code{r}^{N-i}\,\code{!}$ extends the last token of $x$ over the symbols $\code{r}$
and halts at the $\code{!}$, which leads the path of that token to $p_0$ and ends that of an obligation standing there,
so after $\code{r}^i$ the pending states are $p_i$ and the obligations of $x$, and the fourth fact applies. As
$\code{c}\,\code{m}\,\code{d}$ leads $p_1$ to $p_0$ and every other state back to itself, $z_\emptyset$ leads every
state of $R$ to $g$, and the run from $g$ last visits $g$ just before reading the $\code{!}$. By the first fact the scan
of $x\,z_\emptyset$ extends the token of the first obligation of $x$ to the end, or if there is none the last token to
the $\code{!}$, where the scan halts.

For the third, let $J \subseteq R$ be neither empty nor $R$, let $b_J = p_i$ be its member of least index for which
$p_{i-1}$ is outside $J$, and let $x^J$ be a word for the list of the members of $J$ other than $b_J$ as they occur in
$p_{i+1}, p_{i+2}, \ldots, p_{i-1}$, then $b_J$, and then the states outside $J$ as they occur in the same sequence. The
third group holds the pairs $(x^J, v_{R \setminus J}\,y_{\{b_J\}})$. If $x$ is completely tokenizable and not empty and
$b \in J$, the word $x\,v_{R \setminus J}\,y_{\{b\}}$ is completely tokenizable exactly when $b$ is no obligation of $x$
or an obligation outside $J$ precedes it. By the first fact the scan of $x\,v_{R \setminus J}$ extends the token of the
first obligation outside $J$ to the end, or the last token if there is none. The obligations before it lie in $J$ and
are back in their states, and the criterion for $x\,y_Y$ applies.

These are $2^{n-1} + (2^{n-1} - 1) + (2^{n-1} - 2) = 3 \cdot 2^{n-1} - 3$ pairs. For each pair $(x, z)$ the word $xz$ is
not completely tokenizable, in the first group by the criterion for $u_P\,v_Y$, in the second because every obligation
of $x_O$ is in $O$, and in the third because $b_J$ precedes every state outside $J$ among the obligations of $x^J$.
Arguing as before, two pairs have distinct midpoints if one of $xz'$ and $x'z$ is completely tokenizable. For two pairs
of the first group, say with a member of $P$ outside $P'$, that is $u_P\,v_{R \setminus P'}$, and for two of the second,
say with a member of $O$ outside $O'$, it is $x_O\,z_{O'}$. A pair of the first group and a pair of the other two give a
word $x\,v_Y$ with $x$ completely tokenizable and not empty, which is completely tokenizable by the fourth fact. A pair
of the second group and a pair of the third give $x^J z_O$, which is completely tokenizable because every state of $R$
is an obligation of $x^J$. Take two pairs of the third group, with sets $J \neq J'$, $b = b_J$ and $b' = b_{J'}$. It
suffices that a state outside $J$ precedes $b$ among the obligations of $x^{J'}$, or one outside $J'$ precedes $b'$
among those of $x^J$. If $b = b'$, a member of one set outside the other, say of $J$, precedes $b'$ in $x^J$. Otherwise
let $b = p_i$ and $a = p_{i-1}$, which is outside $J$. If $b \notin J'$, every member of $J'$ precedes $b$ in $x^{J'}$,
so either one outside $J$ does, or $J' \subseteq J$ and $a$, outside $J'$ and other than $b'$, precedes $b$ there. The
case $b' \notin J$ is the same with the sets exchanged. In the remaining case we may assume that $b$ has the smaller
index, and then $a$ is in $J'$ by the choice of $b'$, is not $b'$ since $b' \in J$, and precedes $b$ in $x^{J'}$. Hence
the $3 \cdot 2^{n-1} - 3$ midpoints are distinct.
\end{proof}

The two propositions do not settle the worst case over all token automata with $n$ live states, the largest number of
states that the smallest nondeterministic automaton for their complete scans, or for their failed ones, can need. The
automata $S$ and $F$ bound that number from above once the states no input reaches are deleted. Let $f \geq 1$ of the
$n$ live states accept and let $r = n - f$. A transition of $S$ exists only when no member of $O'$ accepts. In a
reachable state $(O, q)$ with $q$ not fresh, therefore, no obligation accepts, and a reachable fresh state other than
the initial one is $(O' \cup \{q'\}, \mathrm{fresh})$ with $q'$ its only accepting obligation. There are at most $n 2^r$
states of the first form and $f 2^r$ of the second, so $S$ has at most $(n + f)2^r + 1$ reachable states. A reachable
failure state $(O, q, \mathrm{fail})$ of $F$ likewise has no accepting obligation, and its $q$ is $\bot$ or one of the
$r$ live states that do not accept, so $F$ has at most $(r + 1)2^r$ further reachable states, which makes
$(2n + 1)2^r + 1$.

Smaller automata accept the same words. Keep the reachable states only, and consider a state $(O, q)$ of $S$ in which
$q$ is not fresh and $q \in O$. A transition from it on a byte $a$ has $q' = \delta^+(q, a)$ among the members of $O'$,
so $q'$ does not accept and the only successor is $(O', q')$, with $q' \in O'$ again. No run from such a state reaches a
fresh state, so no accepting run of $S$ visits one, and these states can be deleted. Of the states left that are not
fresh, at most $f 2^r$ have an accepting $q$, and at most $r 2^{r-1}$ have a $q$ that does not accept, since $O$ then
lies among the other $r - 1$ states that do not accept. Next, a fresh state other than the initial one is
$(O \cup \{q\}, \mathrm{fresh})$, where $q$ is its only accepting obligation and $O$ holds the others. It and $(O, q)$
have exactly the same predecessor-byte pairs, the pairs of a state and a byte from which a transition enters, because
every transition to a state $(O', q')$ with $q'$ accepting comes with one from the same state on the same byte to
$(O' \cup \{q'\}, \mathrm{fresh})$, and no other transition leads to a fresh state. The state $(O, q)$ was not deleted,
since $q$ accepts and no member of $O$ does. Merge the two into one state, final and with the transitions that leave
either. The merged state stands for a token that may have ended. On the next byte it either continues the token from $q$
under the obligations $O$, or commits it and starts the next token from $q_0$ under $O \cup \{q\}$. Every accepting run
of $S$ is one of the merged automaton, and an accepting run of the merged automaton lifts to one of $S$ from its last
state backward. Where the run stands in a merged state, take the fresh state if the word ends there and otherwise the
one of the two that has the transition the next step uses. Whichever is taken, the predecessor and the byte by which the
run entered the merged state enter that one as well. The merged automaton therefore accepts the completely tokenizable
words with at most $f 2^r + r 2^{r-1} + 1 = (n + f)2^{r-1} + 1$ states.

The same two changes apply to the states that $F$ shares with $S$. Failure is entered from fresh states only, so a
deleted state reaches no final state of $F$, and the two states of a merged pair still have the same predecessor-byte
pairs, neither being final in $F$. In a failure state $(O, q, \mathrm{fail})$ the state $q$ is followed as the
obligations are and must never accept. Let $P$ be $O$ when $q = \bot$ and $O \cup \{q\}$ otherwise. On a byte $a$ the
failure state has no successor when $\delta^+(p, a)$ accepts for some $p \in P$, and otherwise a single one, whose set
is the set of those $\delta^+(p, a)$ that are defined. Failure states with the same $P$ thus accept the same words from
there on and can be merged, which leaves at most $2^r$ of them, one for each set of states that do not accept. The
resulting automaton accepts the words the scanner does not tokenize completely with at most
$(n + f)2^{r-1} + 1 + 2^r = (n + f + 2)2^{r-1} + 1$ states. Both of these bounds are largest at $f = 1$. With
Proposition~\ref{prop:five}, for $n \geq 3$ the worst case therefore lies between $2^n$ and $(n + 1)2^{n-2} + 1$ for the
complete scans, and between $3 \cdot 2^{n-1} - 3$ and $(n + 3)2^{n-2} + 1$ for the failed ones, so the bounds are within
the factors $(n + 2)/4$ and $(n + 6)/6$ of each other.

On the family of Proposition~\ref{prop:exponential}, where $f = 2$ and $r = k$, these bounds put the smaller automaton
for the complete scans within a factor of $(n + 2)/2$ of the smallest one and that for the failed scans within
$(n + 4)/2$. A closer count gives the factor $3$ for the smaller automata, and two further steps that the family allows
bring both within the factor $3/2$. An obligation that enters $O$ as $q_0$ outlasts the next symbol only if that symbol
is $\code{!}$, which takes it to $h$, since either letter would lead it to an accepting state. One that enters as $s_1$
stands at $s_{i+1}$ after $i < k$ further letters, and after $k$ or more it stands at $h$ or is gone. A $\code{!}$ is
read only when no obligation stands at $h$, which it would return to $q_0$, and it removes every obligation in the
chain, so every obligation carried across a $\code{!}$ stands at $h$ after it. Let $(O, q)$ be a reachable state of $S$.
If $q = s_j$, the token being read entered $s_1$ on an $\code{a}$ and has read $j - 1$ letters since. Every obligation
was added before that $\code{a}$ and has read at least $j$ letters since. One that entered as $s_1$ therefore stands
beyond $s_j$ in the chain or at $h$, and one that entered as $q_0$ stands at $h$, so
$O \subseteq \{s_{j+1}, \ldots, s_k, h\}$, which allows $2^{k-j+1}$ sets. If $q = h$, the token came to $h$ on a
$\code{!}$, or through the chain on the last of $k + 1$ letters, and in both cases $O \subseteq \{h\}$, which further
letters preserve. Summed over $j$, and with at most $2^k$ sets for $q = q_0$, that makes at most
$(2^{k+1} - 2) + 2 + 2^k = 3 \cdot 2^k$ reachable states that are not fresh. A failure state of $F$ is entered from a
fresh state on a symbol that leads $q_0$ to no accepting state, which here is $\code{!}$ alone, so it is entered with
$q = h$ and $O \subseteq \{h\}$. Letters preserve both, and on a further $\code{!}$ there is no successor, since $h$
moves to $q_0$, which accepts, so $F$ has at most two reachable failure states. In the smaller automata every fresh
state other than the initial one is merged with a state that is not fresh, and of the states that are not fresh
$(\{h\}, h)$ is deleted, so the automaton for the complete scans has at most $3 \cdot 2^k - 1$ states beside the initial
one. The two failure states have the same set $\{h\}$. That makes at most $3 \cdot 2^k$ states for the complete scans
and $3 \cdot 2^k + 1$ for the failed ones, both less than $3\,(2^k + 2)$.

The family allows two further steps. An obligation at $h$ stays there on a letter and would move to $q_0$, which
accepts, on $\code{!}$. A state with $h$ among its obligations therefore reads no $\code{!}$, and its successors on a
letter have $h$ among theirs. No letter leads $h$ or a state of the chain to an accepting state, so on a letter no
obligation accepts. In the smaller automaton for the complete scans, let such a state have an accepting $q$. On either
letter it has a successor of the same kind. If $q = q_0$ the token continues to $q_0$ or to $s_1$, and if $q = s_1$ the
token ends there, $s_1$ joins the obligations, and the next token reads the letter from $q_0$. The state thus accepts
every word of letters and no other word, and all these states are merged into one. If $q$ does not accept, it is not
$h$, since $(\{h\}, h)$ is deleted, so it lies in the chain beyond $s_1$, from where the letters lead along the chain
and to $h$, so every successor is of this second kind, no final state is reached, and the state is deleted. In the
automaton for the failed scans every state shared with $S$ that has $h$ among its obligations is deleted, since such a
state reads no $\code{!}$ and failure is entered on $\code{!}$ only. In both automata the initial state and
$(\emptyset, q_0)$ have the same successors and are merged, both being final in the first automaton and neither in the
second. A token that ends at $q_0$ leaves an obligation that either letter leads to an accepting state and that
$\code{!}$ leads to $h$, where the next token then stands as well, and a failure entered there has the set $\{h\}$ as
one entered from the initial state has. Of the states left, at most $2^{k-1}$ have $q = q_0$ and at most $2^{k-j}$ have
$q = s_j$, their obligations lying among $s_{j+1}, \ldots, s_k$, and one has $q = h$. These are at most
$3 \cdot 2^{k-1}$ states, the initial state being $(\emptyset, q_0)$ among them. With the state merged from those that
accept the words of letters, the automaton for the complete scans has at most $3 \cdot 2^{k-1} + 1$ states, and with the
failure state in its place so has the automaton for the failed scans. Both counts are less than
$\tfrac{3}{2}\,(2^k + 2)$.

The two propositions concern automata for the scans and say nothing about a decision procedure of another shape. By
Theorem~\ref{thm:pspace}, a polynomial-time decision procedure for general token automata exists only if P equals
PSPACE. Both bounds are measured in the size of the compiled automaton. Where the tokens are listed, the obligations
collapse.

\begin{theorem}[Literal vocabularies]
\label{thm:literal}
Let the tokens be a finite list of nonempty words with kinds, equal words resolved by priority. Then whether $b$ is safe
modulo $I$ is decidable in time polynomial in the total length of the words, $|\Sigma|$ and $|T|$, with a replayable
counterexample when it is not.
\end{theorem}

\begin{proof}
Let $L$ be the length of the longest word. Let $P$ be the number of distinct prefixes of the words, the empty prefix
included; $P$ is at most the total length plus one. Define a scanner $B$ whose state is a buffer $x$ that is a prefix of
some word, initially empty, plus a failure state. $B$ is a delayed-output variant of the splitting companion's pending
automaton~\cite{nidhogg2026certifiedsplitting}, specialized to a listed vocabulary. On a byte $a$, if $xa$ is a prefix
of some word, the buffer becomes $xa$ and nothing is emitted. Otherwise $B$ drains the remainder, which starts as $xa$.
While the remainder is not a prefix of some word, $B$ fails if no word is a prefix of it, and otherwise emits the code
of the longest word that is a prefix of it and removes that word. The remainder left, a prefix of some word, becomes the
buffer. At end of input $B$ flushes the buffer: while the buffer is nonempty, $B$ emits and removes its longest word
prefix, failing if none exists.

$B$ emits the scanner's tokens with delay, and its buffer is the uncommitted suffix. Suppose, after reading $w_{[0,i)}$,
that the scanner's tokens on $[0, j)$ are the committed ones, the scanner restarts at $j$ in $q_0$ whatever follows, and
the buffer is $w_{[j,i)}$. That holds at $i = 0$. A byte $a$ that keeps $xa$ a prefix changes nothing but the buffer. A
byte that does not starts the drain at the remainder $xa$. While the remainder $y$ is not a prefix of a word, no word
extends past $y$. If no word is a prefix of $y$, no continuation of $y$ has an accepted prefix, since such a prefix
would either be a prefix of $y$ or have $y$ as a prefix. The scanner fails at $j$ on every continuation, and $B$ fails
for good. Otherwise the longest word prefix of $y$ is the scanner's token at $j$ on every continuation, $B$ emits that
token, and the remainder restores the invariant with $j$ moved past that token. The flush applies longest match to the
buffer alone, as the scanner does at end of input. Hence $B$ emits $\mathrm{code}(\pi_I(\mathrm{tok}(w)))$ and ends
outside failure exactly when $w$ is completely tokenizable.

$B$ has $P + 1$ states. A transition commits at most $L + 1$ bytes, so it emits at most $L + 1$ tokens, whose codes have
at most $2(L + 1)$ letters together. The same bounds hold for the flush. Run two copies on the marked inputs $M_b$ with
an explicit end marker: the whole copy ignores $\#$, the split copy flushes at $\#$ and restarts empty, and both flush
at the end. A failed whole copy is dropped, since no continuation makes the input completely tokenizable. A failed split
copy is kept as an absorbing state, so that a path reaching a successful whole flush with a failed split copy is a
chunk-completeness violation, found by reachability. On the remaining graph, of $O(P^2)$ vertices and $O(|\Sigma| P^2)$
edges with outputs of $O(L)$ letters, Lemma~\ref{lem:equality} decides output equality in $O(|\Sigma| P^4 L)$ letter
operations and constructs the counterexample. Both tests and the construction of $B$ are polynomial.
\end{proof}

Two token sets separate the decision from simpler tests in both directions: per-token testing admits a byte the decision
refuses, and both local tests refuse a byte the decision admits. \code{validation.cpp} asserts the outcome on both. The
first shows that testing each severed token on its own is not enough. Over discarded \code{a}, \code{ab}, \code{b},
\code{bc} and kept \code{c}, every token containing \code{b} can be cut there on its own. Yet \code{abc} scans as
\code{ab} then \code{c}, while the chunks \code{a} and \code{bc} are both discarded, so a kept token vanishes with both
chunks complete. The restarted scan consumed past the severed token's end, which the obligations of
Lemma~\ref{lem:complete} track and a per-token test cannot. The second shows that the observable is the kept kinds and
lengths, not their positions. Over discarded \code{a}, \code{ab}, \code{b}, \code{d} and kept \code{bc} and \code{cd} of
one kind, every cut before \code{b} is safe, although on \code{abcd} the kept token moves, from \code{cd} at offset two
to \code{bc} at offset one. Both local tests refuse \code{b} on their clause 2, since the scan restarted at \code{b} can
still reach the kept \code{bc} while the interrupted scan, past \code{ab}, can reach nothing, so the two states are
neither equal nor equivalent. The decision admits \code{b}.

\section{Deriving and using the certificate}
\label{sec:derivation}

The certificate is derived by a linear-time analysis of the compiled transition table at construction time. First a
reverse pass computes co-accessibility, inverting the table into a predecessor list, seeding a worklist with the
accepting states, and marking each state that reaches one. A forward pass from $q_0$ then marks the reachable states,
and intersecting the two marks yields $Q^+$, the trim subautomaton the results are stated over. Subset construction
emits only reachable states, but the analysis is exposed on a simulator that accepts any transition table, and an
unreachable state in a hand-assembled table would otherwise produce false negatives. One scan then detects whether any
reachable entry targets $q_0$, establishing re-entrancy, and a sweep over the $|\Sigma| \times |Q|$ table records, for
each byte, whether any non-exempt reachable live state consumes it into $Q^+$. The result is a $|\Sigma|$-bit set. Every
pass visits each entry of the logical $|\Sigma| \times |Q|$ table a constant number of times, so the cost stays
$O(|Q|\,|\Sigma|)$, with the same bound in auxiliary space for the predecessor list. Where the physical table is
class-compressed, a pass rereads a shared class row. The cost is paid once, at the same asymptotic order as building the
table, and the scanner's inner loop is untouched.

Three operations expose the property. A predicate reports whether a byte belongs to the set of useful certified split
symbols of Corollary~\ref{cor:useful}. A planner divides an input into chunks by sliding each interior boundary forward
from its equal-division target to the next certified byte (Figure~\ref{fig:split}), and returns a single chunk spanning
the whole input when that useful set is empty. An executor runs the plan, one thread per chunk. By
Theorem~\ref{thm:main} the concatenated per-chunk streams equal the serial stream for completely tokenizable input, and
Corollary~\ref{cor:malformed} governs input that is not. There is no speculation to retry, no overlap to verify, and no
merge beyond concatenation.

For $k$ requested chunks the planner performs at most $k-1$ forward searches, so its worst case is $O(kN)$ over an input
of $N$ bytes when certified occurrences are rare; $k$ is normally a small hardware-thread count, and the searches start
at equally spaced offsets. A one-pass planner would reduce the worst case to $O(N + k)$. Section~\ref{sec:planning}
measures the range.

\subsection{Deciding the relaxed condition}
\label{sec:decide}

Condition 3 needs no per-state set of tokens. Because $I$ is fixed when the automaton is compiled, $T(q) \subseteq I$ is
the complement of ``some kept token is still reachable from $q$'', and one reverse search settles that for every state
at once. The search is seeded with the accepting states whose token lies outside $I$ and closed backwards along
transitions, after which $T(q) \subseteq I$ exactly when $q$ was not reached. Two details keep that linear. The closure
walks a reverse index of the transition table, since rescanning every symbol and source state for each state it reaches
would be quadratic in the state count. Whether a state accepts a discarded token is resolved once per state, outside the
symbol loop. The index needs no space of its own, since determining co-accessibility already builds one and the relaxed
analysis reuses it. It is built once per class row, so its size stays of the order of the class-compressed table
actually stored. The closure then costs $O(|Q| + E)$, where $E$ counts the stored table edges, and the per-symbol test
that follows costs $O(|Q|\,|\Sigma|)$. The implementation holds the discarded set in an ordered set, so constructing it
costs $O(|I| \log(1 + |I|))$ and resolving one state's membership costs $O(\log(1 + |I|))$. The overall cost is
$O(|Q|\,|\Sigma| + (|Q| + |I|) \log(1 + |I|))$ in $O(|Q| + |I|)$ additional space, excluding the predecessor index that
the exact analysis already builds and this one reuses. Propagating the sets $T(q)$ themselves would instead need
$O(|Q|\,|T|)$ storage.

\begin{figure}[!ht]
\centering
\begin{tikzpicture}[x=4.6mm, y=1mm, font=\small]
  % One cell per input byte, so the token boxes and the input line up exactly.
  % i n t _ x = 4 2 ; \n i f ( x )
  %  0 1 2 3 4 5 6 7 8  9 10 11 12 13 14
  \node[anchor=east] at (-0.3, 13) {serial scan};
  \node[anchor=east] at (-0.3, 5) {input};

  % token extents, as outlined boxes
  \foreach \s/\e in {0/2, 3/3, 4/4, 5/5, 6/7, 8/8, 9/9, 10/11, 12/12, 13/13, 14/14} {
    \draw (\s, 9) rectangle (\e + 1, 17);
  }

  % the input bytes, one per cell
  \foreach \i/\c in {0/i, 1/n, 2/t, 4/x, 5/=, 6/4, 7/2, 8/;, 10/i, 11/f, 12/(, 13/x, 14/)} {
    \node at (\i + 0.5, 5) {\ttfamily \c};
  }

  % the certified byte, and the boundary immediately before it
  \draw[very thick] (9, 1) -- (9, 21);
  \node[anchor=south] at (9, 21.5) {certified byte: newline};
  % the two chunks
  \node[anchor=east] at (-0.3, -4) {chunk 0};
  \draw[|-|] (0, -4) -- (9, -4);
  \node[anchor=east] at (-0.3, -9) {chunk 1};
  \draw[|-|] (9, -9) -- (15, -9);
  \node[anchor=west] at (9.2, -13) {\scriptsize restarts in $q_0$};
\end{tikzpicture}
\caption{On completely tokenizable input, a certified byte is a point where the serial scan is provably between
tokens. No state reachable mid-token consumes it into a state that can still accept, so the scan has finished its last
token there and restarts in $q_0$. Failed lookahead may have read past the offset but recorded no acceptance beyond
it. A chunk may start there with no entry-state uncertainty, and the chunks' streams concatenate to the serial one.}
\label{fig:split}
\end{figure}
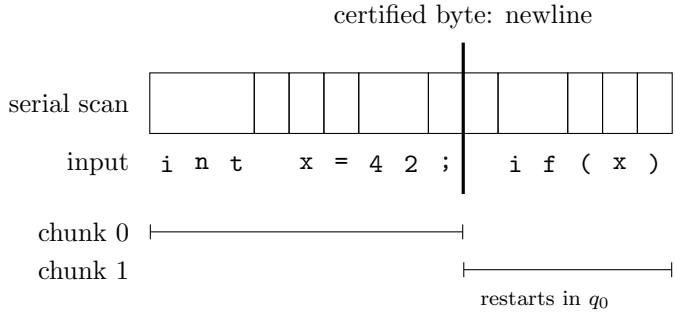

Testing one symbol then walks the states as the exact condition does, adding a membership test and one lookup of
$\delta^+(q_0, b)$ that is constant across the walk. The outcome packs into a second 256-bit map, held beside the exact
one, so the run-time predicate remains a single bit test and the search for a chunk boundary is unchanged. The
relaxation is free at scan time. At build time it adds work and replaces nothing: a backward closure over a reverse
index that is built either way, one membership test per state, and a second sweep of the same $|Q|\,|\Sigma|$ shape as
the exact one. That second sweep runs whether or not a caller ever asks the weaker question.

One consequence constrains the interface. The map depends on $I$, so a lexer must be told which tokens the caller
discards when it is built, instead of at each call. In practice the consuming tool fixes the discarded set when the
lexer is built, and the set does not vary with the input, so this is a mild restriction. The relaxed certificate is
still not a drop-in query on an existing lexer, and two consumers of one grammar that discard different tokens, a
compiler and a formatter say, hold different certificates.

The coalesced test of Section~\ref{sec:coalesced}, the rule of release 2.2.0 of munch, replaces the comparison of two
states by a comparison of two classes. The classes are computed by partition refinement over the completed automaton.
The refinement first partitions the states by observation, with every discarded kind one color, and then splits them by
the classes of their successors on every stored class row until nothing splits. Each round either splits a class or is
the last, so there are at most $|Q| + 1$ rounds. A round keys the states by their signatures in an ordered map, so it
costs $O(|Q|\,E' \log |Q|)$, with $E'$ the number of stored class rows, and the whole refinement is polynomial. The
clause on the state past the cut reads the same backward closure at the target of the transition. The implementation
computes the classes on the first question that needs them and never for a lexer that discards nothing, and the
per-symbol sweep then looks up two classes where it compared two states. The outcome is the same 256-bit map, and
nothing at scan time changes.

Both local conditions are sound and conservative. We tested both directions against an exhaustive oracle, and settled
what the oracle leaves open with the decision of Section~\ref{sec:exact}. Over 400 randomly generated token sets on a
three-symbol alphabet, every nonempty string up to length eight that the token set tokenizes completely was split at
every noninitial occurrence of every symbol, and the $\pi_I$ images of the spliced and serial streams were compared. The
initial cut is tautologically safe, and strings the serial scan leaves unconsumed are skipped before any comparison. The
corpus exercised 1125 pairs of symbol and token set. Definition~\ref{def:trivia-certified} admitted 265 of them and none
of those failed, and no symbol admitted by Definition~\ref{def:certified} was ever lost, which checks
Proposition~\ref{prop:extends} mechanically. A further 97 pairs had no counterexample through length eight yet were
rejected. Raising the bound from six to eight had reduced that count from 99 to 97, so the count is bound-sensitive, and
the exact decision shows by how much. Of the 97, 95 are safe at every length and 2 are refuted by a counterexample of
nine bytes, one byte beyond the oracle's bound; both counterexamples replay through the scanner. In all, 360 of the 1125
exercised pairs are safe modulo their discarded set. The state-equality condition admits 265 of them and the coalesced
test 306. Neither admits an unsafe pair, and the coalesced test loses none that the state-equality condition admits. The
decision answered every exercised pair and agreed with the oracle on every pair the oracle settles. These counts
describe one generator and are stated only because \code{validation.cpp} in the artifact reproduces them exactly. The
generator is seeded and its draw order is pinned, so the counts do not vary between compilers.

A proportion from a random sweep describes the generator as much as the condition, so it is worth naming examples. Take
the discarded tokens \code{ab*} and \code{b+} alongside a kept token \code{c}, whose letter neither of them uses. An
occurrence of \code{b} sits in one of three places. At the start of a token it can only open a \code{b+}, so the cut
falls on a token boundary and the spliced scan is the serial one. Inside an \code{ab*} token the left piece is a shorter
\code{ab*} and the right piece is a \code{b+}, both discarded. Inside a \code{b+} token the cut leaves two shorter
\code{b+} pieces, discarded as well. Every occurrence of \code{b} is thus safe modulo the discarded kinds, at every
length. Condition 2 of Definition~\ref{def:trivia-certified} nonetheless fails. Advancing on \code{b} from inside
\code{ab*} reaches a state accepting \code{ab*}, and advancing on \code{b} from the initial state reaches one accepting
\code{b+}. Since those accept different tokens, minimization keeps them apart. Rejecting this safe \code{b} is the cost
of insisting, as the state-equality condition does, that the two scans reconverge on one state. The coalesced test
admits this \code{b}, since with the two discarded kinds one color the two states have the same future: both observe
discard, both stay put on \code{b}, and neither has another transition. This token set shows both that the
state-equality condition is conservative and what the coalesced test gains.

The token set of Proposition~\ref{prop:coalesced-strict} with \code{a}, \code{b}, \code{c}, \code{abc} discarded and
\code{k}, \code{kk} kept shows that the coalesced test is conservative in turn. There every cut before \code{b} is safe,
the state past the cut accepts nothing while the state the restart reaches accepts \code{b}, and only the exact decision
admits \code{b}. Because the automaton is minimized before the test, the state equality in condition 2 is as tight as
the minimizer of Section~\ref{sec:necessity} makes it. This equality is tighter than under no minimization and weaker
than under Myhill-Nerode, since that minimizer keeps empty-right-language states apart. The coalesced test is the
quotient that treats accepting labels in $I$ as one, and the first example is a pair it merges. Both examples, the two
further token sets of Proposition~\ref{prop:coalesced-strict} and the two of Section~\ref{sec:exact} are asserted by
\code{validation.cpp} under four judges: the state-equality condition, the shipped rule, the exact decision, and the
bounded oracle to length six.

Which rule a program queries matters for what its agreement shows. The library's own property tests query the shipped
predicate over randomly generated automata on a four-symbol alphabet to a shorter bound, and in release
v1.2.0~\cite{munch} the predicate they tested was the state-equality condition. A test of a reimplementation would
justify only the reimplementation, so the sweep here queries the shipped predicate for the coalesced test, the rule of
release 2.2.0 of munch. For the state-equality condition it evaluates a copy inside the program, since that release no
longer ships that rule. As a regression check, the copy reproduces the counts 265, 97 and 99 that the shipped predicate
gave on the pinned draw before release 2.2.0 replaced it. The exact decision is likewise a program beside the library,
written from Section~\ref{sec:exact}, and its negative answers are checked by the scanner itself, since every
counterexample is replayed before it is counted. Its positive answers rest on Theorem~\ref{thm:exact} and are checked
against the oracle wherever the oracle has a verdict. The library's own \code{modulo\_counterexample()}, in release
3.1.0, implements the same decision within a cap on the states each search holds, and takes no part in this sweep.

\section{Applicability}
\label{sec:applicability}

The certificate needs the token set's cooperation, and this section asks how much real token sets give.
Table~\ref{tab:applicability} was produced by compiling each token set and reading the predicate for all 256 byte
values. The predicate implements Corollary~\ref{cor:useful}, so it already excludes the vacuous certificates of
Proposition~\ref{prop:vacuous}, and the column needs no further filtering. A ``none'' cell asserts the absence of useful
certificates, not of certified bytes, since a byte certified only vacuously, like $a$ in the example after
Theorem~\ref{thm:necessity}, is not counted. Rows two through four each add a single token kind to the grammar of row
one, and the additions do not accumulate, so each collapse is attributable to the token kind named. The conventional and
split-friendly rows recognize the same byte language and differ only in the tokenization, and the block-comment row
directly after them adds one token kind to the split-friendly grammar, so those three are read together.

\begin{table}[!t] \centering
\begin{tabular}{>{\raggedright\arraybackslash}p{7.0cm}>{\raggedright\arraybackslash}p{3.6cm}
    >{\raggedright\arraybackslash}p{3.6cm}}
\toprule
Token set & Useful certified & Useful modulo $I$ \\
\midrule
C-like: identifiers, numbers, whitespace runs, operators, punctuation & all operator and punctuation bytes & the same,
    plus space, tab and newline \\
\addlinespace
the first row plus string literals alone (no raw newline inside) & none & newline \\
\addlinespace
the first row plus \code{//} line comments alone & none & newline \\
\addlinespace
the first row plus block comments alone & none & the same \\
\addlinespace
JSON, the RFC 8259 lexical forms over bytes & none & tab, newline, carriage return \\
\addlinespace
log lines: a run of non-newline bytes, and newline & newline & the same \\
\addlinespace
C-like, conventional tokenization: line-bounded strings and \code{//} comments, whitespace runs including newline & none
    & newline \\
\addlinespace
the same recognized language, split-friendly: newline its own token, spaces and tabs a separate run & newline & the same
    \\
\addlinespace
the same plus block comments & none & the same \\
\addlinespace
the same three kinds, every body barred from the bytes that open another & none & the same \\
\addlinespace
the same three kinds with their ordinary bodies, the block comment alone barred from crossing a line & newline & the
    same \\
\addlinespace
Zig subset, conventional tokenization & none & newline \\
\addlinespace
the Zig subset, split-friendly & newline & the same \\
\addlinespace
\code{keyword\_scale\_builder()}, the construction-cost grammar & 16 of those 24 bytes & the same, plus space, tab and
    newline \\
\addlinespace
\code{build\_lexer(false)}, the scaling grammar & 13 of its own 14 & the same, plus space, tab and newline \\
\bottomrule
\end{tabular}
\caption{Certified bytes by token set, exactly and modulo the discarded set $I$. Three mechanisms explain the collapses
in the middle column. Free-content tokens absorb the alphabet, run tokens de-certify their own bytes, and a multi-byte
token de-certifies the bytes it continues past. Discarding undoes only the second. A severed whitespace run leaves two
whitespace runs and both are deleted, so the run-token obstruction alone disappears, while a string, comment or kept
token that also admits the byte still excludes it. The other two mechanisms survive. A state inside a string literal or
a block comment accepts nothing, so no cut there ends the left chunk on a complete discarded token. The continuation
bytes of a kept multi-byte token stay lost, because the token such a cut severs is kept and cannot vanish from the
streams. Both columns are asserted for every row by \code{figures/applicability.cpp}, the relaxed column under the
state-equality condition, the shipped coalesced test and the exact decision of Section~\ref{sec:exact} alike. Four rows
are also confirmed by splitting a corpus. The two benchmark-grammar rows are asserted against transcriptions whose
provenance the text details.}
\label{tab:applicability}
\end{table}

The program that produces the table is \code{figures/applicability.cpp}, archived with this paper. It links the library
and asserts every row, with one qualification for the two rows that restate grammars living in the benchmark tool. For
the scaling grammar the probe compiles the real \code{build\_lexer(false)} beside its transcription and checks agreement
on all 256 certified bits and on the exact token lengths of a corpus. The exact certified bitset and the sample token
lengths are thereby bound to the shipped grammar, while the modulo cell and the candidate denominator in the row's cell
still rest on the transcription. The construction-cost grammar has no such binding in the pinned v1.2.0 probe, whose
comment calls its transcription ``verified by eye only''. Later releases bind that row's exact certified bitset and
sample token lengths mechanically too, and its candidate denominator and modulo cell likewise still rest on the
transcription.

Three mechanisms explain the collapses. First, \emph{free-content tokens absorb the alphabet}: a string literal that may
contain \code{(} makes \code{(} a mid-token byte and de-certifies it everywhere, so one token kind with a near-total
interior alphabet removes almost every candidate.

Second, \emph{run tokens de-certify their own bytes}: with whitespace tokenized as a run over space, tab, and newline,
the second newline of a blank line is consumed by the run's continuation state, so newline becomes mid-token-consumable.
The byte's own token kills it, which is easy to miss.

Third, and in the same family, \emph{a multi-byte token de-certifies every byte that can follow its first}. The study's
first row spells each operator as a single byte, so each is consumed only at the initial state and all twenty-four
certify. Section~\ref{sec:evaluation} measures two different C-like grammars, and the last two rows of
Table~\ref{tab:applicability} name each by its builder function, since they disagree. The construction-cost grammar,
\code{keyword\_scale\_builder()} in \code{tools/benchmark/src/main.cpp}, spells operators as literals including
\code{++}, \code{==} and \code{->} and admits a decimal point inside a number. Every byte occurring after the first
position of some operator is consumed by a live mid-token state, so the seven candidates \code{= < > \& | + -} fail to
certify. The decimal point falls independently, because \code{decimal\_float()} admits it inside a number, bringing the
total to eight. The scaling grammar, \code{build\_lexer(false)} in \code{tools/benchmark/src/harness.cpp}, has a smaller
operator set in which every multi-byte operator has \code{=} as its only continuation byte, so only \code{=} is lost and
thirteen of its fourteen candidates certify. Both collapses are partial. They show that ``operators certify'' is a claim
about the particular literals a grammar registers, so the grammar has to be named.

These two C-like benchmark grammars recognize different operator and number languages, and they certify different sets.
The sharper claim, that certification depends on the tokenization and not on the recognized language, needs a pair
recognizing the same language. The table's conventional and split-friendly rows are that pair. They accept identical
input and differ only in whether newline is folded into the whitespace run or carries a kind of its own, and the first
certifies no useful byte while the second certifies newline. The kinds differ by that one token, which is the change
under study.

All three mechanisms point at the same design lever. Certification is a property of the \emph{tokenization} rather than
an intrinsic property of the recognized language, and a tokenization can sometimes be refactored without changing what
is recognized at all. The split-friendly row of Table~\ref{tab:applicability} keeps the identical C-like language but
tokenizes newline as its own single-byte token, leaves spaces and tabs as whitespace runs, and keeps strings and line
comments line-bounded. The row directly above it is the same language under the conventional tokenization and certifies
nothing, which isolates the change. In the split-friendly row newline certifies, and the folklore rule follows as a
practical corollary: for tokenizations in which newline is its own token, line-based splitting is sound when no other
token can contain a newline. The automaton-level statement remains the exact one; ``no token spans a line'' is
sufficient here but not necessary in general, since a token that merely \emph{begins} with newline leaves the
certificate intact. In the split-friendly tokenization studied here the one token kind that spans lines is the block
comment, and adding it destroys the certificate again. That is the formal shape of both ``you cannot chunk C by lines''
and the quoted-newline problem that pushed CSV parsing into speculation.

Four further rows measure what that lever costs, and the cost is smaller than it looks. The obvious repair for the
block-comment collapse is to stop the token kinds from colliding by barring every string and comment body from holding a
byte that opens another of them. That grammar certifies no useful byte, exactly or modulo, so the expensive restriction
buys nothing at all. The row below it restores the ordinary string and comment rules, whose bodies may hold one
another's opening bytes, and instead adds one restriction, that the block comment may not cross a line. There newline
satisfies the exact condition. Both rows sit on the split-friendly base, where newline is already a token of its own.
Line-boundedness is therefore the lever and separation is not, and a language designer who wants a certified split byte
pays for newline as a token and one line-bounded comment form, and keeps the ability to write a slash inside a string.

The last two of these four rows are that choice already made. Zig's reference~\cite{zig2026reference} states that there
are no multiline comments and that each line of code can be tokenized out of context, and its multiline strings are
written one \code{\textbackslash\textbackslash} line at a time. The subset is adapted from the released 0.16.0
reference's grammar appendix over the C-like base of the rows above, with the adaptations stated. The subset transcribes
the appendix's byte classes, the well-formed UTF-8 sequences its table lists, printable ASCII in string and character
bodies and every byte from space upward in line bodies. From them it transcribes the bodies of the appendix's strings,
quoted identifiers, character literals, line comments and line strings, with the strict escapes and the one-character
literal that grammar has. It adapts the tokens around those bodies. The whitespace bytes of the \code{skip} rule, space
and newline, are read as tokens of their own, and so are the line comments. The whitespace bytes form one whitespace run
in the conventional row, while the split-friendly row keeps newline apart from the space runs. The grammar folds
\code{skip}, which covers whitespace and line comments alike, into every token's tail, and besides that gives a line
string and a doc comment of either spelling a trailing run of space and newline bytes of their own. Here a line string
instead ends at the newline, leaving the run after its body outside it. A multiline string and a doc comment over
consecutive lines are one token each in the grammar and are read here one line at a time. The grammar tells the plain,
doc and container-doc comment spellings apart, but here they are one token, since, as byte strings, their union is
exactly the set of lines beginning \code{//}. In its byte classes the subset follows the appendix, even where the
appendix departs from the language of conforming source. The reference's source-encoding section forbids the delete byte
and malformed UTF-8 everywhere, but the appendix's line bodies admit both, and so do the subset's. The reference's
source encoding admits tab as a token separator and carriage return only immediately before a line feed. Neither occurs
in any rule of that appendix, so neither occurs in a token of the subset, and both are certified vacuously and withheld.

A shipped language chose the property this section formalizes, for its own reasons, and the rows show what it bought:
newline modulo the discarded set under the conventional tokenization, and newline certified outright once newline is its
own token. The Zig pair draws the same comparison as the conventional and split-friendly C-like rows, on a subset
adapted from a shipped language's grammar.

The JSON row of Table~\ref{tab:applicability} appears to contradict the literature, and the reconciliation lies in the
equivalence each result preserves. Prior claims that JSON may be divided at
newline~\cite{barenghi2015parallel,li2021plex} rest on a weaker one. Barenghi et al.'s sequential lexing algorithm
appends a byte to the lexeme buffer only when it is not reading a whitespace character outside a string or a comment, so
whitespace is unobservable in its output. Plex restates the claim on the different ground that newline occurs in no
lexeme, which holds only once whitespace runs are not themselves counted as lexemes. Either way the division can split
one maximal whitespace match into two discarded matches whose lengths sum to the original, leaving the surviving stream
unchanged. Theorem~\ref{thm:main} preserves something stronger, the complete sequence of emitted kinds and lengths, and
under that semantics a newline inside a maximal whitespace token is not a certified boundary. The grammar studied here
follows the RFC 8259 lexical forms~\cite{bray2017json} over byte input, with the full string escapes, signed numbers
with fraction and exponent, the three literal names, and whitespace runs emitted as tokens, so it is refused as it
should be. The grammar is a lexer over bytes, not a conforming JSON processor, and its string interiors admit any byte
from \code{0x20} up except quote and backslash. UTF-8 well-formedness, which RFC 8259 requires of JSON exchanged outside
a closed ecosystem, is assumed of the input and not checked here. Well-formedness is orthogonal to certification, since
validating it would only remove bytes from string interiors and so could not de-certify anything that certifies now.

The two results do not conflict; they quantify over different streams, and this paper supplies both.
Definition~\ref{def:certified} rejects newline for JSON, since the whitespace continuation state consumes it either way,
and Theorem~\ref{thm:trivia} accepts it once whitespace is declared discarded, which is the JSON row of
Table~\ref{tab:applicability}, none exactly and tab, newline and carriage return modulo $I$. The prior claim is derived
from the token set. Recovering certification under the \emph{exact} equivalence still needs the tokenization changed, by
making newline its own token as the split-friendly row does.

The right column of Table~\ref{tab:applicability} places the whitespace and comment tokens of each row in $I$ and reads
the relaxed condition. Eight of the fifteen rows gain bytes, and the pattern behind them is the second of the three
collapse mechanisms named above. Wherever whitespace is a run token the caller discards, its own bytes come back,
because severing a whitespace run at a whitespace byte yields two whitespace runs and both are deleted. Discarding
undoes this one mechanism, and the other two survive it. A byte that a string literal or a block comment absorbs is
recovered nowhere, because a state inside either accepts nothing at all. The left chunk cannot end there on a complete
discarded token, which is the first thing $C$ requires. A byte that a kept multi-byte token continues past is likewise
not recovered, because the state that has consumed its first byte can still reach the kept token. The third clause of
$C$ fails, and rightly, since severing a kept token is a difference the deletion preserves.

In the exact column each outcome is a fact about the tokenization, since Theorem~\ref{thm:necessity} makes that
condition necessary as well as sufficient. The modulo column is read from a conservative condition, so on its own a byte
absent there might be safe modulo $I$ and merely refused. On these fifteen rows no byte is merely refused.
\code{applicability.cpp} runs the decision of Section~\ref{sec:exact} on every useful byte of every row, one decision
per class of bytes the tables do not tell apart. It finds every byte absent from the column unsafe, with a
counterexample replayed through the scanner, and every byte present safe, so the column is exact on every row. The
coalesced test of Section~\ref{sec:coalesced} gives the same column on every row, which the same program asserts from a
copy of the state-equality condition beside the shipped predicate.

Newline is recovered for the conventional C-like token set because line-bounded strings and \code{//} comments cannot
contain one, so a whitespace run is the only token whose interior admits it. Space is recovered in no C-like row
carrying strings or comments, since it sits legally inside both; the first row, which has neither, recovers it along
with tab and newline. The JSON row is the same lexer over bytes, its UTF-8 caveat unchanged from above. JSON gains tab
and carriage return as well as newline because RFC 8259 excludes every raw byte below \code{0x20} from string interiors,
which confines all three to whitespace, while space is admitted there and so is not recovered. That the recovered set is
exactly JSON's whitespace minus space follows from the RFC's own exclusion. The rows carrying unrestricted block
comments continue to hold no useful certificate, because a comment severed at a newline leaves \code{/*\,a} and
\code{b\,*/}, and both halves re-tokenize completely. In this token set the halves become the operators \code{/} and
\code{*} and an identifier, six kept tokens in place of one discarded comment, which is the difference a caller can see.
The correct conclusion carries the certificate's own scope. On completely tokenizable inputs, this tokenization of a
C-like language with unrestricted block comments cannot be split at every occurrence of a byte value chosen from the
grammar alone, newline included, without state, overlap, speculation or repair, and the certificate says so. Schemes
that inspect the document and split at some occurrences but not others are outside the claim.

Four rows are cross-checked by splitting as well as by reading the predicate: the conventional and split-friendly C-like
pair, the split-friendly row with block comments, and JSON. On those rows the condition agrees with brute-force
splitting on every candidate byte the corpus exercises. These candidates are a declared set of sixteen bytes, not all
256: the four whitespace bytes, eleven representative structural characters and the letter \code{t}. The corpus check is
weaker than the exact decision above and predates it. The check is kept because it exercises the scanner on documents,
whereas the decision reads the tables.

The practical content is narrower than a count of moved cells suggests, and more useful. The exact certificate already
admitted newline for the split-friendly tokenization in which newline is its own token and whitespace runs exclude it.
Getting newline that way means redesigning the token set, a change the parallelization technique imposes on the user.
Theorem~\ref{thm:trivia} removes the imposition, since the conventional whitespace rule that hand-written lexers already
contain certifies newline without modification. The conventional and split-friendly rows give the same answer, and the
redesign buys nothing for a caller comparing streams modulo $I$. The relaxation is better read as widening the class of
token sets the method accepts as written than as adding split points to a fixed one.

\subsection{Framing at the producer}
\label{sec:framing}

Definitions~\ref{def:certified} and~\ref{def:trivia-certified} quantify over the automaton, the second also over the
declared discarded set, and mention no input. Certification is static. It depends on the token set alone for the exact
condition and on the pair of token set and $I$ for the relaxed one, and no transformation applied when reading a
document can make a byte certified that was not. This directly rules out the appealing idea of normalizing input at read
time to manufacture split points. The same fact also explains why the obvious attempt is circular, since deciding which
occurrences of a byte lie inside a token is the sequential problem that chunking was meant to avoid.

A particular document can carry a guarantee of its own, supplied by whatever produced it. Reserve a byte $r$ that no
token admits at any noninitial position, give it a rule of its own, and require the producer to emit $r$ only at token
boundaries and to escape it wherever it would otherwise appear inside a token. Noninitial exclusion is the load-bearing
clause, since if some accepted token $urv$ contained $r$ noninitially, a scan could consume $r$ from the live state
after $u$ and the byte would not certify. Excluded everywhere but the start, the reserved byte's only live consumption
is from $q_0$, and Definition~\ref{def:certified} admits it with no appeal to Theorem~\ref{thm:trivia}. For the
block-comment row of Table~\ref{tab:applicability}, whose certified set is empty under both conditions, reserving
\code{0x1E} yields a certified split symbol. A framed document splits at every occurrence with the token stream
reproduced exactly, including across block comments spanning lines and string literals containing punctuation.

None of that construction is new. ASCII reserves \code{0x1E} as \textsc{record separator}, one of four information
separators defined for this kind of framing~\cite{cerf1969ascii}, and JSON text sequences standardize that byte as a
record prefix~\cite{williams2015rfc7464}. Informal newline-delimited JSON makes the same bargain in different notation
without the reserved byte, while length-prefixed framing and columnar row groups are only analogies, buying random
access by other means. Barenghi et al.~\cite{barenghi2015parallel} reach a comparable place by constraining the source
language instead. Chassot and Kun{\v{c}}ak~\cite{chassot2026ziplex} study a different but closely related maximal-munch
boundary question for invertible lexing. Given already formed adjacent tokens, their sound R-Path relation checks
whether the first token's text followed by the next token's first symbol remains a prefix of any rule; their separator
construction can instead insert a separator during printing and compare the re-lexed stream modulo separator tokens.
These operations protect known token sequences during printing and recombination, but they neither locate cuts in raw
input nor derive bytes safe at every occurrence before lexing.

The certificate contributes the check. A format designer choosing a delimiter is making a claim about every token the
grammar admits, and Definition~\ref{def:certified} decides that claim from the compiled automaton, with nothing left to
inspection. Read this way an applicability table is not only a report of where the method applies but a design rule for
formats intended to be chunked. By Theorem~\ref{thm:necessity}, a row whose exact column is empty is a warning that no
single-byte delimiter preserves the exact stream without changing the token set. An empty modulo column is a refusal by
a conservative condition, which the decision of Section~\ref{sec:exact} turns into a verdict, and on the rows here every
refusal stood.

The limit is equally clear. A document already written cannot be reframed without reading it, so this recovers nothing
for a compiler consuming source it did not generate. Framing applies where the writer cooperates, which covers generated
code, logs, exports and interchange formats, and not to arbitrary input.

\section{Evaluation}
\label{sec:evaluation}

The evaluation reports two machines. The first is the one the method was developed on, and it is retained because the
contrast with the second is a result: a figure that looks like a property of the API on one machine does not hold its
value, or even its sign, across environments and benchmark revisions.

Everything measured here uses the exact certificate. Both \code{chunk\_boundaries()} and the parallel executor built on
it consult the exact map, so no figure in this section reports what the relaxed condition would buy at run time.
Section~\ref{sec:applicability} establishes which bytes it recovers, a statement about token sets that says nothing
about throughput. A token set that certifies no useful byte exactly can gain a boundary set from the relaxation, as five
of the eight such rows do and the block-comment rows do not, and eight of the fifteen rows gain bytes in all. Turning
that into a measured speedup, however, needs a planner this library does not yet ship.

\paragraph{Environment A, virtualized.} An Intel Core i9-12900K, a hybrid part with eight performance cores carrying
two-way SMT and eight efficiency cores, so sixteen physical cores and 24 logical processors, with 32\,GiB of memory. The
host is Windows 11 build 26200 with virtualization-based security enabled. The measurements ran under WSL 2.7.11 on
kernel 6.18.33.2-microsoft-standard-WSL2, Ubuntu 24.04.4 LTS, glibc 2.39, with all 24 logical processors and 15\,GiB of
memory visible to the guest. Compiled by GCC 13.3.0 at \code{-O2}. Corpora are 16\,MiB and fixed-seed, and scenarios run
in a fixed order. The archive holds summaries only (\code{data/benchmark.txt}), whose header records the machine,
operating system, compiler and command, while the Windows build, WSL and kernel versions, memory and topology details
above were recorded by the author and are not archived.

Four properties of this environment bound what its figures establish, and we state them without adjusting for them. The
guest does not see the host's hybrid topology, reporting twelve cores of two threads each, so a performance core is
indistinguishable from an efficiency core from inside the measurement. Threads are not pinned, and could not usefully be
while that holds. The guest is given no \code{cpufreq} interface, so the clock is neither fixed nor observable. The
machine was not quiesced. The corpora also fit in cache, 16\,MiB against a 30\,MB last-level cache, so the run targets
warm-cache conditions without certifying residency throughout.

\paragraph{Environment B, bare metal.} An AMD Ryzen 9 9950X3D, sixteen physical cores with two-way SMT across two L3
domains of eight cores each, 128\,MiB of L3 in total, 59\,GiB of memory, Ubuntu 26.04 on kernel 7.0.0-28-generic,
compiled by GCC 15.2.0 at \code{-O2}. The \code{performance} governor was set, which biases the clock toward its maximum
but does not fix it. Boost remains enabled, and the archived \code{lscpu -e} of the pinned run reported in
Table~\ref{tab:scaling} shows 22 distinct frequencies between 624 and 5711\,MHz across the 32 processors at the instant
it was taken. The clock is observable here, unlike in Environment A, but not controlled. The topology is visible and
archived. The ten scaling scenarios run in interleaved rounds, each reshuffled from a fixed seed, while the
construction, planning and thread-launch scenarios run as blocks. The corpora sweep 1, 16, 128 and 512\,MiB in that
fixed ascending order. Every pass of the ten scaling scenarios is recorded individually in
\code{data/bare-metal-unpinned-run2/} and \code{data/bare-metal-pinned-run2/}; the construction, planning and
thread-launch figures below are recorded only as summaries over 15 passes.

This environment is measured in two placements, each run twice, giving four archives. The unpinned runs may place
threads on any of the 32 logical processors. The pinned run confines the process to the eight physical cores of a single
L3 domain (\code{taskset -c 0-7}), which excludes the SMT siblings and the other L3 domain from the set the scheduler
may use, so eight threads have eight distinct physical cores available to them. The archives certify the width of that
mask but not its members. The collector recorded \code{nproc} as 8 from inside the affinity mask, while the archived
\code{lscpu -e} describes the whole machine. Neither the mask nor the launch command was captured, so the specific
cores, the single L3 domain and the SMT exclusion rest on the collection notes alone. The mask constrains placement
without fixing it, since threads remain free to migrate within those eight cores and nothing prevents two from sharing
one. The machine was not quiesced either, and load averages stood at 0.55 and 1.83 before the unpinned and pinned runs
of Tables~\ref{tab:scaling} and~\ref{tab:planning}. The two runs started two minutes apart, so the second load average
is likely the first run's own residue. Either way, the machine was busy when the second run began.

The 512\,MiB corpus is four times the machine's aggregate last-level cache, and the pinned process is confined to one of
the two L3 domains, so the corpus exceeds whatever share is actually reachable by a comfortable margin. The archive does
not record the two domains' capacities separately, so the sweep is read conservatively: 512\,MiB is out of cache under
any split of the 128\,MiB aggregate.

In both environments, and before the scaling rows it protects, the benchmark checks by an exact (kind, length)
comparison that the eight-chunk token stream on each corpus and size is identical to the serial one. The timed rows
carry a lighter per-pass check of each chunk's consumed length against a recomputed plan and of the total token count.

Both corpora are favorable to the method by construction, and the scaling figures should be read with that in mind.
Neither contains a string literal or a comment, the two constructs that absorb the alphabet and empty the certificate,
and the grammar driving them certifies 13 of its 14 operator and punctuation candidates. The measurements below
therefore show how the parallel machinery scales on a token set the certificate suits, and are not evidence about how
often a token set does suit it. Table~\ref{tab:applicability} speaks to that, and eight of its fifteen rows admit no
useful exact certificate at all. The two questions are separate. Even when a grammar admits a certified byte, that byte
yields no parallelism on a corpus in which it never occurs, as automatically generated JSON without newlines already
shows.

\subsection{Scaling}

Table~\ref{tab:scaling} reports median throughput on Environment B, pinned, which is the most constrained configuration
measured: visible topology, eight physical CPUs of one L3 domain with the SMT siblings excluded, a clock biased to
maximum, and a corpus sweep that leaves the cache. Two baselines matter. They answer different questions, and only one
of them turns out to hold still between collections. The parallelism is measured against the same validating wrapper
driven with one requested chunk, which plans, uses the per-chunk sink, and spawns nothing. Because the per-pass
validation recomputes the plan, every chunked row, this baseline included, plans twice per pass. When certified bytes
are as frequent as these, planning costs a fraction of a microsecond (Table~\ref{tab:planning}) against milliseconds to
hundreds of milliseconds of scanning, so the ratios are unaffected and the comparison against the plain scan is
conservative. Against the one-chunk baseline the certified splitting reaches $1.96\times$, $3.87\times$, and
$7.63\times$ on the 512\,MiB dense corpus, or 98\%, 97\% and 95\% parallel efficiency; the other collection of the same
placement gives 95\%, 93\% and 93\%. A user choosing between the serial and parallel entry points compares instead
against the plain scan, and gains $3.46\times$ at four chunks in this collection and $3.94\times$ in the other
collection with the same placement. End to end here means the benchmark wrapper as described above, validation path and
second planning pass included. Conflating the two baselines, as an end-to-end table alone would, charges parallelization
for the sink and inherits an unstable denominator.

\begin{table}[!ht]
\centering
\begin{tabular}{lrrrrrr}
\toprule
Corpus & size & plain scan & 1 chunk & 2 chunks & 4 chunks & 8 chunks \\
\midrule
dense (1.83 B/token)  & 1\,MiB    & 767.5 & 728.5 & 1436.1 & 2780.4 & 4685.9 \\
                      & 16\,MiB   & 771.5 & 686.9 & 1374.3 & 2679.5 & 5074.6 \\
                      & 128\,MiB  & 788.4 & 722.8 & 1432.1 & 2782.3 & 5420.4 \\
                      & 512\,MiB  & 818.0 & 730.1 & 1434.2 & 2828.0 & 5568.3 \\
\midrule
source (3.53 B/token) & 512\,MiB  & 757.5 & 711.1 & 1415.0 & 2763.0 & 5500.5 \\
\bottomrule
\end{tabular}
\caption{Median throughput in MiB/s on Environment B, second run, with the process confined to an eight-CPU affinity
mask, which the collection notes place on one L3 domain. The one-chunk column is the parallel API without parallelism,
and is the baseline the parallel efficiencies use. The plain-scan column does not hold still between collections and is
discussed below.}
\label{tab:scaling}
\end{table}

In the measured sweep, on this machine, efficiency did not deteriorate as the input footprint exceeded the aggregate
last-level cache capacity. The 512\,MiB corpus is four times the machine's 128\,MiB of last-level cache, and its
eight-chunk efficiency, 95.3\%, is the highest of the four sizes. The figures are 80.4\%, 92.3\%, 93.7\% and 95.3\% at
1, 16, 128 and 512\,MiB, so efficiency rises with input size and keeps rising across the cache boundary. The sweep does
not locate a break-even input size, because it never reaches one. Even at the smallest size measured, eight chunks reach
$6.43\times$ the throughput of the one-chunk API. The sweep does show where fixed overhead becomes visible, at 1\,MiB,
where efficiency falls to 80.4\%. Coordinating an eight-thread scan, seven spawned workers joined by the calling thread,
costs about 40\,$\mu$s against a scan of roughly 1.4\,ms, which is consistent with the drop without isolating
coordination as its only cause.

\paragraph{Observed differences between affinity runs.} Measured unpinned, free to use all 32 logical processors, the
same sweep in the second revision shows lower eight-chunk efficiency, 93\% against 95\% pinned at 512\,MiB and 68\%
against 80\% at 1\,MiB, a gap of twelve points at 1\,MiB and under three at 512. The gap does not narrow monotonically.
It is narrowest at 128\,MiB, one point, and widens again at 512. The first revision has the same shape one level down,
92.6\% pinned against 88.5\% unpinned at 512\,MiB, so the placement cost there is about four points, against under three
in the second revision. The two- and four-chunk figures move by about 2.2 points or less at 128 and 512\,MiB and by more
at the smaller sizes. One hypothesis is SMT: with eight threads free on 32 logical processors, some pairs share a
physical core. This data cannot settle it, since for each benchmark revision there is one run per placement, taken at
different times with neither replication nor interleaving, and a gap also appears at two chunks where collisions should
be rare. Settling it would need paired repeated runs, or per-thread CPU residence recorded during the scan. The two runs
support something narrower: on these runs the restricted CPU set had the higher eight-chunk efficiency, though not
always the higher absolute eight-chunk throughput, and certified splitting delivered most of its benefit in both.

\paragraph{The two collections are not a controlled repeat.} Environment B was measured twice in each placement, but the
two collections were taken at different commits. Between them the benchmark began writing round-trip-safe CSV values and
gained the four plan-and-execute scenarios that now run before the scaling sweep. The scaling code is unchanged, but the
executable, its layout, and what executes before the measurement are not. These are two benchmark revisions on one
machine, and the difference between them cannot be attributed to the machine.

The largest differences are in the single-threaded rows, and their direction is the same under both placements. At
512\,MiB on the dense corpus the plain scan rose 14.3\% pinned and 17.4\% unpinned, while the one-chunk row fell 3.7\%
and 3.8\% respectively. At that size the multi-chunk rows move much less, by about 0.8\% on the dense corpus and about
1.4\% on the source corpus when pinned, and by 1.1\% and 0.6\%, respectively, when unpinned. That stability is specific
to 512\,MiB and does not hold across the sweep, where multi-chunk rows move by as much as 5.9\% pinned and 6.3\%
unpinned, both at 16\,MiB. The one-chunk row's shift agrees within 0.2 points across the two affinities, 0.16 on the
dense corpus and 0.01 on the source one, which points to the revision rather than the scheduler, though these runs do
not isolate the cause.

One estimator is used throughout, and it is worth naming because a second one is available and does not always agree.
Every ratio quoted here divides one scenario's median throughput by another's, each median taken over that scenario's
passes. It is not the median of the per-round paired ratios, which would weight rounds equally, while the ratio of
medians summarizes each scenario first. The harness neither prints nor asserts the ratios, which are recomputed from the
archived per-scenario medians, so the archive carries the checkable numbers and the divisions are ours. The two do
diverge: at 512\,MiB pinned, two chunks over one is $1.96\times$ as a ratio of medians and $1.99\times$ as a median of
paired ratios. Below, paired results are reported only as same-round win counts.

On Environment B at 512\,MiB the comparison between the plain scan and the one-chunk row moves between collections far
more than any efficiency or multi-chunk row does. Three reported comparisons on the dense corpus show how far it moves:
the one-chunk row stood 14\% \emph{below} the plain scan on Environment A at 16\,MiB, 5.9\% \emph{above} pinned in the
first Environment B collection at 512\,MiB, winning 13 of 15 same-round pairs, and 10.7\% below in the second, reported
in Table~\ref{tab:scaling}, winning none of 15. We report it without an explanation. The consequence for the figures
above is that efficiencies measured against the one-chunk baseline moved by two to four points between the collections,
93\% to 95\% at eight chunks, while the end-to-end ratio, which divides by the plain scan, moved by 12\%. The
efficiencies are quoted as a range of a few points and the end-to-end ratio as a range nearly $0.5\times$ wide.

\begin{table}[t]
\centering
\begin{tabular}{lrrr}
\toprule
Certificate density & chunks planned & median & bytes scanned \\
\midrule
newline every 40 bytes    & 8 & 0.1\,$\mu$s     & $<$ 1\,KiB \\
newline every 1\,MiB      & 8 & 1021.0\,$\mu$s  & 3.5\,MiB (derived) \\
certified byte absent     & 1 & 16365.5\,$\mu$s & 56\,MiB (derived) \\
nothing certified         & 1 & 0.0\,$\mu$s     & 0 \\
\bottomrule
\end{tabular}
\caption{Boundary planning for $k=8$ over 16\,MiB on Environment B, pinned. The two scanning rows show similar effective
scan rates, 3.35 and 3.34\,GiB/s, consistent with the $O(kN)$ bound, and the last row is the empty-certificate fast
path. Each pass repeats the plan until the sample outlasts the clock's resolution, so the sub-microsecond rows still
measure planning. The bytes-scanned column comes from the search geometry, with no instrumentation. Seven interior
boundary searches of half a megabyte each give $7 \times 0.5 = 3.5$\,MiB, and with the byte absent the searches from the
boundaries at $2, 4, \ldots, 14$\,MiB each run to the end of the input, $14 + 12 + \cdots + 2 = 56$\,MiB.}
\label{tab:planning}
\end{table}

\subsection{Planning cost}
\label{sec:planning}

Table~\ref{tab:planning} measures the planner alone, over a 16\,MiB input divided into eight chunks, as certified bytes
grow scarce. In the first three rows the certified byte is newline, placed every 40 bytes, every 1\,MiB, or nowhere, and
the last row uses a token set that certifies nothing at all. Scarcity is not a contrived case, since Plex gives
automatically generated JSON, which usually lacks newline characters, as the reason splitting there is
infeasible~\cite{li2021plex}, and Barenghi et al. make the general point that a language need not offer any separator
identifiable from a bounded window at all~\cite{barenghi2015parallel}.

In the two middle rows, which cross-check each other, 3.5\,MiB scanned in 1.02\,ms and 56\,MiB in 16.4\,ms give similar
effective scan rates, which is consistent with the $O(kN)$ bound and indicates that the forward searches dominate the
planner's cost. Even its worst case, a certificate whose byte never occurs in a 16\,MiB input, costs 16.4\,ms once. A
caller planning a much smaller chunk count, or planning once and scanning repeatedly, pays proportionally less. The
empty-certificate case is answered without touching the input.

That worst case is a cost the caller pays before any scanning begins, and on such an input the planner returns a single
chunk, so the scan that follows is the serial one. Stating that cost end to end needs planning and scanning timed on one
grammar over one input, and the earlier planning rows do not provide that. The planning rows use a two-token grammar
over a synthetic input while the scaling rows use the C-like grammar over generated source, so adding those would
compare different lexers over different inputs. A separate scenario times both phases together on the planning workload.
With the certified byte absent, the parallel entry point takes 34.5\,ms against 18.2\,ms for the serial one, so choosing
it costs about $1.9\times$. With certified bytes every 40 bytes it takes 1.7\,ms against 12.4\,ms, a gain of about
$7\times$. Both ratios are of one-decimal summary medians. The serial and parallel entry points report the same
aggregate token count on each workload, 838861 tokens on the one and one token on the other.

The penalty is a property of the planner as implemented, not of the certificate. The planner here is simple, with one
forward search per boundary. The one-pass $O(N + k)$ formulation above would reduce, but not eliminate, the worst-case
planning cost by making a single pass over the input. At the observed scan rate, that pass would take an estimated
5\,ms, which is not a measured figure. The distinction the worst case draws attention to is that \emph{deciding} the
certificate needs no input at all, while \emph{locating} its occurrences is a runtime scan whose cost depends on how
often the byte actually appears.

\subsection{Correctness}

Three levels of evidence support the implementation, and they check different things. Before the scaling rows it
protects, the benchmark's preflight compares the serial stream token for token against planned chunks scanned
independently, at eight chunks on each corpus and size. The unit suite carries the counterexamples of
Section~\ref{sec:subtlety}, including the nullable and cyclic re-entry cases. A fuzzer generates arbitrary token sets
and inputs, checks every planned boundary against the predicate, and exercises the actual parallel executor, comparing
the concatenated stream against the serial one on every execution. Equality is asserted for inputs the serial scan
consumes completely, and prefix equality for malformed ones, matching the theorem's scope. The pinned workflow
configures a bounded fuzzing job on every continuous-integration run. No violation has been reported by those jobs or by
longer local runs; because their logs are not archived, those outcome claims are author reports and cannot be checked
from the artifact.

\section{Limitations and future work}

The result is scoped to one fixed token automaton restarting from one fixed initial state. A scanner carrying state the
automaton does not represent falls outside it. Lexer modes, indentation stacks, semantic predicates, and hand-written
scanning for constructs the regular grammar cannot express all carry information across token boundaries. To remain
within the present theorem such finite external state must be folded into an enlarged automaton. Otherwise it must be
recovered or communicated at chunk starts, or conservatively fenced into regions the plan does not cross. The theorem
also specifies the concatenated token sequence, not callback interleaving, so the sink must tolerate concurrent calls
and order-sensitive effects need per-chunk buffering followed by ordered replay.

The approach trades generality for certainty. When the grammar does not cooperate it offers no usable split points, by
design, and the composition and speculation families of Section~\ref{sec:prior} remain the applicable answers we know
of. Safety modulo $I$, which the exact decision of Section~\ref{sec:exact} settles, is PSPACE-complete for general token
automata by Theorem~\ref{thm:pspace}, so a polynomial algorithm for it exists only if P equals PSPACE. The hard
instances encode a family of automata in their discarded tokens and resemble none of the fifteen token sets studied
here. For literal vocabularies, safety modulo $I$ is decidable in polynomial time by Theorem~\ref{thm:literal}.

Several extensions look natural, and the first has since been carried out. A companion paper generalizes the certificate
from single bytes to short byte windows, certifying a cut at a fixed offset inside every occurrence of a multi-byte
string, which recovers splitting for some grammars where no single byte certifies~\cite{nidhogg2026splitwindows}. The
useful certificate of this paper is exactly the length-one case of the companion's window model. A hybrid plan could
split at certified bytes where they exist and fall back to speculative entry elsewhere, keeping the guarantee where it
is free and paying for it only where it is not. On completely tokenizable input, a candidate scan begun in a live state
cannot emit a token crossing the cut immediately before a certified occurrence, since an emitted token's match path lies
in live states and only a non-re-entrant $q_0$ consumes the certified symbol among them. If such a scan commits emitted
tokens past the occurrence, it places a boundary there. The serial scan places the same boundary, and both restart in
$q_0$ at it, so their continuations agree when run over the same remaining input. Whether a concrete speculative
protocol can exploit this to confine misprediction repair to the gap before the next certified occurrence is future
work, as is the distribution of such gaps over representative corpora. The grammar-refactoring lever of
Section~\ref{sec:applicability} could be automated as a search, under a stated observational equivalence and edit-cost
metric, for a minimum-cost re-tokenization that makes a chosen byte usefully certify.

The certificate is also orthogonal to a separate worst-case cost of naive longest-match scanning. Some token sets and
inputs cause failed lookahead that grows at successive token starts, producing $\Theta(N^2)$ work, while bounded failed
lookahead stays linear. Where sufficiently frequent certified points divide that work into reasonably balanced chunks,
fixed-$k$ parallel execution can reduce its coefficient, but each chunk retains the same quadratic worst case, and
sparse or absent certificates provide no relief. Over the tokens \code{a}, \code{a+b} and \code{;} on input $a^N$, the
semicolon certifies yet never occurs, so the planner returns one chunk. The quadratic pathology can be removed with
per-position information, obtained either during the scan or before it. Reps keeps the backtracking loop but tabulates
the state and position pairs that have already failed, so no doomed transition is repeated~\cite{reps1998maximal}, while
the uniform-tokenization algorithms precompute what each suffix admits in an initial right-to-left
pass~\cite{li2025uniform}.

Finally, the evaluation rests on two x86-64 machines, one Alder Lake under virtualization and one Zen 5 on bare metal,
with threads never pinned individually and confined to a CPU set only in the restricted runs. A non-x86 architecture,
and per-thread placement, would strengthen the scaling claims and the implementation's portability evidence, though not
the architecture-independent theorem.

\section{Related work}
\label{sec:related}

\paragraph{Reset words and synchronizing codes.} The live subautomaton $A^+$ is a partial DFA, and on a partial DFA a
reset word is one whose action carries a non-empty set of states to a single state and is undefined
elsewhere~\cite{berlinkov2021partial}. A useful certified symbol satisfies that definition. Certification says that no
state other than $q_0$ has a $b$-transition surviving in $A^+$, and usefulness says $q_0$ has one, so the action of $b$
on $A^+$ is defined on $\{q_0\}$ alone and carries it to $\delta^+(q_0, b)$. Every useful certificate is a one-letter
reset word of $A^+$ whose action domain is $\{q_0\}$; a vacuous one is a letter whose action is nowhere defined. The
converse need not hold, since certification additionally requires that $q_0$ is not re-entrant, a maximal-munch boundary
condition absent from the reset-word definition. The identification is specific to the partial-DFA convention of
Berlinkov et al.~\cite{berlinkov2021partial}. Under the classical complete-DFA reading~\cite{volkov2008synchronizing},
in which a reset word must map \emph{every} state to a common state, the correspondence fails, and a scanner's live
automaton is partial by construction.

The connection is more than terminological. A code is synchronizing when it admits a word $w$ such that an occurrence of
$ww$ lets a decoder resume independently from the position after the first $w$, and this allows a coded message to be
decoded in parallel~\cite{berlinkov2021partial}. Theorem~\ref{thm:main} has that shape, and UTF-8 is its everyday case.
Any byte beginning a form resynchronizes a decoder, and Section~\ref{sec:certificate} recovers exactly those bytes as
the useful certificates from the compiled table. That such a letter permits independent resumption is therefore known,
and the novelty lies in four things around it. First, the re-entrancy side condition ties the reset to
\emph{token boundaries under maximal munch} rather than to automaton states alone, so the resumed scan agrees with the
serial one on token lengths and kinds and not merely on where it sits. Second, the characterization of
Section~\ref{sec:necessity} makes the condition necessary as well as sufficient on the trim live automaton, so that a
token set certifying no useful byte has no byte that both occurs in a completely tokenizable input and universally
begins a token. Third, we derive the entire certified set from a compiled token set at build time, with no assumption
that a distinguished word occurs. Last, the relaxation modulo discarded tokens has no direct counterpart in the cited
synchronizing-code model, a model that designates no discarded output-token class. The equivalence behind that
relaxation is not the novelty. ZipLex's separator construction proves a printing guarantee modulo inserted separator
tokens~\cite{chassot2026ziplex}, but there the class marks separators inserted between formed tokens, while here it
marks tokens a raw-input cut may sever.

\paragraph{What certification adds.} Certified split points differ from simultaneous automata in requiring no enlarged
automaton and no composition, and from speculation in requiring no guess and no repair. Certification is a per-symbol
property derived from the token loop's reset, where $k$-locality is a uniform property of the automaton. The local
parsability that operator precedence languages enjoy~\cite{barenghi2015parallel} is a property of the syntax grammar
that permits substrings to be parsed independently using bounded surrounding context, while certification is a property
of the token automaton that licenses a scanner to restart at a byte. Unlike realignment after the fact, certification
merges nothing, and unlike delimiter practice it derives the safe set from the compiled token set instead of assuming
it, or proving it, per format~\cite{ge2019csv}. The derivation also correctly refuses the grammars for which the
folklore is unsound.

Relocating cuts to candidate separators is classical, and what this paper adds is narrow. Barenghi et al. split JSON at
whitespace with explicit ambiguity resolution, observe newline as a safer but potentially sparse alternative, and
constrain Lua before using newline as a split character~\cite{barenghi2015parallel}. An automatic grammar-only test
deciding when a byte permits direct restart from $q_0$ appears to remain missing. Among the surveyed generators,
re2c~\cite{trofimovich2020re2c} validates a user-selected end-of-input sentinel by rejecting rules in which it may occur
before the end of a lexeme. Its documentation records the check as added in release~1.3 in December
2019~\cite{re2cdocs}. That is a terminal-in-lexeme test for one declared byte, a safe-after property under the
completely-tokenizable-input model used here. It analyzes the compiled automaton statically and is incomparable with the
certificate, since neither test implies the other. Over a token set whose only token is \code{ab}, the byte \code{b}
passes the sentinel check because it ends the lexeme, yet it lies inside a token and is not certified. Over \code{ba},
the byte \code{b} certifies, because every occurrence begins a token and the initial state is not re-entrant, yet it
fails the sentinel check because the lexeme continues past it. No implementation we know of derives the complete set of
symbols whose every relevant occurrence begins a token, or checks the initial-state condition that guarantee requires.
The certificate supplies that test from the compiled automaton, with no hand analysis and no input-dependent
context-recovery pass, and the study of Section~\ref{sec:applicability} turns the same machinery into a statement about
which token sets admit such bytes at all.

The analysis is not merely sound. Theorem~\ref{thm:necessity} shows the condition is also necessary once the analysis is
restricted to $A^+$, as in the derivation of Section~\ref{sec:necessity}, so on a compiled token set the
\emph{condition} rejects a byte exactly when a counterexample input exists that places it inside a token. The shipped
predicate is narrower. It also withholds vacuously certified bytes, because no input the token set accepts contains them
and a caller could not find one to split at. The condition admits those bytes, and no counterexample exists for them.
That matters for reading Section~\ref{sec:applicability}, where an empty exact-certificate entry is a fact about the
tokenization. A modulo-column negative is conservative on its own, and on every row of the table the decision of
Section~\ref{sec:exact} settles it.

We are not aware of prior work stating this condition, in particular the re-entrancy requirement, as a static per-symbol
property of the token DFA, though its components are all classical. We checked that claim against the parallel lexing
and parallel finite-automata literature, the theory of synchronizing automata and synchronizing codes, incremental
lexical analysis, verified invertible lexing~\cite{chassot2026ziplex}, and the recent work on
sequential~\cite{reps1998maximal,li2025uniform} and streaming~\cite{li2026streaming} tokenization. Each answers a
neighboring question: where to restart after an edit, which words reset an automaton and how long such a word must be,
whether a code admits a word after which decoding may resume independently, whether printing a token sequence, plainly
or with inserted separators, re-lexes to the same tokens exactly or modulo the separator class, how to avoid repeated
rescanning, how to recover a chunk's entry state after the fact~\cite{voetter2021gpu}. The question here is which single
bytes a given maximal-munch token set renders safe in advance, together with the token-length guarantee that makes the
answer usable. We did not find it posed in those terms.

\paragraph{Enumeration and structural prescans.} Mytkowicz et al.~\cite{mytkowicz2014dpfsm} enumerate every possible
start state and select the correct computation afterwards, so convergence reduces redundant work and there is no guess
to repair; the relaxed certificate instead removes the need to consider more than one state. Sin'ya et
al.~\cite{sinya2013sfa} study the size of the simultaneous finite automaton, giving worst-case bounds and empirical
data, which is the dual question to ours, since they bound the cost of carrying many states and we give a condition
under which none need be carried. Plex~\cite{li2021plex} arrives in the same territory from the other side. It reports
splitting JSON at newline as prior work, on the stated ground that newline is a delimiter ``not allowed in any lexeme'',
and then abandons delimiters altogether. A prescanning automaton derived from the scanner builds a transfer function per
chunk, and combining the functions determines the states each chunk's thread begins from. The thread tries those states
in turn, falling through to the next whenever one would force a backtrack into the preceding chunk. Plex's contribution
belongs with the composition family above. This paper adds the certification. For an arbitrary token set it gives an
exact characterization that decides when a delimiter byte is sound. For an arbitrary discarded set it gives a local
certificate that is sound and conservative, beside an exact decision that settles what the certificate refuses.

Structural prescans for JSON, such as Mison~\cite{li2017mison} and simdjson~\cite{langdale2019simdjson}, do not exploit
a fact like the one Table~\ref{tab:applicability} records. They search for structural characters, the object and array
brackets and \code{:} among them, all of which may legally occur inside a JSON string, and then suppress the occurrences
that do by deriving a string mask from the quote and backslash positions. Mison's Algorithm 1 constructs that mask, and
simdjson performs the same masking branchlessly. Theirs is a format-specific scan that computes which occurrences are
real, where the certificate instead decides, for an arbitrary compiled token set, which bytes need no such computation.
Reps~\cite{reps1998maximal} studies maximal-munch scanning, and the lookahead hazard noted in the discussion of state
composition below is a consequence of the same backtracking behavior.

\paragraph{Static entry-state properties.} Ko, Jung, Han and Burgstaller~\cite{ko2014speculative} parallelize DFA
membership by speculation and reduce the speculation by means of a static property of the automaton: for each symbol,
the set of non-error states that some transition on that symbol enters. A chunk is then matched from the states in the
set for the last symbol of the preceding chunk, its reverse lookahead symbol. Where that set is a singleton, its state
is the sole non-error candidate, and the actual entry state is either that state or the error sink of their complete
automaton, at which the match has already failed. The certificate is a condition on the other end of the transitions,
asking which live states have a transition on the byte at all and admitting the byte when $q_0$ is the only one and is
not re-entrant. On one trim automaton, completed by a single error sink for their construction, the useful certificate
is the stronger of the two, since by determinism a useful certified byte, one only $q_0$ consumes, has one live target.
Their set is a singleton whenever the byte is usefully certified, and a vacuously certified byte, one no token contains,
has the empty set. The converse, that a singleton set implies a useful certificate, fails, since over the single token
\code{ab} the byte \code{b} has one target and begins no token. The two differ in what each must reproduce. A membership
test needs the state after the symbol, so the sole non-error candidate is all it needs, while a restart from $q_0$ that
knows nothing of what precedes needs a token boundary before the byte. The boundary requirement is why the certificate
speaks of sources, requires that $q_0$ is not re-entrant, and is taken over the trim live automaton $A^+$, where an
untrimmed presentation's dead and unreachable states, and any equivalent copies of live states it carries, can only
enlarge their count.

\paragraph{Deciding the exact question.} Yang~\cite{yang1996mealy} observes that an object composed for maximal-munch
lexing must carry token output and lookahead as well as a map on states. The complete-scan automaton of
Section~\ref{sec:exact} is that object built directly for one scan. Its obligation sets carry exactly the lookahead a
longest-match decision still depends on, and the $z^\ell K_t$ coding carries the output without storing a length. The
failure automaton beside it recognizes the words the scanner rejects by a direct construction, without complementing the
first. The equality test by delays is elementary, and we claim no priority for it or for the automata; the proofs are
included so that the result is checkable without further reading. Specific to this paper is the question they decide: a
cut's safety at every occurrence of a byte, together with the chunk-completeness obligation that
Definition~\ref{def:safe} carries and an output comparison alone would miss. The library computes the exact certificate
at construction, by a linear-time analysis of the transition table, and does not run the decision there. The decision's
cost is singly exponential in $n$ in the worst case, as Theorem~\ref{thm:exact} bounds it, and release 3.1.0 of
munch~\cite{munch310} runs it only when a caller asks, as \code{modulo\_counterexample()}, for one symbol and one
discarded set at a call and within a cap on the states each search holds. The result reports whether the searches
settled the question or the cap stopped one. The implementation behind the results reported here ships as a figure
program outside the library.

\paragraph{Chosen cut points.} Two lines outside compiler construction, content-defined chunking and document
fingerprinting, choose positions in a byte stream by rules of their own, with objectives other than lexical safety.
\citet{bjorner2010chunking} call a chunking method local when the cutpoints of every file are exactly the positions
whose $h$-vicinity, the $2h + 1$ entries around the position, lies in a chosen criterion set. A nonempty useful
certified set is the case $h = 0$ of that shape, a rule the token set imposes rather than one a designer picks, and a
local rule need not respect token boundaries at all. Over files whose entries are drawn independently and uniformly,
they define the slack of two files that agree from a point onward as the distance from that point to their first common
cutpoint, and analyze its expectation normalized by the expected chunk length for all but the local maximum method,
whose slack they write they could not estimate accurately. Under the same model they bound, across four chunking
methods, the probability that a long interval carries no cutpoint, with the local maximum method's figures resting on
their assumption that no two positions tie. \citet{schleimer2003winnowing} define the density of a fingerprinting scheme
as the expected fraction of the hashes computed that it selects, under a given input distribution. Assuming independent,
uniformly distributed hashes with ties negligible they prove winnowing's density asymptotically $2/(w+1)$ for window
size $w$, and a lower bound of $1.5/(w+1)$ for every local algorithm, one that selects a position from each window of
$w$ hashes by that window's contents alone. The longest fingerprint-free run their Web experiment reports belongs to the
competing $0 \bmod p$ selector, whereas winnowing selects in every window by construction.

Neither result transfers here, since a certified set is under no obligation to select in every window, so the coverage
the lower bound assumes need not hold. The density results also rest on an independence model, whereas this paper asks
only which bytes a given token set admits. There the rule is chosen and judged by what it selects, winnowing by its
density under coverage and the chunking methods by slack and cutpoint-free tails at matched expected chunk lengths. Here
the token set forces the rule, Section~\ref{sec:applicability} inventories what it admits, and how often those bytes
occur is a property of the input.

\paragraph{Relation to state composition.} We computed the entry contexts a line start can present for the block-comment
row of Table~\ref{tab:applicability}, whose automaton has fourteen states, and the question needs no corpus. A line
start follows a newline, so it is either a token boundary, needing no state at all, or the scan sits in the target of a
newline transition out of a live reachable state, and the token continues only if that target has an outgoing
transition. Collecting those targets yields every mid-token context that can arise on any input whatsoever. For this
token set there is exactly one, the interior of a block comment, since every in-comment state moves to the same target
on a newline, and the initial state produces a complete newline token that cannot continue. A line start presents one of
two entry contexts, that comment-interior state and the boundary, against an automaton of fourteen states. The
conventional tokenization of the same three kinds, block comment included, gives three, since a whitespace run may also
be open across a line start, and neither count is close to the state count.

Most line starts are already token boundaries, though nothing short of the scan identifies which. On a corpus of 400
generated documents no failed lookahead after a token's final accepted position crossed a line end. The counter in
\code{composition.cpp} checks this, asserting that bytes read and discarded beyond the last acceptance never include a
newline. The absence of such crossings is a property of the corpus rather than of the token set. Schemes composing only
scanner states require it as an assumption, and a token set with longer lookahead could violate it on inputs the corpus
does not contain. The entry-context counts of both tokenizations are asserted by \code{composition.cpp} in the artifact,
which also checks that the contexts a corpus realizes lie inside the structural bound, and that they are the states this
section names.

The comparison should be drawn narrowly. The counts are a structural upper bound on the \emph{entry contexts} a line
start can present, two here and three for the conventional tokenization with block comments, and they bound only the
state component of the object a prefix scan for maximal-munch lexing composes over. Yang~\cite{yang1996mealy} gives the
reason: such a scan must represent token output and lookahead behavior as well as a map from states to states, and an
automaton unable to recognize a token's start from its first character must be transformed before the scan applies.
Nothing here implements a prefix scan or measures its time, memory, function representation or output handling, so this
is not evidence about the cost of composition as a whole. Composition is also available wherever the certificate reports
nothing, which is the larger class.

Certified splitting requires no full-input state-composition pass, and the claim needs care because chunks are not
serialized. Mytkowicz et al.~\cite{mytkowicz2014dpfsm} dispatch every chunk in parallel during a first phase to build
its transfer function, resolve the entry states from those, and dispatch the chunks again. Every byte is visited before
any chunk can be scanned \emph{from its resolved entry state} or emit final output. A certified byte needs none of that,
since a boundary is found by scanning forward from the ideal offset until a certified byte is hit, and each chunk is
scanned once. The distance that search covers depends on the input, and Section~\ref{sec:planning} measures it as
certified bytes grow scarce. The mandatory full-input pass does not amortize when the number of chunks is small, or when
only part of the input is to be scanned, and it is why the two approaches complement each other.

\section{Conclusion}

A byte $b$ is a certified split symbol of a compiled token set when no reachable state other than $q_0$ has a
$b$-transition whose target can reach acceptance, and, if $q_0$ has one, $q_0$ is not re-entrant. Every occurrence of
such a byte in a completely tokenizable input begins a token, so chunk boundaries placed there preserve the token stream
exactly, with no speculation, no overlapping chunk scans, and no merge beyond ordered concatenation. The re-entrancy
clause is needed for this, since without it the natural criterion, that no live state other than $q_0$ consumes the
byte, certifies $a$ for the token set $a^*$ and splits $aa$ into two tokens where the scanner yields one.

The condition is necessary as well as sufficient once the analysis is restricted to states an input can reach and from
which acceptance is still reachable. A token set with no useful certified byte therefore has no byte that both occurs in
a completely tokenizable input and universally begins a token, and the relaxed condition may still recover one modulo a
discarded set. Deriving the certificate is a linear-time sweep of the transition table at construction time, and the
same sweep decides which certified bytes can occur in valid input at all. That makes the applicability question
mechanical. The method is fragile, since a single token kind with a near-total interior alphabet, such as a block
comment, removes every useful certificate from an otherwise cooperative C-like grammar.

On the favorable scaling grammar and 512\,MiB dense corpus, useful exact certificates are frequent. Across the two
revisions in the restricted-CPU-set placement, exact-certificate splitting there reached 92.6--95.3\% eight-way
efficiency against the one-chunk wrapper and a $3.46$--$3.94\times$ four-chunk ratio against the plain scan, measured
end to end through the benchmark wrapper. Both are conventionally rounded ranges on one machine, and neither is a
confidence interval. The unrestricted placement reaches 88.5\% at eight chunks on the earlier revision. On the two
eight-chunk rows at 512\,MiB, the per-scenario best-to-worst spreads, relative to the medians behind them, reach about
8.8\% and 14.5\%. The wider end-to-end range comes from the serial baseline used as its denominator, since at 512\,MiB
the multi-chunk rows move far less between the revisions than that baseline does, as Section~\ref{sec:evaluation}
reports separately.

This paper therefore adds the automaton-derived certification step to established delimiter-based parallel lexing, a
step that prior systems have handled through language-specific reasoning. Weakening the guarantee from token stream
equality to equality after deleting discarded tokens recovers split points that the exact certificate must reject, with
the same constant-time one-bit query. At construction the relaxation adds a backward closure and a second
$O(|Q|\,|\Sigma|)$ sweep, plus $O((|Q|+|I|)\log(1+|I|))$ ordered-set work, and the coalesced test of release 2.2.0 of
munch adds a partition refinement. The recovered cases include ones that matter in practice: the conventional C-like
token set studied here and a JSON lexer both gain newline without any change to their token definitions, which removes a
token set redesign that the exact method had imposed on its users. Both local tests are conservative, and the exact
decision measures by how much. Of the 360 pairs of symbol and token set the fixed seeded sweep finds safe, the
state-equality condition admits 265 and the coalesced test 306. A named token set is safe yet certified by neither. The
exact question is decidable all the same. Safety modulo the discarded set is PSPACE-complete for general token automata
and decidable in polynomial time for literal vocabularies. Wherever it fails, an input and a cut replay the failure. On
the fifteen token sets of the applicability study the decision confirms every cell of the relaxed column, so there the
local tests give nothing away. The relaxed result remains a precomputed one-bit query; the shipped planner and every
throughput measurement use the exact certificate.

Two boundaries are worth restating, because both were initially unclear to us. First, certification depends only on the
token set, and for the relaxed condition also on the declared discarded set. It does not depend on a particular input,
so a read-time transformation cannot make an uncertified byte certified under the unchanged token set. A cooperating
format and producer can instead reserve a byte, exclude or escape it within tokens, and emit it only at boundaries,
which rescues even a token set containing block comments. Second, prefix-scan composition, the general alternative, need
carry only two line-start entry contexts in its state component for the measured split-friendly block-comment grammar.
That count bounds the state domain alone, since a maximal-munch prefix scan must also carry lookahead and token
output~\cite{yang1996mealy}. The state domain is not where the two approaches differ. The certificate's structural
distinction is the absence of a mandatory context-recovery pass over the whole input, and what composition costs in full
is not measured here. The planner's own forward searches can themselves degenerate to $O(kN)$ aggregate search work when
certified bytes are scarce, which Section~\ref{sec:planning} measures.

\section*{Tools}

Large language models assisted with drafting, with checking citations against primary records and with reviewing the
implementation. The author verified each suggestion by derivation, against primary sources, or against the
implementation, its tests and the archived artifacts, and is responsible for all content.

{\small \bibliographystyle{plainnat}
\bibliography{refs}
}

\end{document}